\documentclass[11pt]{amsart}

\usepackage{amssymb}
\usepackage{mathrsfs}
\usepackage{mathtools}
\usepackage{multirow}
\usepackage{comment}

\usepackage[top=1.3in, bottom=1.3in, left=1.3in, right=1.3in]{geometry}

\newcommand{\norm}[1]{\lVert#1\rVert}
\newcommand{\normlll}[1]{{\left\vert\kern-0.25ex\left\vert\kern-0.25ex\left\vert #1 
    \right\vert\kern-0.25ex\right\vert\kern-0.25ex\right\vert}}
\newcommand{\im}{\mathrm{i}}
\newcommand{\dee}{\mathrm{d}}
\newcommand{\ztwo}{\{\pm 1\}}
\newcommand{\hs}{\mathscr{H}}

\newtheorem{theorem}{Theorem}[section]
\newtheorem{corollary}[theorem]{Corollary}
\newtheorem{lemma}[theorem]{Lemma}
\newtheorem{proposition}[theorem]{Proposition}

\theoremstyle{remark}
\newtheorem{remark}[theorem]{Remark}

\theoremstyle{definition}
\newtheorem{definition}[theorem]{Definition}

\theoremstyle{definition}
\newtheorem{example}{Example}[section]

\begin{document}

\title{On the $K$-theoretic classification of topological phases of matter}
\author{Guo Chuan Thiang}
\address{Mathematical Institute, University of Oxford, UK}
\email{guo.thiang@balliol.ox.ac.uk}
\date{\today}

\begin{abstract}
We present a rigorous and fully consistent $K$-theoretic framework for studying gapped topological phases of free fermions such as topological insulators. It utilises and profits from powerful techniques in operator $K$-theory. From the point of view of symmetries, especially those of time reversal, charge conjugation, and magnetic translations, operator $K$-theory is more general and natural than the commutative topological theory. Our approach is model-independent, and only the symmetry data of the dynamics, which may include information about disorder, is required. This data is completely encoded in a suitable $C^*$-superalgebra. From a representation-theoretic point of view, symmetry-compatible gapped phases are classified by the super-representation group of this symmetry algebra. Contrary to existing literature, we do not use $K$-theory to classify phases in an absolute sense, but only relative to some arbitrary reference. $K$-theory groups are better thought of as groups of obstructions between homotopy classes of gapped phases. Besides rectifying various inconsistencies in the existing literature on $K$-theory classification schemes, our treatment has conceptual simplicity in its treatment of all symmetries equally. The Periodic Table of Kitaev is exhibited as a special case within our framework, and we prove that the phenomena of periodicity and dimension shifts are robust against disorder and magnetic fields.
\end{abstract}

\maketitle

\section{Introduction}\label{section:introduction}
In his short and influential paper \cite{kitaev2009periodic}, Kitaev proposed that gapped phases of non-interacting fermions can be classified using the techniques of topological $K$-theory. In his approach, there are $2+8=10$ classes of systems to consider in each spatial dimension $d$, based on the presence or absence of time reversal and/or $\mathrm{U}(1)$ symmetry. Their classification groups exhibit a certain periodicity with respect to $d$ and was attributed, somewhat mysteriously, to Bott periodicity. A Periodic Table was partially drawn up, with each symmetry class in each spatial dimension having one of the $K$-theory groups of a point as its classification group. Little detail was provided in the original paper, which led to a series of authors providing their own accounts \cite{ryu2010topological,stone2011symmetries,abramovici2012clifford,freed2013twisted}. Upon careful investigation, these contain inequivalent treatments of distinct families of free-fermion systems. Furthermore, subsequent work on crystalline \cite{chiu2013classification} and weak topological insulators revealed the existence of phases which are not directly accounted for by the Periodic Table. Despite the lack of a proper proof of (or even the necessary definitions or assumptions in) the Periodic Table, a consensus that it unambiguously provides a complete $K$-theoretic classification of free-fermion phases appears to have been reached; for instance, see the review papers \cite{hasan2010colloquium,qi2011topological}.

Unfortunately, there seems to be a number of inconsistencies in the mathematics, and more crucially in the physical interpretation of $K$-theory groups, in the existing literature on topological phases. With the exception of the excellent Freed--Moore paper \cite{freed2013twisted} (which does not address the matter of dimension shifts in the Periodic Table), there has been very little attempt to put $K$-theoretic classification ideas on a firm mathematical footing. Consequently, the full machinery of $K$-theory has yet to be substantially utilised. This is not due to an incompatibility between the mathematics of $K$-theory and the physics, as rigorous work on the Integer Quantum Hall Effect \cite{bellissard1994noncommutative}, bulk-edge correspondence \cite{kellendonk2004boundary,kellendonk2004quantization} and Fermi surfaces \cite{horava2005stability} demonstrate. In addition, the Chern numbers commonly used in the physics literature as topological invariants can more fundamentally be understood as $K$-theory invariants \cite{prodan2014non}. Thus $K$-theoretic objects have really been lurking in the background in condensed matter physics for a long time.

This paper seeks to address these issues by providing a complete and consistent framework for the use of $K$-theory in the study of gapped topological phases. Our treatment of quantum mechanical symmetries borrows heavily from the comprehensive analysis in \cite{freed2013twisted}. Subsequently, this paper diverges from existing work in two very important ways. First, we utilise operator $K$-theory\footnote{More precisely, a version due to Karoubi \cite{karoubi1978k,karoubi1968algebres,karoubi1970algebres,karoubi2008clifford}.} rather its commutative (topological) version, which makes available powerful theorems such as the Connes--Thom isomorphisms, the Packer--Raeburn decomposition and stabilisation theorems, and various exact sequences for the $K$-theory of crossed product algebras. The second difference is physical: our representation spaces for the symmetries and Hamiltonians are single-particle Hilbert spaces for charged free-fermions, as opposed to Dirac--Nambu spaces (see Item \ref{item:diracnambu} below). 

Next, the Clifford algebras are presented as twisted group algebras of time-reversal and/or charge-conjugation symmetries. Therefore Clifford algebras enter our $K$-theoretic framework in a fundamental way, rather than by a separate ad-hoc analysis of time-reversal and charge-conjugation. Furthermore, our important physical definitions are completely new and relates to Clifford algebras and $K$-theory in a mathematically precise way. For instance, Definition \ref{definition:symmetrycompatiblehamiltonianshomotopy} gives a precise notion of homotopic phases, and Definition \ref{definition:differencegroup} illustrates how a $K$-theoretic \emph{difference-group} $\mathbf{K}_0$ classifies obstructions in passing between phases. No unnatural Grothendieck completion is carried out, and inverses arise simply by taking differences in the opposite order. 

Thus the collection of topological phases should really be thought of as a torsor for the $K$-theoretic difference group, in the same way that an affine space has forgotten its origin. Once a fixed Hamiltonian has been chosen as a standard reference, other phases are measured with respect to it. Indeed, the very idea of classifying phases up to homotopy in an absolute sense is problematic (see Example \ref{example:isomorphicnothomotopic}). The philosophy of using $K$-theory to classify differences between phases appears, in any case, to be the original intention of Kitaev in \cite{kitaev2009periodic}. Furthermore the concept of a \emph{relative index} has already been studied in the context of the Quantum Hall Effect in \cite{avron1990quantum,avron1994charge}. Having set up the crucial definitions, the machinery of $K$-theory takes over and allows us to derive the Periodic Table of Kitaev (when appropriately interpreted) as a simple corollary. Our main result is Theorem \ref{theorem:generaldimensionshift}, which demonstrates that the phenomenon of ``dimension shifts'' and periodicity in the $K$-theoretic classification remains even in the presence of disorder and magnetic fields.

Conceptually, our starting point is Wigner's theorem, which says that the topological symmetry group $G$ (which is not assumed to be compact) for the dynamics of a quantum mechanical system must be represented projectively on a complex Hilbert space $\hs$ as unitary or antiunitary operators. A continuous homomorphism $\phi:G\rightarrow\ztwo$, distinguishes the unitarily-implemented subgroup $G^u\coloneqq\mathrm{ker}(\phi)$ from the antiunitarily implemented subset $G^a=G - G^u$. For any two $x,y\in G$, their representatives $\theta_x, \theta_y$ satisfy $\theta_x\theta_y=\sigma(x,y)\theta_{xy}$, with $\sigma:G\times G\rightarrow\mathrm{U}(1)$ a generalised 2-cocycle,
\begin{equation}
	\sigma(x,y)\sigma(xy,z)=\sigma(y,z)^x\sigma(x,yz),\label{generalised2cocycle}
\end{equation}
where for $\lambda\in\mathrm{U}(1)$, $\lambda^x\coloneqq\lambda$ if $\phi(x)=+1$ and $\lambda^x\coloneqq\overline{\lambda}$ if $\phi(x)=-1$. Thus, interesting topology resides not only in the group of symmetries, but also in the cohomological data of projective unitary-antiunitary representations (PUA-reps). The PUA-representation theory of $(G,\phi,\sigma)$ is a special case of a twisted covariant representation of a twisted $C^*$-dynamical system as explained in Section \ref{section:twistedcovariantrepresentations}. We go a step further and consider charge-conjugating symmetries on the same fundamental level as other symmetries, leading to $\mathbb{Z}_2$-graded versions of twisted covariant representations. This step is already suggested by the central role of charge-conjugation in relativistic quantum theories, and is ultimately vindicated in our context by the unifying role of super-algebra in $K$-theory.

Our main novel physical insight is then the following: \emph{the topology which appears in free-fermion topological phases such as topological insulators has its origin in symmetry}. Associated to the algebra of symmetries is a derived ``space'' (e.g.\ a Brillouin torus) which is noncommutative in general, and is of secondary importance. We remark that we are not considering topology arising from, e.g., putting physical systems on topologically interesting physical spaces (commutative or otherwise) --- the latter is more closely related to \emph{topological order} \cite{chen2010local}. It may be interesting to investigate whether these two sources of interesting topology can be combined.

\subsection{Outline}
In Section \ref{section:inconsistencies}, we point out some important inconsistencies in the existing literature. In Section \ref{section:positiveenergyquantization}, we set our conventions for describing symmetry-compatible free-fermion dynamics. They motivate the definitions of (graded) covariant representations, twisted $C^*$-dynamical systems, and twisted crossed products in Sections \ref{section:twistedcovariantrepresentations}--\ref{section:twistedcrossedproducts}. These latter two sections review material which may not be familiar to researchers in topological phases, and can be skimmed over by experts. The intimate relation between Clifford algebras, twisted group algebras, and the tenfold way is explained in Section \ref{section:cliffordtenfoldway}. We move on to the $K$-theoretic classification of symmetry-compatible gapped Hamiltonians proper in Sections \ref{section:superrepresentationgroups} and \ref{section:karoubidifferenceconstruction}, whose computation is illustrated by examples in Section \ref{section:decompositioncrossedproducts}. We treat the special case of topological band insulators in Section \ref{section:bandinsulators}. A Periodic Table in the general sense of Kitaev is derived in Section \ref{section:dimensionshift}, and we prove that periodicity and dimension shifts persist under very general conditions in Theorem \ref{theorem:generaldimensionshift}.

\section{Remarks on the existing literature and some inconsistencies}\label{section:inconsistencies}
We begin with a list of some inconsistencies in both the mathematics and physical interpretation in the existing literature on $K$-theoretic free-fermion classification schemes. They will be further elaborated upon and rectified in the main body of the paper.

\begin{enumerate}
\item A common definition of a symmetry-compatible topological phase begins with the space $Y$ of Hamiltonians (gapped or otherwise) which are compatible with a certain given representation of some symmetry data. This space is sometimes called the ``classifying space'', which should be distinguished from the mathematical notion with the same name. Two Hamiltonians which are path-connected within this space are identified. The phases, up to homotopy, are then given by the path-components $\pi_0(Y)$. This is merely a \emph{set}, with no distinguished identity element, composition law, or inverse. On the other hand, a $K$-theory invariant has the crucial and useful additional structure of an abelian \emph{group}. 

\begin{example} In a ``no symmetry'' situation in zero spatial dimensions, one has a bare Hilbert space $\mathbb{C}^N$, and a spectrally-flattened compatible gapped Hamiltonian is just a grading operator on $\mathbb{C}^N$. The space of such Hamiltonians is the union of the Grassmanians of $k$-planes in $\mathbb{C}^N$, with $k=0,1,\ldots,N$, each of which is connected. Therefore the phases, up to homotopy, form an $(N+1)$-element \emph{set}. A ``large-$N$'' limit is often taken so that the set of phases becomes a countably infinite set. This set is then conferred the status of the free abelian group $\mathbb{Z}\cong K^0(\star)$ in an unclear manner.
\end{example}

\item The (mathematical) classifying space $C_0$ for complex $K^0$ is $\mathbb{Z}\times \mathrm{U}/(\mathrm{U}\times \mathrm{U})$, where $\mathrm{U}/(\mathrm{U}\times \mathrm{U})$ is the infinite complex Grassmanian. This means that $\tilde{K}^0(X)$ is isomorphic to $[X,C_0]$, the homotopy classes of (based) maps from $X$ to $C_0$. The $\mathbb{Z}$ factor in the classifying space does not arise directly from finite Grassmanians, but is related to Bott periodicity. Writing $S^0$ for the $0$-sphere,
\begin{equation*}
	\mathbb{Z}\cong\pi_0(C_0)=[S^0,C_0]\cong\tilde{K}^0(S^0)\cong K^0(\star),
\end{equation*}
as befits an extraordinary cohomology theory. In general, one needs to be careful when defining topological phases in zero dimensions.

A related issue arises for topological phases in higher spatial dimensions. A popular, if somewhat mathematically misguided, point of view is to regard $\tilde{K}^{-n}(X)$ as the ``homotopy group'' of (based) maps $[X,C_n]$, where $C_n$ is a classifying space for complex $K$-theory; similarly for $R_n$ in the real case. Then if $X$ is the $m$-sphere, $[S^m,C_n]$ is identified with $\pi_m(C_n)$. However, it is only possible to form a group from the \emph{set} of homotopy classes of (based) maps $[X,Y]$, when $Y$ is a $H$-group (homotopy-associative) or when $X$ is a $H$-cogroup. The two possible ways of defining a composition on $[X,Y]$ from the $H$-group or $H$-cogroup structures are a priori different. Nevertheless, if both choices are available, the resulting groups are isomorphic (pp.\ 44 of \cite{spanier1966algebraic}). Not all $X$ are $H$-cogroups, but all the $C_n$ are $H$-groups, so $\tilde{K}^{-n}(X)=[X,C_n]$ is well-defined as a group for all $X$, with respect to the $H$-group structure on $C_n$. In particular, \emph{the group structure on $K^0(X)$ is the direct sum of virtual bundles over $X$, whereas the $\pi_m(C_0)=[X,C_0]$ interpretation of $\tilde{K}^0(S^m)$ only makes sense for $X=S^m$}; yet the case of $X=\mathbb{T}^d$ is most certainly of interest in the study of topological insulators. Well-established topological invariants, such as the Chern numbers associated to filled Landau bands in the Integer Quantum Hall Effect, have the important property of being \emph{additive} with respect to the direct sum of bands (vector bundles), whereas the group structure of the homotopy groups $\pi_m(C_0)$ is a priori a distinct one.

\item When there is a discrete lattice of translational symmetries, one is led to the study of vector bundles over the Brillouin torus $X=\mathbb{T}^d$. For consistency with the zero-dimensional case, one should look for the \emph{unreduced} $K$-theory groups of $\mathbb{T}^d$. For convenience, it is often assumed that the ``interesting'' Brillouin zone $X'$ is a $d$-sphere, upon which the \emph{reduced} $K$-theory of $X'=S^d$ may be identified with a homotopy group of one of the classifying spaces $C_n$ or $R_n$. The reduced theory does not have the same physical interpretation as the unreduced theory. In fact, $\tilde{K}^{-n}(S^d)\cong K^{-n}(\mathbb{R}^d)$, and the latter has an interpretation in terms of vector bundles over $\mathbb{R}^{d+n}$ trivialised outside a compact set, which are in turn related to systems with \emph{continuous} translational symmetry. Thus there is a conflation of the reduced and unreduced theories. The discrepancy in the $K$-theory groups in the passage from $\mathbb{T}^d$ to $S^d$ is then attributed to the physical difference between ``strong'' and ``weak'' topological insulators. 

\item Reduced $K$-theory is sometimes motivated by restricting attention to stable isomorphism classes of vector bundles over $X$. Physically, stabilization entails identifying systems which differ by some ``topologically trivial" subsystem. For band insulators (without symmetries beyond the discrete translations $\mathbb{Z}^d, d\geq 1$), this means that adding trivial bands should not affect the topological classification. However, in the zero-dimension case, the reduced $K$-theory of a point is trivial, because the complex vector spaces in question are all trivially stably equivalent! This contradicts the $\mathbb{Z}$-classification which appears in many tables \cite{ryu2010topological,kitaev2009periodic}. Furthermore, the higher $K$-theory functors (reduced or otherwise) do not have a direct analogous formulation in terms of vector bundles over $X$. A consistent notion of ``triviality'' is thus lacking.

\item For topological insulators, i.e.\ graded vector bundles, the precise equivalence relation defining a topological phase/class is seldom consistently chosen or clearly defined. Vector bundles, possibly with extra structure dictated by symmetries, can be organised into isomorphism classes, graded or otherwise. Intuitively, a notion of ``homotopy classes of bundles'' is desired, but this cannot be in the trivial sense of deforming the total bundle space since a vector bundle can always be retracted onto its base space. Isomorphism classes of ungraded vector bundles correspond to homotopy classes of \emph{maps} from the base space to an appropriate classifying space, not those of the bundle itself. On the other hand, a gapped Hamiltonian determines a grading on a vector bundle, and it is homotopy within the space of allowed gradings (i.e.\ deformations of possible Hamiltonians) which actually captures the physical intuition of ``homotopic gapped phases''.

\item Within a single fixed realisation of relevant symmetries on a given representation space, it makes sense to consider homotopies between symmetry-compatible Hamiltonians. For two \emph{different} representation spaces (corresponding to two different physical systems), a choice of isomorphism (if available) is required before the question of whether a Hamiltonian on the first space can be deformed into a Hamiltonian on the second space can be asked. However, two candidate symmetry-compatible Hamiltonians on the same space can be isomorphic without being homotopic (see Example \ref{example:isomorphicnothomotopic}). Consequently, it is not straightforward to define a notion of homotopy between two Hamiltonians defined on different spaces.

\item A detailed treatment of the Integer Quantum Hall Effect (IQHE), which does not make the assumption of rational flux and includes the effects of disorder, utilises tools from non-commutative geometry and operator $K$-theory \cite{bellissard1994noncommutative}. Despite this, the IQHE is included in Kitaev's Periodic Table, which actually assumes the presence of a meaningful commutative Brillouin zone.

The Periodic Table of \cite{kitaev2009periodic}, when applied to topological insulators, should therefore be interpreted more carefully. Furthermore, it has already been recognised that the presence of point symmetries leads to different classification groups from those in his table. This indicates that vital topological information is present in the symmetry group itself.

\item \label{item:altlandzirnbauer}The Altland--Zirnbauer (AZ) classification of disordered femionic systems \cite{altland1997nonstandard,heinzner2005symmetry} is based on the compact classical symmetric spaces which provide spaces of symmetry-compatible time evolutions. While large-$N$ versions of symmetric spaces also feature in the classifying spaces of $K$-theory, the AZ classification (and indeed its Wigner--Dyson predecessor) makes no explicit reference to time evolutions generated by \emph{gapped} Hamiltonians, whereas $K$-theory is supposed to classify gapped phases.

An attempt to reconcile this disconnect was made in \cite{stone2011symmetries}. The general approach there and elsewhere in the literature (with \cite{freed2013twisted} being an exception) is to keep only the data of the ground states (or valence bands of topological insulators) for the purposes of classification. As long as charge conjugating (i.e.\ Hamiltonian reversing) symmetries are not present, this makes good sense and can even be motivated physically. In such cases, the valence and conduction bands \emph{separately} determine $K$-theory invariants of the insulating system, with the former usually more interesting. The availability of an interpretation of $K$-theory groups referring only to the valence band, distinguishes the A, AI and AII classes from the other classes in the tenfold way (see Section \ref{subsection:finiteversusinfinite}). These are the three classes whose topological invariants (the Chern numbers and Kane--Mele $\mathbb{Z}_2$ invariants) have been studied most closely in the condensed matter literature.

In general, a charge-conjugation symmetry may constrain the topology of \emph{both} the conduction and valence bands. Two gapped Hamiltonians which are non-homotopic may posses homotopic valence bands; indeed, remembering only the data of the valence band means that one loses information about the presence or absence of a charge-conjugation symmetry.

\item \label{item:diracnambu} The precise treatment of charge-conjugation symmetry differs among authors, and the related notion of particle-hole symmetry in a Dirac--Nambu formalism is often thrown into the mix as well \cite{freed2013twisted,heinzner2005symmetry,abramovici2012clifford,ryu2010topological}. The Dirac--Nambu space is a vector space of \emph{second-quantized} creation operators and annihilation operators, and is a useful auxiliary space often used for studying Bogoliubov de Gennes Hamiltonians. However, it also comes with extra structure such as a canonical conjugation which exchanges the creation operators with the annihilation operators. Furthermore, it is necessarily even-dimensional. These structures are not referred to in $K$-theory. For instance, not all elements of $K^0(X)$ or $KO^0(X)$ are represented by even-rank vector bundles over $X$. Despite this, $K$-theory is sometimes claimed to classify gapped systems in the Dirac--Nambu formalism.

Whether symmetries commute with the Hamiltonian (e.g.\ \cite{heinzner2005symmetry}), or are allowed to anticommute with the Hamitonian (e.g.\ \cite{ryu2010topological,freed2013twisted}), depends on whether one is in a first-quantized or second-quantized setting. This is related to whether charge-conjugation is implemented antiunitarily or unitarily. We explain this ambiguity in Section \ref{section:positiveenergyquantization}.

\item The use of $KR$-theory to account for the effect of time-reversal and charge-conjugation symmetries is actually redundant in our operator $K$-theory approach. In fact, band insulators with these discrete symmetries are not the actual Real bundles required for a $KR$-theory analysis, and the proper construction turns out to be quite involved (see Corollary 10.25 of \cite{freed2013twisted}).

\begin{example}\label{example:isomorphicnothomotopic}[Homotopic versus isomorphic band insulators]
Consider a band insulator in one spatial dimension which has one valence band and one conduction band. Suppose that there is also a sublattice symmetry $S$ which is unitary and squares to the identity (also called a \emph{chiral} symmetry). Physically, this is a Class $AIII$ insulator; mathematically, it is a rank-two $\mathbb{Z}_2$-graded complex hermitian vector bundle $E$ over $S^1$, with an \emph{odd} fibrewise action of $S$. Let us ignore the grading into valence and conduction bands for now. Let $\sigma_1,\sigma_2,\sigma_3$ be the standard Pauli matrices, and suppose that $S$ acts via the matrix $\sigma_3$ on each fibre with respect to global coordinates $(\theta,v)\in S^1\times\mathbb{C}^2\cong E$. According to the usual prescription in the literature, we should look for gradings $\Gamma$ of $E$, such that the grading operator $\Gamma_\theta$ on each fibre $E_\theta$ anticommutes with $\sigma_3$. Since $\Gamma_\theta$ is traceless and hermitian, it must be of the form $\Gamma_\theta=\cos(f(\theta))\sigma_1+\sin(f(\theta))\sigma_2$, for some $f(\theta)\in [0,2\pi]=S^1$. In other words, the set of gradings of $E$ which are $S$-compatible corresponds to the set of continuous functions $f:S^1\rightarrow S^1$, and a homotopy between two such functions is precisely a homotopy between the two bundle gradings ($AIII$ band insulator structures) that they determine. Therefore, the set of homotopic phases is $\pi_1(S^1)\cong \mathbb{Z}$.

Define $\Gamma^0$ and $\Gamma^1$ by $\Gamma^0_\theta=\sigma_1$ and $\Gamma^1_\theta=\cos(2\theta)\sigma_1+\sin(2\theta)\sigma_2$, then $\Gamma^0$ and $\Gamma^1$ are not homotopic. However, $(E,\Gamma^0)$ and $(E,\Gamma^1)$ are \emph{isomorphic} (in the graded sense), via the unitary bundle map $\Phi:(\theta,v)\mapsto(\theta,e^{-\im\theta \sigma_3}v)$, which is just a change of coordinates. One checks that $\Phi\circ \Gamma^0=\Gamma^1\circ \Phi$ and $\Phi\circ S=S\circ\Phi$, so $\Phi$ is indeed an even bundle map respecting the action of $S$. When $E$ and the $S$-action are fixed as in this example, writing $\Gamma_\theta$ in the manner that we did picks out a canonical choice of reference, namely $\Gamma^0$, which corresponds to the function $f\equiv 0\in[0]\in\pi_1(S^1)$. Any other compatible $\Gamma$ determined by some other $f'$ is measured against $\Gamma^0$ through $[f']\in \pi_1(S^1)$.

Suppose $(E',\Gamma')$ encodes the band structure of another two-band insulating system with sublattice symmetry, and is isomorphic to $(E,\Gamma^0)$ (hence to $(E,\Gamma^1)$ as well). \emph{Whether $\Gamma'$ should be considered to be homotopic to $\Gamma^0$ or $\Gamma^1$ (or neither) is dependent on the choice of identification between $E$ and $E'$}. Thus suggests that in a homotopy classification, it is more meaningful to classify \emph{relative} phases.

Incidentally, the valence band is always a trivial line bundle, so we do lose something if we forget about the conduction band and the $S$-action.

\end{example}

\begin{example}\label{example:assumetrivialAIII}
Example \ref{example:isomorphicnothomotopic} is related to the general construction of band insulators with $S$ symmetry (Class $AIII$) in \cite{ryu2010topological}. A trivial rank-$2N$ bundle $E$ over $X$ is fixed, and $S$ acts fibrewise as diag($1_N,-1_N$) with respect to some global coordinates. Then the possible compatible gradings $\Gamma$ are continuous choices over $X$ of $\Gamma_x=\begin{pmatrix}0 & q_x \\ q_x^\dagger& 0\end{pmatrix}$ with $q_x\in \mathrm{U}(N)$. Up to homotopy, these are given by $[X,\mathrm{U}(N)]$. For $X=S^d$ the phases (for a fixed \emph{trivial} $E$ and $S$) are thereby given by a homotopy group of the unitary group, which stabilises if $N$ is large enough relative to $d$. This is related to the fact that homotopy classes of continuous maps $X\rightarrow \mathrm{U}$ provide a model for $K^{-1}(X)\cong K_1(C(X))$. However, a map $X\rightarrow \mathrm{U}(N)$ only determines a band insulator structure when given some given trivial rank-$2N$ vector bundle $E$ with $S$-action. For completeness, non-trivial bundles $E$ with $S$-symmetry should also be treated.
\end{example}

\subsection{Some remarks on the Freed--Moore approach}
In the Freed--Moore approach, higher $K$-theoretic groups are constructed using bundles of graded Clifford modules, quotiented by a certain algebraic relation, in analogy to the Atiyah--Bott--Shapiro construction of the $K$-theory ring of a point. This is closely related to our construction of super-representation groups in Section \ref{section:superrepresentationgroups}. However, we make the important observation that there are two inequivalent ways of taking parity reversals in the construction, with each choice leading to \emph{opposite} orderings of the classification groups. Furthermore, the super-representation groups do not coincide with standard $K$-theory groups except for certain spaces (e.g.\ a point). An example is $K^{-1}(S^1)\cong\mathbb{Z}$, whereas all bundles of graded $\mathbb{C}l_1$-modules over $S^1$ are ``trivial'' in the sense of Definition \ref{definition:superrepresentationgroup} and Definition 8.5 in \cite{freed2013twisted}.

Usual assumptions on the group of symmetries are: (i) distinguished time-reversing and charge-conjugating elements which are involutary in the full symmetry group, and (ii) there is a direct product factorisation into translational symmetries, point group symmetries, and time-reversal or charge-conjugation symmetries. As emphasized by Freed--Moore, these assumptions do not hold in many realistic systems. Because of this, they are led to twisted equivariant $K$-theory, although only some special twistings occur. In our approach, twistings appear in the form of twisted group algebras, and only the ordinary $K$-theory of these algebras enters. Furthermore, abstract results of Packer--Raeburn \cite{packer1989twisted} allow these twistings to be untwisted without compromising the $K$-theory (see Section \ref{section:decompositioncrossedproducts}).
\end{enumerate}

\section{Symmetries, spectral-flattening, and positive energy quantization}\label{section:positiveenergyquantization}
Following the general arguments of \cite{freed2013twisted}, elements of the symmetry group $G$ for the dynamics of a quantum mechanical system are presumed to be endowed with Hamiltonian and/or time preserving/reversing properties, which are encoded by a pair of continuous homomorphisms $c,\tau:G\rightarrow\ztwo$. An element $g\in G$ preserves (resp.\ reverses) the arrow of time if $\tau(g)=+1$ (resp.\ $\tau(g)=-1$); similarly, it commutes (resp.\ anticommutes) with the Hamiltonian if $c(g)=+1$ (resp.\ $c(g)=-1$). A third homomorphism $\phi:G\rightarrow \ztwo$ specifies whether $g$ is implemented unitarily ($\phi(g)=+1$) or antiunitarily ($\phi(g)=-1$). Writing $u_t$ for the unitary dynamics generated by the Hamiltonian $H$, and $\mathsf{g}$ for the unitary/antiunitary representative of $g$, the time-reversal equation 
\begin{equation*}
	\mathsf{g}u_t\mathsf{g}^{-1}=u_{\tau(g)t}
\end{equation*}
leads to $\phi\circ \tau\circ c\equiv 1$, so any two of $\phi,\tau,c$ specifies the third. Often, \mbox{$c\equiv 1$} is assumed (i.e.\ all symmetries commute with the Hamiltonian), then $\phi=\tau$ and antiunitarity becomes synonymous with time-reversal. However, in our description of free-fermion dynamics, we want to consider symmetries that effect charge-conjugation (see Section \ref{section:unitaryquantization}), so we allow for $c(g)=-1$. Then any two of $\phi,\tau,c$ may be independently specified. We also allow the symmetries to be projectively realised, i.e.,\ there may be a non-trivial cocycle $\sigma$.

The possibility of charge-reversing symmetries (present or otherwise) for free-fermion dynamics requires, logically, a notion of \emph{charged} dynamics and \emph{charged} representations of the canonical anticommutation relations (CARs), as opposed to their \emph{neutral} counterparts. The latter more correctly describes neutral (Majorana) fermions. Non-degeneracy of the dynamics (or a gapped Hamiltonian) allows us to distinguish between particle and antiparticle sectors, and we would like both species to have positive energy in second quantization. For instance, the Fermi level of a band insulator (which may be set to $0$) lies in a gap of the Hamiltonian, providing the particle-hole distinction. We recall the algebraic formalism of positive energy charged field quantization, and refer to \cite{derezinski2010positive,derezinski2013mathematics,gracia2001elements} for the neutral case and technical details. Then, we establish our conventions for dynamical symmetries, including time and charge reversal.

\subsection{CAR representations}
Let $\mathcal{Y}$ be a complex Hilbert space with inner product $h$. A \emph{charged CAR representation} over $(\mathcal{Y},h)$ in a complex Hilbert space $\hs$ is a complex-linear map $\mathrm{a}^*$ from $\mathcal{Y}$ to the bounded operators on $\hs$, such that
\begin{eqnarray}
	\{\mathrm{a}(y_1),\mathrm{a}(y_2)\}&=&0\,\,\,=\{\mathrm{a}^*(y_1),\mathrm{a}^*(y_2)\},\nonumber\\
	\{\mathrm{a}(y_1),\mathrm{a}^*(y_2)\}&=&h(y_1,y_2)1_\hs,\quad\qquad\qquad y_1,y_2\in\mathcal{Y}.\nonumber
\end{eqnarray}
Here, $\mathrm{a}(y)$ is the adjoint of $\mathrm{a}^*(y)$, and thus $\mathrm{a}(\cdot)$ is an anti-linear map from $\mathcal{Y}$ to the bounded operators on $\hs$. A charged CAR representation gives rise to a \emph{neutral} CAR representation over $(\mathcal{Y},b=\mathrm{Re\,} h)$ with a charge-$1$ $\mathrm{U}(1)$-symmetry, and conversely.

\subsection{Quantization of non-degenerate unitary dynamics}\label{section:unitaryquantization}
Let $u_t$ be a strongly continuous $1$-parameter unitary group on $(\mathcal{Y},h)$, with self-adjoint generator $H$ (the Hamiltonian). We assume that $u_t$ is non-degenerate, meaning that $\mathrm{ker}(H)=\{0\}$. We may define
\begin{equation}
	Q=\mathrm{sgn}(H),\quad J=\im\,\mathrm{sgn}(H)=\im Q,\quad |H|=\sqrt{H^2}>0,\nonumber
\end{equation}
and rewrite $u_t$ as $e^{tJ|H|}$. Note that $J$ is unitary, skew-adjoint and commutes with $H, Q$ and $|H|$. Furthermore, $\mathcal{Y}$ is graded by the charge operator $Q$ (``spectrally-flattened'' Hamiltonian) into $\mathcal{Y}^+\oplus \mathcal{Y}^-$, where $\mathcal{Y}^{\pm}$ is the $\pm 1$ eigenspace of $Q$. Writing $\mathcal{Z}$ for the space $\mathcal{Y}$ equipped with the modified complex unit $J$ instead of $\im$, we have $\mathcal{Z}=\mathcal{Y}^+\oplus\overline{\mathcal{Y}^-}$, where $\overline{\mathcal{Y}^-}$ is given the inner product dual to $h|_{\mathcal{Y}^-}$. The subspaces $\mathcal{Y}^{\pm}$ are invariant for $Q, H$ and $|H|$, so we may regard these operators as self-adjoint operators on $\mathcal{Z}$, in which case a subscript is appended, e.g., $H_\mathcal{Z}$.

On the Fock space $\bigwedge^*\mathcal{Z}$, the \emph{charged fields} are
\begin{equation}
	\mathrm{a}^*(y)=a^*(y^+)+a(\overline{y^-}),\quad	\mathrm{a}(y)=a(y^+)+a^*(\overline{y^-}),\qquad y=(y^+,y^-)\in\mathcal{Y}^+\oplus \mathcal{Y}^-,\nonumber
\end{equation}
where $a^*$ and $a$ are the standard creation and annihilation operators on $\bigwedge^*\mathcal{Z}$. The maps $\mathrm{a}^*$ and $\mathrm{a}$ furnish a charged CAR representation over $(\mathcal{Y},h)$, called the \emph{positive energy Fock quantization} for the non-degenerate unitary dynamics $u_t$. There are second quantized versions of the Hamiltonian and charge operators,
\begin{equation}
	\mathcal{H}=\mathrm{d}\Lambda(|H|_\mathcal{Z})\geq 0,\qquad \mathcal{Q}=\mathrm{d}\Lambda(Q_\mathcal{Z}),\nonumber
\end{equation}
which implement the dynamics and charge symmetry on Fock space,
\begin{equation}
	e^{\im t\mathcal{H}}\mathrm{a}(y)e^{-\im t\mathcal{H}}=\mathrm{a}(u_t y),\quad e^{\im \theta\mathcal{Q}}\mathrm{a}(y)e^{-\im \theta\mathcal{Q}}=\mathrm{a}(e^{\im\theta} y).\nonumber
\end{equation}

\subsubsection{Charge and/or time reversal in non-degenerate unitary dynamics}\label{unitarychargetimereversal}
A symmetry operator $\mathsf{g}$ on $(\mathcal{Y},h)$ is required to be unitary or antiunitary according to $\phi(g)$, and time preserving or reversing according to $\tau(g)=c\circ\phi(g)$, i.e.,\ $\mathsf{g}u_t=u_{\tau(g)t}\mathsf{g}$. A short computation leads to the following commutation relations
\begin{equation}
	\mathsf{g}H\mathsf{g}^{-1}=c(g)H,\quad \mathsf{g}|H|\mathsf{g}^{-1}=|H|,\quad\mathsf{g}Q\mathsf{g}^{-1}=c(g)Q,\quad	\mathsf{g}J\mathsf{g}^{-1}=\tau(g)J,\nonumber
\end{equation}
from which we find that $\mathsf{g}_\mathcal{Z}$ (i.e.\ the map $\mathsf{g}$ considered as an operator on $\mathcal{Z}$) is unitary or antiunitary according to $\tau(g)$. We may then amplify $\mathsf{g}_\mathcal{Z}$ to an (anti)unitary operator $\hat{\mathsf{g}}=\Lambda(\mathsf{g}_\mathcal{Z})$ on Fock space.

\begin{remark}
The modified imaginary unit $J$ is determined by the dynamics $u_t=e^{\im t H}$ only through the spectrally-flattened Hamiltonian $\mathrm{sgn}(H)$.
\end{remark}

We stress that presence of a time/charge revers\emph{ing} symmetry does not imply that of a distinguished charge/time revers\emph{al} operator. Indeed, Freed--Moore \cite{freed2013twisted} have pointed out that there are physically relevant examples that do not fit into the tenfold way \cite{heinzner2005symmetry,ryu2010topological}, which requires distinguished involutary charge/time reversal operators $\mathsf{C},\mathsf{T}$. We prefer to work more generally, and think of time/charge reversal as \emph{properties} of a symmetry $g\in G$. Under certain splitting assumptions on $G$, we can recover the usual $\mathsf{T}$ and/or $\mathsf{C}$ operators, see Section \ref{section:cliffordtenfoldway}.

\subsection{Remarks on conventions for free-fermion dynamics}\label{subsection:remarkschargedversusneutral} In many treatments of the tenfold way \cite{abramovici2012clifford,altland1997nonstandard,freed2013twisted,heinzner2005symmetry,kitaev2009periodic,ryu2010topological,schnyder2008classification}, the single-particle ``Hamiltonian'' in certain symmetry classes is taken to act on a \emph{Nambu space} $W=V\oplus\overline{V}$ rather than a single-particle Hilbert space $V$. A Nambu space has a canonical real structure $\Sigma:(v_1,\overline{v_2})\mapsto(v_2,\overline{v_1})$. The fixed points of $\Sigma$ form the real mode space $\mathcal{M}$ of Majorana operators, and $\mathcal{M}$ inherits a real inner product from $W$ by restriction. The operator $J=\im\oplus-\im$ on $W=V\oplus\overline{V}=\mathcal{M}\otimes\mathbb{C}$ restricts to an orthogonal complex structure on $\mathcal{M}$, and $(\mathcal{M},J|_\mathcal{M})\cong V$. One begins with second-quantized dynamics on Fock space $\bigwedge^* V$, generated by a Hamiltonian $H_\mathrm{F}$ which is required to be quadratic in the creation and annihilation operators. Such dynamics can be reformulated on Nambu space $V\oplus \overline{V}$, with generating ``Hamiltonian'' $H_\mathrm{N}$ subject to certain symmetry constraints. Alternatively, the dynamics can be specified by a skew-symmetric operator $A$ on $\mathcal{M}$, whose complexification is $\im H_\mathrm{N}$. The gapped condition is sometimes imposed on $H_N$. An example is the Bogoliubov--de Gennes (BdG) Hamiltonian for the quasi-particle dynamics of a superconducting system. It is important to note that the polarization $J$, and thus the Fock space in second quantization, are already implicit in the Nambu space formulation, whereas they are determined by $H$ in positive energy Fock quantization. Also, particle number is not necessarily conserved (because $a(v)a(v')$ and $a^*(v)a^*(v')$ terms are allowed in the second-quantized Hamiltonian $H_\mathrm{F}$), so $A$ may not have a $\mathrm{U}(1)$-symmetry (i.e.\ it may not commute with $J$). The definition of symmetries of a Hamiltonian, especially those of charge-conjugation and time-reversal, also differ between authors.

In our approach, the Hamiltonian $H$ generating the non-degenerate unitary dynamics on $(\mathcal{Y},h)$ determines the particle-antiparticle distinction in second quantization. In practice, we impose a stronger \emph{gapped} condition on $H$; namely, we require $0\not\in\mathrm{spec}(H)$. In this case, we call $H$ a \emph{gapped} Hamiltonian. We allow for antiunitary symmetries, as well as charge-reversing symmetries which reverse $Q=\mathrm{sgn}(H)$. Two symmetry-compatible gapped Hamiltonians are identified if they have the same spectral flattening, i.e., if they result in the same grading operator $Q$ on $\mathcal{Y}$. Then, the specification of a homotopy class\footnote{The Hamiltonian $H$ is unbounded in general, and care must be taken in order to interpret spectral-flattening as a homotopy in a precise sense, see Appendix D of \cite{freed2013twisted}.\label{fn:unboundedhomotopy}} of charged free-fermion dynamics respecting the symmetry data $(G,c,\phi,\sigma)$ is precisely that of a \emph{graded} projective unitary-antiunitary representation for $(G,c,\phi,\sigma)$, as defined in Section \ref{subsection:gradedpua}). This provides the physical motivation for the constructions in Section \ref{section:twistedcovariantrepresentations}.

We remark that a graded PUA-rep for $(G,c,\phi,\sigma)$ may also be interpreted as an ordinary quantum mechanical system. We usually combine such quantum mechanical systems using the tensor product. On the other hand, at the one-particle level, we combine free-fermion systems using the direct sum operation, which gets translated into the tensor product at the Fock space level. We are only interested in describing free-fermion dynamics and its symmetries at the one-particle level, so the direct sum applies. This allows us to construct commutative monoids of free-fermion systems, paving the way for the use of $K$-theoretic methods in their classification.

\section{The general notion of twisted covariant representations}\label{section:twistedcovariantrepresentations}
We give an outline of the basic definitions and constructions of twisted covariant representations of twisted $C^*$-dynamical systems \cite{busby1970representations,leptin1965verallgemeinerte,packer1989twisted}. We make a simple generalisation to $\mathbb{Z}_2$-graded twisted covariant representations, and show that they arise naturally as (graded) PUA-reps in the context of quantum systems with time/charge-reversing symmetries. All gradings will be $\mathbb{Z}_2$-gradings unless otherwise stated.

\subsection{Ungraded covariant representations}\label{subsection:ungradedcovariantrepresentation}
Let $\mathcal{A}$ be a separable, possibly non-unital, real or complex $C^*$-algebra\footnote{A reference for basic facts about real $C^*$-algebras is Chapter 1 of \cite{schroder1993k}.}. We denote its multiplier algebra by $\mathcal{M}\mathcal{A}$, and its group of unitary elements  $\{u\in\mathcal{M}\mathcal{A}:u^*u=uu^*=1_{\mathcal{M}\mathcal{A}}\}$ by $\mathcal{U}\mathcal{M}\mathcal{A}$. If $\mathbb{F}$ is the ground field of $\mathcal{A}$, we write  $\mathrm{Aut}_{\mathbb{F}}(\mathcal{A})$ for the group of $\mathbb{F}$-linear $*$-automorphisms of $\mathcal{A}$. Let $G$ be a locally compact, second countable, amenable\footnote{Amenability holds in all the physical examples that we consider in this paper, and is made in order to avoid having to distinguish between reduced and full crossed products later on.} group, with left Haar measure $\mu$ and identity element $e$. As in  Section 2 of \cite{packer1989twisted}, we give $\mathcal{U}\mathcal{M}\mathcal{A}$ the strict topology, and $\mathrm{Aut}_{\mathbb{F}}(\mathcal{A})$ the point-norm topology.
\begin{definition}[Twisted $C^*$-dynamical system \cite{busby1970representations,leptin1965verallgemeinerte}]
A pair $(\alpha,\sigma)$ of Borel maps \mbox{$\alpha: G\rightarrow \mathrm{Aut}_{\mathbb{F}}(\mathcal{A})$} and $\sigma: G\times G\rightarrow \mathcal{U}\mathcal{M}\mathcal{A}$ satisfying
\begin{subequations}	
\begin{align}
	\alpha(x)\alpha(y)&=\mathrm{Ad}(\sigma(x,y))\circ\alpha(xy), \label{multidentity}\\
	\sigma(x,y)\sigma(xy,z)&=\alpha(x)(\sigma(y,z))\sigma(x,yz)\label{multidentity2},\\
	\sigma(x,e)&=1=\sigma(e,x),\label{cocycleassumption}\\
	\alpha(e)&=\mathrm{id}_\mathcal{A},\qquad\qquad\qquad\qquad x,y,z\in G,\label{automorphismassumption}
\end{align}
\end{subequations}
is called a \emph{twisting pair} for $(G,\mathcal{A})$. The map $\sigma$ is called a \emph{2-cocycle} with values in $\mathcal{U}\mathcal{M}\mathcal{A}$, or simply a \emph{cocycle}, and the quadruple $(G,\mathcal{A},\alpha,\sigma)$ is called a \emph{twisted $C^*$-dynamical system}.
\end{definition}
For notational ease, we will often write $\alpha_x\equiv \alpha(x)$ and $a^x\equiv \alpha_x(a)\equiv \alpha(x)(a)$.

\begin{definition}[Twisted covariant representation]
A \emph{twisted covariant representation} of a twisted $C^*$-dynamical system $(G,\mathcal{A},\alpha,\sigma)$ is a non-degenerate $*$-representation of $\mathcal{A}$ as bounded operators on a separable Hilbert space $\hs$ over $\mathbb{F}$, along with a compatible Borel map $\theta:x\mapsto\theta_x$ from $G$ to the unitary\footnote{When $\mathbb{F}=\mathbb{R}$, we also use ``orthogonal'' for emphasis.} operators on $\hs$, in the sense that 
\begin{subequations}
\begin{align}
	\theta_x\theta_y &=\sigma(x,y)\theta_{xy}, \label{gencovrep1}\\
	\theta_x(am) &=a^x(\theta_x m),\qquad\qquad x,y,\in G, a\in\mathcal{A},\, m\in \hs. \label{gencovrep2}
\end{align}
\end{subequations}
\end{definition}
Note that \eqref{gencovrep2} can be restated as $a^x=\mathrm{Ad}(\theta_x)(a)$, and then we see that \eqref{gencovrep1} is consistent with \eqref{multidentity}. In the untwisted case, i.e.\ $\sigma\equiv 1$, the Borel map $\alpha$ is a homomorphism, hence continuous (Theorem D.11 of \cite{williams2007crossed}). Then $(G,\mathcal{A},\alpha,1)$ is a (untwisted) $C^*$-dynamical system $(G,\mathcal{A},\alpha)$ in the usual sense (e.g.\ 7.4.1 of \cite{pedersen1979c}, 2.1 of \cite{williams2007crossed}, or 10.1 of \cite{blackadar1998k}). Similarly, $\theta$ becomes a strongly-continuous homomorphism from $G$ to the unitary group of $\hs$. Thus, $\theta$ is a (untwisted) covariant representation of $(G,\mathcal{A},\alpha)$ in the usual sense (e.g.\ 7.4.8 of \cite{pedersen1979c} or 10.1 of \cite{blackadar1998k}), and no harm is done by dropping the adjective ``twisted'' when $\sigma\equiv 1$. We say that two twisted covariant representations $(\theta,\hs),(\theta',\hs')$ of $(G,\mathcal{A},\alpha,\sigma)$ are $\emph{equivalent}$ if there is a unitary $\mathcal{A}$-linear intertwiner $U:\hs\rightarrow \hs'$ such that $U\theta_x U^{-1}=\theta'_x$ for all $x\in G$.

There is an action of the group of Borel functions $\lambda:G\rightarrow\mathcal{U}\mathcal{M}\mathcal{A}$ on twisting pairs (Section 3 of \cite{packer1989twisted}), defined by
\begin{subequations}
\begin{align}
	\alpha'(x)		&=	\mathrm{Ad}(\lambda(x))\circ\alpha(x),\label{possibletwistingautomorphism}\\
	\sigma'(x,y)	&=	\lambda(x)\alpha_x(\lambda(y))\sigma(x,y)\lambda(xy)^{-1}.\label{possibletwistingcocycle}
\end{align}
\end{subequations}
Two twisting pairs $(\alpha,\sigma)$ and $(\alpha',\sigma')$ are \emph{exterior equivalent} if they are related by such a transformation, and there is a $1\textrm{--}1$ correspondence between the covariant representations of $(G,\mathcal{A},\alpha,\sigma)$ and those of $(G,\mathcal{A},\alpha',\sigma')$, via the adjustments $\theta_x\mapsto\lambda(x)\theta_x$. This generalises the familiar notion of equivalence of cocycles for projective unitary group representations (i.e.\ $\mathcal{A}=\mathbb{C}$). If the cocycle $\sigma$ is assumed to be central in $\mathcal{A}$, there is no effect of $\lambda$ on $\alpha$ in \eqref{possibletwistingautomorphism}. The conjugation in \eqref{multidentity} and the condition \eqref{automorphismassumption} are then redundant, and we also have $\alpha(x^{-1})=\alpha(x)^{-1}$. A central cocycle is said to be \emph{trivial} if there is a Borel function $\lambda:G\rightarrow\mathcal{Z}(\mathcal{U}\mathcal{M}\mathcal{A})$ such that $\sigma(x,y)=\lambda(x)\lambda(y)^x\lambda(xy)^{-1}$, i.e.,\ $\sigma$ is a coboundary in the sense of cohomology. We say that two central cocycles $\sigma_1,\sigma_2$ are \emph{equivalent}, or in the same cocycle class, if $\sigma_1\sigma_2^{-1}$ is a trivial cocycle. In many special cases of physical interest, the representative $\sigma$ in a cocycle class can be chosen to make certain computations more convenient, e.g. Proposition \ref{proposition:parityoperatorstandardform} and Lemma \ref{lemma:pm1cocycle}. Note that if $\sigma$ is not necessarily central, $\alpha$ and $\sigma$ must be considered \emph{concurrently} when making an adjustment $\theta_x\mapsto\lambda(x)\theta_x$.

\subsection{Graded covariant representations}\label{subsection:gradedcovariantrepresentation}
Let $\mathcal{A}$ be a \emph{graded} real or complex $C^*$-algebra, i.e. $\mathcal{A}$ has a direct sum decomposition into two self-adjoint closed subspaces $\mathcal{A}=\mathcal{A}_0\oplus\mathcal{A}_1$, satisfying $\mathcal{A}_i\mathcal{A}_j\subset\mathcal{A}_{i+j\,\mathrm{(mod\,2)}}$. Let $\mathrm{Aut}_{\mathbb{F}}(\mathcal{A})$ now denote its group of \emph{even} $\mathbb{F}$-linear  $*$-automorphisms, i.e., $*$-automorphisms that preserve the decomposition $\mathcal{A}=\mathcal{A}_0\oplus\mathcal{A}_1$. We assume that the cocycles $\sigma$ take values in the \emph{even} elements $\mathcal{U}\mathcal{M}\mathcal{A}_0$ of $\mathcal{U}\mathcal{M}\mathcal{A}$. These restrictions are consistent with equations \eqref{multidentity} and \eqref{multidentity2} for a twisting pair $(\alpha,\sigma)$. Suppose that the group $G$ is also equipped with a continuous homomorphism $c:G\rightarrow\ztwo$. The quintuple $(G,c,\mathcal{A},\alpha,\sigma)$ is called a \emph{graded twisted $C^*$-dynamical system}.
\begin{definition}[Graded twisted covariant representation] A \emph{graded} covariant representation of a graded twisted $C^*$-dynamical system $(G,c,\mathcal{A},\alpha,\sigma)$ is a graded $*$-representation of $\mathcal{A}$ on a graded Hilbert space $\hs=\hs_0\oplus \hs_1$ over $\mathbb{F}$ (i.e.\ $\mathcal{A}_i\hs_j\subset \hs_{i+j\,\mathrm{(mod\,2)}}$), along with a compatible Borel map $\theta:G\rightarrow\mathrm{U}(\hs)$, in the sense of \eqref{gencovrep1}--\eqref{gencovrep2}, with the additional condition that $\theta_x$ is an even (resp.\ odd) operator if $c(x)=+1$ (resp.\ $c(x)=-1$).
\end{definition}
Two graded covariant representations $(\theta,\hs),(\theta',\hs')$ for $(G,c,\mathcal{A},\alpha,\sigma)$ are \emph{graded equivalent}, or simply \emph{equivalent}, if there is an \emph{even} unitary $\mathcal{A}$-linear map $U:\hs\rightarrow \hs'$ intertwining $\theta$ with $\theta'$. For instance, a graded Hilbert space over $\mathbb{F}$ and its parity-reverse are equivalent as ungraded representations of $\mathbb{F}$, but are inequivalent in the graded sense unless the two homogeneous subspaces have the same dimension. There is also the notion of \emph{super-equivalence} of graded covariant representations, which we will consider in Section \ref{superrepresentationdefinitionsection}. In many of our applications, $\mathcal{A}$ is trivially graded, i.e., purely even, and the only complication comes from the data of $c:G\rightarrow\ztwo$.

\subsection{Special cases I: Projective unitary-antiunitary representations}\label{subsection:ungradedpua}
A complex Hilbert space $(\hs,h)$ is equivalently a real Hilbert space $(\hs,b)$ with real inner product $b=\mathrm{Re}(h)$, along with a $b$-orthogonal complex structure $J$ (i.e.\ $J^2=-1$) playing the role of multiplication by $\im$. The complex inner product $h$ may be recovered from $b$ and $J$ by setting $h(u,v)=b(u,v)+\im b(Ju,v)$. Note that $h$ induces the same norm on $\mathcal{Y}$ as $b$ does. An orthogonal operator on $(\mathcal{Y},b)$ is (anti)-unitary as an operator on $(\mathcal{Y},h)$, iff it (anti)-commutes with $J$.

Let $\phi:G\rightarrow\ztwo$ be a continuous homomorphism, and $\sigma$ be a $\mathrm{U}(1)$-valued 2-cocycle as in \eqref{generalised2cocycle}. A projective unitary-antiunitary representation\footnote{See \cite{parthasarathy1969projective,williams2007crossed} for some topological matters.} (PUA-rep) $\theta$ of $(G,\phi,\sigma)$ on a complex Hilbert space $(\hs,h)$ is a Borel map $x\mapsto\theta_x$ such that $\theta_x$ is a unitary (resp.\ antiunitary) operator on $(\hs,h)$ if $\phi(x)=+1$ (resp.\ $\phi(x)=-1$), and \mbox{$\theta_x\theta_y=\sigma(x,y)\theta_{xy}$}. By regarding $(\hs,b)$ as a real Hilbert space, and $\im$ as a complex structure $J$ as above, we can equivalently define a PUA-rep of $(G,\phi,\sigma)$ as a map $\theta$ from $G$ to the orthogonal operators on $(\hs,b)$, subject to
\begin{eqnarray}
	\theta_x\theta_y&=&\sigma(x,y)\theta_{xy},\quad x,y\in G,\nonumber\\
	\theta_x J&=& \phi(x)J\theta_x.\nonumber
\end{eqnarray}

Suppose $\phi$ is surjective, and let $\mathcal{A}=\mathbb{C}$ as an ungraded real $C^*$-algebra. Thus $\mathcal{A}=\mathbb{R}\oplus\im\mathbb{R}$ as a real vector space, with basis $\{1,\im\}$, $\im^2=-1$, and the $*$-operation taking $\im$ to $-\im$. There are two elements of $\mathrm{Aut}_{\mathbb{R}}(\mathbb{C})$, namely complex conjugation $K$ and the identity $\mathrm{id}_{\mathbb{C}}$. A $*$-representation of $\mathcal{A}=\mathbb{C}$ is a real Hilbert space $(\hs,b)$ along with a linear operator $J$ representing $\im$, such that $J^2=-1$ and $J^*=-J$, i.e., $J$ is an orthogonal complex structure. Define the map $\alpha:G\rightarrow\mathrm{Aut}_{\mathbb{R}}(\mathbb{C})$ by
\begin{equation}
	\alpha_x\equiv\alpha(x)\coloneqq\begin{cases} \mathrm{id}_\mathbb{C} & \text{if } \phi(x)=+1, \\
													K & \text{if } \phi(x)=-1.\label{conjugationautomorphism}\end{cases}
\end{equation}
Equations \eqref{gencovrep1}-\eqref{gencovrep2} say that a covariant representation $\theta$ of $(G,\mathbb{C},\alpha,\sigma)$ on a real Hilbert space $(\hs,b)$, is precisely a PUA-rep of $(G,\phi,\sigma)$ on $\hs$.

\subsection{Special case II: Gapped Hamiltonians and graded projective unitary-antiunitary representations}\label{subsection:gradedpua}
Following the discussion in Section \ref{subsection:remarkschargedversusneutral}, we assume that a PUA-rep for $(G,\phi,\sigma)$ has, additionally, a gapped self-adjoint Hamiltonian $H$, and that $G$ has a second continuous homomorphism $c:G\rightarrow\ztwo$ such that
\begin{equation}
	\theta_xH=c(x)H\theta_x,\qquad\forall x\in G.\label{definitionofc}
\end{equation}
Thus $H$ is a gapped Hamiltonian compatible with the symmetries specified by the data $(G,c,\phi,\sigma)$. Note that $\mathrm{sgn}(\cdot)$ is a continuous function on the spectrum of $H$ homotopic to the identity function. We can deform $H$ to its spectral flattening $\Gamma=\mathrm{sgn}(H)$ while preserving the commutation/anticommutation relations with $G$ as expressed in \eqref{definitionofc} (but see Footnote \ref{fn:unboundedhomotopy}). A $(G,c,\phi,\sigma)$-compatible Hamiltonian is then identified with its spectrally-flattened version for the purposes of a homotopy classification.

A \emph{graded} PUA-rep $\theta$ of $(G,c,\phi,\sigma)$ on a graded complex Hilbert space $\hs=\hs_0\oplus\hs_1$ is a PUA-rep of $(G,\phi,\sigma)$, along with a self-adjoint grading operator $\Gamma$ satisfying $\Gamma^2=1_\hs$, such that $\theta_x$ is even or odd according to whether $c(x)=+1$ or $-1$. The grading operator $\Gamma$ is exactly a representative of the class of $(G,c,\phi,\sigma)$-compatible Hamiltonians on $\hs$ whose spectral flattening is $\Gamma$. Suppose $\phi$ is surjective. Let $\mathcal{A}$ be the purely even real $C^*$-algebra $\mathbb{C}$, and define the (even) automorphisms $\alpha_x$ as in \eqref{conjugationautomorphism}. Then a graded covariant representation of $(G,c,\mathbb{C},\alpha,\sigma)$ is precisely a graded PUA-rep of $(G,c,\phi,\sigma)$.

\subsection{Special case III: Disordered systems and covariant representations}
Disordered systems are often modelled on a disorder space $\Omega$, on which the group $G$ acts by homeomorphisms. More generally, the disorder space can be noncommutative, and so $G$ acts as automorphisms on an algebra $\mathcal{A}$. We can generalise PUA-reps to include disorder, by replacing $\mathbb{C}$ with the algebra $\mathcal{A}$ and working with twisted dynamical systems and their covariant representations. Such objects were considered in the analysis of the IQHE in \cite{bellissard1994noncommutative}, but without the additional data of $\phi$ or $c$.

\section{Graded twisted crossed products and covariant representations}\label{section:twistedcrossedproducts}
In the previous sections, we explained how the implementation of symmetry and compatible gapped Hamiltonians leads to graded twisted $C^*$-dynamical systems $(G,c,\mathcal{A},\alpha,\sigma)$ and their covariant representations. In this section, we explain how all the symmetry data can be completely and faithfully encoded in a graded twisted crossed product $C^*$-algebra $\mathcal{A}\rtimes_{(\alpha,\sigma)} G$, which we may simply call the \emph{symmetry algebra}. This device will be very convenient for the application of $K$-theory in the later sections.

\subsection{Twisted crossed products and covariant representations}\label{twistedsection}
Let $L^1(G,\mathcal{A},\alpha,\sigma)$ be the Banach $*$-algebra of integrable functions\footnote{The integral of a $\mathcal{A}$-valued function on $G$ is a Bochner integral, see Appendix B of \cite{williams2007crossed} and the preliminary section of \cite{packer1989twisted}.} $F:G\rightarrow \mathcal{A}$ with the $L^1$-norm $\norm{F}_1=\int_G \norm{F(x)}\,\dee x$, equipped with a $(\alpha,\sigma)$-twisted convolution product $\star$ and involution $*$,
\begin{subequations}
\begin{align}
	(F_1\star F_2)(y)&\coloneqq \int_G F_1(x)\left(F_2(x^{-1}y)\right)^x\sigma(x,x^{-1}y)\,\dee x,\label{twistedconvolutionproduct}\\
	F^*(x)&\coloneqq \sigma(x,x^{-1})^*\left(F(x^{-1})^*\right)^x\Delta(x^{-1}),\label{twistedinvolution}
\end{align}
\end{subequations}
where $\Delta$ is the modular function on $G$. There is a $1\textrm{--}1$ correspondence between covariant representations $\theta$ of $(G,\mathcal{A},\alpha,\sigma)$, and non-degenerate $*$-representations of $L^1(G,\mathcal{A},\alpha,\sigma)$, given by taking the ``integrated form'' $\tilde{\theta}$ of $\theta$,
see Theorem 3.3 of \cite{busby1970representations} and Remark 2.6 of \cite{packer1989twisted}. A pre-$C^*$-norm is defined on $L^1(G,\mathcal{A},\alpha,\sigma)$ by
\begin{equation}
	\norm{F}_\mathrm{max} =\mathrm{sup}\{\norm{\tilde{\theta}(F)}\,:\,\theta \mathrm{\,\,is\,a\,covariant\,representation\,of\,}(G,\mathcal{A},\alpha,\sigma)\}.\nonumber
\end{equation}
\begin{definition}[Twisted crossed product $C^*$-algebra \cite{busby1970representations}]
Let $(G,\mathcal{A},\alpha,\sigma)$ be a twisted dynamical system. The \emph{twisted crossed product $C^*$-algebra} associated with $(G,\mathcal{A},\alpha,\sigma)$, denoted by $\mathcal{A}\rtimes_{(\alpha,\sigma)} G$, is defined to be the completion of $L^1(G,\mathcal{A},\alpha,\sigma)$ in the norm $\norm{\cdot}_\mathrm{max}$.
\end{definition}
The group $G$ is embedded, not necessarily homomorphically, in the multiplier algebra of the twisted crossed product via the Borel map \mbox{$j_G:G\rightarrow\mathcal{U}\mathcal{M}(\mathcal{A}\rtimes_{(\alpha,\sigma)}G)$}, where the multipliers $j_G(g)$ are defined on functions \mbox{$F\in L^1(G,\mathcal{A},\alpha,\sigma)$} by 
\begin{equation}
	(j_G(g)(F))(x)\coloneqq(F(g^{-1}x))^g\sigma(g,g^{-1}x),\qquad x \in G.\nonumber
\end{equation}
Likewise, the algebra $\mathcal{A}$ is embedded in $\mathcal{M}(\mathcal{A}\rtimes_{(\alpha,\sigma)} G)$ through the homomorphism $j_\mathcal{A}:\mathcal{A}\rightarrow\mathcal{M}(\mathcal{A}\rtimes_{(\alpha,\sigma)} G)$, defined by
\begin{equation}
	(j_\mathcal{A}(a)(F))(x)\coloneqq a(F(x)),\qquad F\in L^1(G,\mathcal{A},\alpha,\sigma), x\in G.\nonumber
\end{equation}

When $\sigma\equiv 1$, we recover the untwisted crossed product $C^*$-algebra associated with the untwisted $C^*$-dynamical system $(G,\mathcal{A},\alpha)$. If $\alpha\equiv 1$, we call $\mathbb{R}\rtimes_{(1,\sigma)} G$ (resp.\ $\mathbb{C}\rtimes_{(1,\sigma)} G$) the real (resp.\ complex) \emph{twisted group $C^*$-algebra} of $(G,\sigma)$. If $\sigma\equiv 1$ as well, we use a shortened notation\footnote{It is also standard to write $\mathbb{R}\rtimes G$ for a semi-direct product of the \emph{group} $\mathbb{R}$ with $G$. Nevertheless, the correct meaning should be clear from the context.} $\mathbb{R}\rtimes G$ (resp.\ $\mathbb{C}\rtimes G$) for the real (resp.\ complex) \emph{group $C^*$-algebra} of $G$. More generally, when $\alpha\equiv 1$ and $\sigma\equiv 1$, we will write $\mathcal{A}\rtimes G\coloneqq \mathcal{A}\rtimes_{(1,1)} G$ to ease notation.

Although there are universal characterisations of twisted crossed products, (see Definition 2.4 of \cite{packer1989twisted}, as well as \cite{packer1990twisted}), we only need the following important result.
\begin{proposition}[\cite{busby1970representations,packer1989twisted,packer1990twisted}]\label{prop:twisteduniversalproperty}
There is a one-to-one correspondence between the covariant representations of $(G,\mathcal{A},\alpha,\sigma)$ and the non-degenerate $*$-representations of $\mathcal{A}\rtimes_{(\alpha,\sigma)} G$. 
\end{proposition}

\subsection{Graded twisted crossed products}
For a graded twisted $C^*$-dynamical system $(G,c,\mathcal{A},\alpha,\sigma)$, we assign a grading to $\mathcal{A}\rtimes_{(\alpha,\sigma)}G$ as follows. Let $G_0\coloneqq \mathrm{ker}(c)$ and $G_1\coloneqq G - G_0$. The even subalgebra of $\mathcal{A}\rtimes_{(\alpha,\sigma)} G$ is the completion in $\norm{\cdot}_\mathrm{max}$ of $L^1(G_0,\mathcal{A}_0,\alpha,\sigma)\oplus L^1(G_1,\mathcal{A}_1,\alpha,\sigma)$, while the odd subspace is the completion of $L^1(G_0,\mathcal{A}_1,\alpha,\sigma)\oplus L^1(G_1,\mathcal{A}_0,\alpha,\sigma)$. Note that \eqref{twistedconvolutionproduct} and \eqref{twistedinvolution} respect this grading due to the restriction to even automorphisms $\alpha_x$ and even cocycles $\sigma$. A graded $*$-representation of the graded twisted crossed product $\mathcal{A}\rtimes_{(\alpha,\sigma)} G$ then corresponds one-to-one with a graded covariant representation of $(G,c,\mathcal{A},\alpha,\sigma)$.

\section{$CT$-symmetries, Clifford algebras, and the tenfold way}\label{section:cliffordtenfoldway}
The graded PUA-representation theory of a direct product of parity groups $\ztwo$ can be simplified by fixing certain choices for the representatives of cocycle classes. Let $G=\ztwo^n, n\geq 0$, where $\ztwo^0$ means the trivial group. Suppose $0\leq m \leq n$ of the $\ztwo$ generators $U_i, 1\leq i\leq m$ are to be represented unitarily ($\phi(U_i)=+1$), while the other $n-m$ generators $A_k, 1\leq k\leq n-m$ are to be represented antiunitarily ($\phi(A_k)=-1$). We will write $\mathsf{U}_i\coloneqq\theta_{U_i}$ and $\mathsf{A}_k\coloneqq\theta_{A_k}$ for their representatives in a graded PUA-rep $\theta$ of $(\ztwo^n,c,\phi,\sigma)$. As always, $\mathsf{U}_i$ and $\mathsf{A}_k$ are odd/even operators according to $c$.
\begin{lemma}\label{lemma:atmosttwoantiunitaries}
We may assume, without loss of generality, that there are at most two antiunitaries $\mathsf{A}_1,\mathsf{A}_2$.
\end{lemma}
\begin{proof}
Let $A$ be the image of the homomorphism $(\phi,c):\ztwo^n\rightarrow \ztwo\times\ztwo$, and $B$ be its kernel. Every non-identity element of $A,B$ and $\ztwo^n$ has order 2. Regarding the groups as finite-dimensional vector spaces over the two-element field $\mathbb{F}_2$, we have $\ztwo^n\cong B\times A$. Any $\mathbb{F}_2$-basis for $B$ provides a set of even unitarily implemented generators $U_i$. Since $\mathrm{dim}(A)\leq 2$, there are at most two $A_k\in A$ with $\phi(A_k)=-1$ providing antiunitary operators $\mathsf{A}_k$.\qed
\end{proof}
A basis for $A$ can be chosen to be one of the following: (i) empty, (ii) \{odd $U_n$\}, (iii) \{even $A_1$\}, (iv) \{odd $A_1$\}, or (v) \{even $A_1$, odd $A_2$\}. We proceed to study the possible cocycles $\sigma$ for $G=\ztwo^n$, and fix representatives for their cocycle classes.

Since $\mathsf{U}_i^2=\sigma(U_i,U_i) \theta_e=\sigma(U_i,U_i)$, we can make the modification $\mathsf{U}_i\mapsto\pm\sigma(U_i,U_i)^{-1/2}\mathsf{U}_i$ to fix $\mathsf{U}_i^2=+1$. This does not work for $\mathsf{A}_k$, since $(\lambda \mathsf{A}_k)(\lambda \mathsf{A}_k)=\lambda\overline{\lambda}\mathsf{A}_k^2=\mathsf{A}_k^2$ for any $\lambda\in\mathrm{U}(1)$. Setting $x=y=z=A_k$ in \eqref{generalised2cocycle} leads to $\sigma(A_k,A_k)=\overline{\sigma(A_k,A_k)}$, so $\mathsf{A}_k^2=\sigma(A_k,A_k)=\pm 1$ are invariants of the cocycle class of $\sigma$. Next, we look at the commutation relations amongst $\mathsf{U}_i$ and $\mathsf{A}_k$. For two unitaries $\mathsf{U}_i,\mathsf{U}_j$, $i<j$, we write $\lambda_{ij}\coloneqq\sigma(U_i,U_j)/\sigma(U_j,U_i)$ so that $\mathsf{U}_i\mathsf{U}_j=\lambda_{ij} \mathsf{U}_j\mathsf{U}_i$. Then
\begin{equation}
	\mathsf{U}_j=\mathsf{U}_i^2\mathsf{U}_j=\lambda_{ij}^2 \mathsf{U}_j\mathsf{U}_i^2=\lambda_{ij}^2 \mathsf{U}_j,\nonumber
\end{equation}
so fixing $\mathsf{U}_i^2=+1$ leads us to $\lambda_{ij}=\pm 1$. For any two complex scalars $a_1, a_2$, we have $(a_1\mathsf{A}_1)(a_2\mathsf{A}_2)=a_1\bar{a}_2\sigma(A_1,A_2)\theta_{A_1A_2}$, as well as $(a_2\mathsf{A}_2)(a_1\mathsf{A}_1)=a_2\bar{a}_1\sigma(A_2,A_1)\theta_{A_1A_2}$. By choosing $a_1=\pm\sigma(A_2,A_1)^{1/2}, a_2=\pm\sigma(A_1,A_2)^{1/2}$, and making the adjustments $\mathsf{A}'_1=a_1\mathsf{A}_1, \mathsf{A}'_2= a_2\mathsf{A}_2$ and $\theta'_{A_1A_2}=a_1a_2\theta_{A_1A_2}$, we obtain $\mathsf{A}'_1\mathsf{A}'_2=\theta_{A_1A_2}'=\mathsf{A}'_2\mathsf{A}'_1$. Finally, we look at $\mathsf{U}_i\mathsf{A}_k=\nu_{ik}\mathsf{A}_k\mathsf{U}_i$, where $\nu_{ik}\coloneqq\sigma(U_i,A_k)/\sigma(A_k,U_i)$. The equation
\begin{equation}
	\mathsf{A}_k=\mathsf{U}_i^2\mathsf{A}_k=\nu_{ik}^2\mathsf{A}_k\mathsf{U}_i^2=\nu_{ik}^2\mathsf{A}_k\nonumber
\end{equation}
means that $\nu_{ik}=\pm 1$, which cannot be fixed by utilising the residual $\pm 1$ phase freedom in $\mathsf{U}_i,\mathsf{A}_k$. We summarise this discussion in the following Proposition:
\begin{proposition}\label{proposition:parityoperatorstandardform}
Let $\theta$ be a graded PUA-rep of $(\ztwo^n,c,\phi,\sigma)$, with $0\leq m\leq n$ unitarily implemented group generators $U_i$ and $0\leq n-m\leq 2$ antiunitarily implemented group generators $A_k$, as in Lemma \ref{lemma:atmosttwoantiunitaries}. We can adjust $\mathsf{U}_i$ and $\mathsf{A}_k$, while staying in the same cocycle class, so that $\mathsf{U}_i^2=+1$, $\mathsf{A}_k^2=\pm 1$, $\mathsf{A}_k\mathsf{A}_l=\mathsf{A}_l\mathsf{A}_k$, $\mathsf{U}_i\mathsf{U}_j=\pm  \mathsf{U}_j\mathsf{U}_i$, and $\mathsf{U}_i\mathsf{A}_k=\pm \mathsf{A}_k\mathsf{U}_i$, for all $1\leq i,j \leq m$ and $1\leq k,l\leq n-m$.
\end{proposition}

Let $\mathbf{i}=\{i_1,\ldots, i_p\}$ and $\mathbf{k}=\{k_1,\ldots, k_q\}$ be (possibly empty) increasing subsets of $\{1,\ldots,n\}$ and $\{1,\ldots,n-m\leq 2\}$ respectively. Let $U_\mathbf{i}$ and $A_\mathbf{k}$ denote the group elements $U_{i_1}\ldots U_{i_p}$ and $A_{k_1}\ldots A_{k_q}$ respectively, with $U_\emptyset=1=A_\emptyset$. In particular, $U_i=U_{\{i\}}$ and $A_k=A_{\{k\}}$. Every element of $\ztwo^n$ can be uniquely written as $U_\mathbf{i}A_\mathbf{k}$ for some $\mathbf{i}, \mathbf{k}$. We can use the remaining phase freedom for the group elements $U_\mathbf{i}A_\mathbf{k}$ with $|\mathbf{i}|+|\mathbf{k}|\geq 2$ to fix the condition $\theta_{U_\mathbf{i}A_\mathbf{k}}=\mathsf{U}_{i_1}\ldots\mathsf{U}_{i_p}\mathsf{A}_{k_1}\ldots\mathsf{A}_{k_q}$ for their representatives. The cocycle for this standardised $\theta$ is then completely determined by the set of $\pm 1$ in Proposition \ref{proposition:parityoperatorstandardform}.

\subsection{Clifford algebras associated with $CT$-subgroups --- the tenfold way}\label{subsection:CTgrouptenfoldway}
We define the $CT$-group  to be $\ztwo^2=\{1,T,C,S\}$, which has $(\phi,c)(T)=(-1,+1), (\phi,c)(C)=(-1,-1),(\phi,c)(S)=(+1,-1)$. The elements $T,C$ and $S=CT=TC$ refer to time-reversal, charge-conjugation, and sublattice symmetries respectively. We are interested in the graded PUA-reps of $(A,\sigma)$, where $A\subset \{1,T,C,S\}$ is a $CT$-subgroup and the homomorphisms $\phi,c$ on $A$ are implicit. The representatives of $T,\,C$ and $S$ (where present) are denoted by $\mathsf{T},\,\mathsf{C}$ and $\mathsf{S}$ respectively. 

First, we consider the full $CT$-group $A=\{1,T,C,S\}$. In the standard form of Proposition \ref{proposition:parityoperatorstandardform}, there are four choices $\mathsf{T}^2=\pm 1, \mathsf{C}^2=\pm 1$, and we may assume that $\mathsf{T}\mathsf{C}=\mathsf{S}=\mathsf{C}\mathsf{T}$. In each of these four cases, there is an associated real algebra generated by $\{\im, \mathsf{T},\mathsf{C}, \Gamma\}$, with commutation relations determined by $\phi$ and $c$, as well as $\im\Gamma=\Gamma\im$. Equivalently, we can choose the generating set $\{\Gamma, \mathsf{C}, \im \mathsf{C}, \im \mathsf{C}\mathsf{T}\}$, whose elements are mutually anticommuting, and are self-adjoint (resp.\ skew-adjoint) if they square to $+1$ (resp.\ $-1$). The latter choice shows explicitly that a graded PUA-rep for the full $CT$-group is precisely an ungraded $*$-representation of a certain real Clifford algebra $Cl_{r,s}$, where $r$ and $s$ are determined by the squares of the algebra generators in $\{\Gamma, \mathsf{C}, \im \mathsf{C}, \im \mathsf{C}\mathsf{T}\}$.

For $A=\{1,T\}$, there are two choices $\mathsf{T}^2=\pm 1$. The anticommuting set $\{\im,\mathsf{T},\im \mathsf{T}\Gamma\}$ generates $Cl_{1,2}$ when $\mathsf{T}^2=+1$ and $Cl_{3,0}$ when $\mathsf{T}^2=-1$. For $A=\{1,C\}$, there are again two choices $\mathsf{C}^2=\pm 1$. The anticommuting set $\{\Gamma, \mathsf{C}, \im \mathsf{C}\}$ generates $Cl_{0,3}$ if $\mathsf{C}^2=+1$ and $Cl_{2,1}$ if $\mathsf{C}^2=-1$. For the subgroup $A=\{1,S\}$, there is only one standard choice $\mathsf{S}^2=+1$, and $\{\Gamma, \mathsf{S}\}$ generates the \emph{complex} Clifford algebra $\mathbb{C}l_2$. Finally, the trivial subgroup $A=\{1\}$ leads easily to $\mathbb{C}l_1$, generated by $\Gamma$.

An alternative approach excludes $\Gamma$ as a Clifford generator. Except for the case $A=\{1,T\}$, the generating set of \emph{odd}, mutually anticommuting operators is taken from $\{\mathsf{C}, \im \mathsf{C}, \im \mathsf{C}\mathsf{T}\}$ or $\{\mathsf{S}\}$, and it generates a \emph{graded} Clifford algebra with one fewer positive Clifford generator compared to the first approach. Then, a graded PUA-rep for $(A,\sigma)$ is precisely a \emph{graded} $*$-representation of the corresponding \emph{graded} Clifford algebra, with grading operator $\Gamma$. However, the subgroup $A=\{1,T\}$ cannot be handled directly in this way, see Remark \ref{gradedversusungradedclifford}. 

We summarise this discussion in Table \ref{table:ctcliffordalgebras}. The Morita equivalence classes in the last column are obtained by the isomorphisms 
\begin{equation}
	Cl_{r,s}\otimes M_2(\mathbb{R})\cong Cl_{r,s}\otimes Cl_{1,1} \cong  Cl_{r+1,s+1}\nonumber,
\end{equation}	
the periodicities
\begin{eqnarray}
	\mathbb{C}l_{n}\otimes M_2(\mathbb{C})&\cong &\mathbb{C}l_{n+2}\nonumber\\ 
	Cl_{n,0}\otimes M_{16}(\mathbb{R})&\cong & Cl_{n+8,0}\nonumber\\ 
	Cl_{0,n}\otimes M_{16}(\mathbb{R})&\cong & Cl_{0,n+8}\nonumber,
\end{eqnarray}
as well as the $1\textrm{--}1$ correspondence between graded representations of $Cl_{r,s}$ and ungraded representations of $Cl_{r,s+1}$ \cite{atiyah1964clifford,lawson1989spin}.
\begin{table}
\begin{center}
\begin{tabular}{ l |c c | l  l l}
	\parbox[c]{1.7cm}{Generators of $A$} & $\mathsf{C}^2$ & $\mathsf{T}^2$ & \parbox[c]{2.0cm}{Associated algebra} & \parbox[c]{2.0cm}{Ungraded Clifford algebra} & \parbox[c]{2.0cm}{Graded Morita class}\\
	\hline
	$T$	    &	 & $+1$			 & $M_2(\mathbb{R})\oplus M_2(\mathbb{R})$	& $Cl_{1,2}$			& $Cl_{0,0}$ 			\\
	$C,T$ 	& $-1$ & $+1$			 & $M_4(\mathbb{R})$					  						& $Cl_{2,2}$			& $Cl_{1,0}$					\\
	$C$		& $-1$ &			 & $M_2(\mathbb{C})$												& $Cl_{2,1}$			& $Cl_{2,0}$ 			\\
	$C,T$ 	& $-1$ & $-1$			 & $M_2(\mathbb{H})$												& $Cl_{3,1}$			& $Cl_{3,0}$ 			\\
	$T$		&    & $-1$			 & $\mathbb{H}\oplus\mathbb{H}$							& $Cl_{3,0}$			& $Cl_{4,0}$ 			\\
	$C,T$ 	& $+1$  & $-1$		     & $M_2(\mathbb{H})$												& $Cl_{0,4}$			& $Cl_{5,0}$ 			\\
	$C$		& $+1$  &		  	 & $M_2(\mathbb{C})$ 												& $Cl_{0,3}$			& $Cl_{6,0}$ 			\\
	$C,T$ 	& $+1$  & $+1$			 & $M_4(\mathbb{R})$												& $Cl_{1,3}$			& $Cl_{7,0}$ 			\\

	\hline\hline
	N/A		& \multicolumn{2}{c|}{N/A}   & $\mathbb{C}\oplus\mathbb{C}$ 						&	$\mathbb{C}l_1$ & $\mathbb{C}l_0$ \\
	$S$	& \multicolumn{2}{c|}{$\mathsf{S}^2=+1$} & $M_2(\mathbb{C})$ 														& $\mathbb{C}l_2$ & $\mathbb{C}l_1$ \\
\end{tabular}
\caption{The ten classes of graded PUA-reps of $(A,\sigma)$, where $A$ is a subgroup of the $CT$-group, along with their corresponding Clifford algebras.}
\label{table:ctcliffordalgebras}
\end{center}
\end{table}

\begin{remark}\label{gradedversusungradedclifford}
In the two cases for $A=\{1,T\}$, the grading operator is not explicitly given as one of the Clifford generators in the associated ungraded Clifford algebra. In fact, $\{\im,\mathsf{T}\}$ generates the purely even algebra $M_2(\mathbb{R})$ or $\mathbb{H}$, which \emph{commutes} with $\Gamma$ in a graded PUA-rep of $(A,\sigma)$. We can rectify this by recalling the Morita equivalence between $Cl_{r,s}$ and $Cl_{r,s}\otimes Cl_{1,1}\cong Cl_{r+1,s+1}$. We introduce two extra Clifford generators $\mathsf{e},\mathsf{f}$ with $\mathsf{e}^2=-1, \mathsf{f}^2=+1$, so that $\{\im,\mathsf{T},\im \mathsf{T} \Gamma,\mathsf{e},\mathsf{f}\}$ forms a mutually anticommuting set of Clifford generators for either $Cl_{1,2}\otimes Cl_{1,1}=Cl_{2,3}$ (when $\mathsf{T}^2=+1$) or $Cl_{3,0}\otimes Cl_{1,1}=Cl_{4,1}$ (when $\mathsf{T}^2=-1$). An alternative generating set of mutually anticommuting operators is $\{\Gamma,\mathsf{e},\im \mathsf{e}, \mathsf{e}\mathsf{T}, \im \mathsf{f} \mathsf{T}\}$. We can now exclude $\Gamma$ and let $\{\mathsf{e},\im \mathsf{e}, \mathsf{e}\mathsf{T}, \im \mathsf{f} \mathsf{T}\}$ generate the graded Clifford algebra $Cl_{2,2}$ or $Cl_{4,0}$. This is an important subtlety: when we deal with $K$-theory later on, we will be interested in the various ways in which a fixed ungraded action of a graded Clifford algebra can be supplemented with a compatible grading operator $\Gamma$.
\end{remark}

\begin{remark}
An illuminating way to interpret the Clifford algebras constructed in this section, is as twisted group algebras for $\mathbb{Z}_2^n$, where the group generators are taken from a subset of $\{C,T,\im,\Gamma\}$ in the real case, and a subset of $\{S,\Gamma\}$ in the complex case. Writing $\mathbb{Z}_2^n$ in additive notation for now, and redefining generators as discussed above, the Clifford algebras can be written as $Cl_{r,s}\cong\mathbb{R}\rtimes_{(1,\sigma_{r,s})}\mathbb{Z}_2^{r+s}$ and $\mathbb{C}l_n\cong\mathbb{C}\rtimes_{(1,\sigma_n)}\mathbb{Z}_2^n$, where for $\mathbf{x},\mathbf{y}\in\mathbb{Z}_2^n$ or $\mathbf{x},\mathbf{y}\in\mathbb{Z}_2^{r+2}$ as appropriate,
\begin{equation}
	\sigma_{r,s}(\mathbf{x},\mathbf{y})=\big(-1\big)^{\sum \limits_{j<i}x_iy_j +\sum \limits_{i\leq r}x_iy_i},\qquad
	\sigma_n(\mathbf{x},\mathbf{y})=\big(-1\big)^{\sum \limits_{j<i}x_iy_j}.\nonumber
\end{equation}
\end{remark}

\section{Super-representation groups and topological triviality}
\label{section:superrepresentationgroups}
\subsection{Prelude: Representation group of a locally compact group}\label{representationringsection}
Let $G$ be a compact group. Its unitary representation theory can be summarised by the Peter--Weyl theorem. In $C^*$-algebraic language, this says that the group $C^*$-algebra $\mathbb{C}\rtimes G$ decomposes as a (possibly countably infinite) direct sum of matrix algebras over the unitary dual $\hat{G}$,
\begin{equation}
	\mathbb{C}\rtimes G \cong \bigoplus_{[V]\in \hat{G}}M_{\mathrm{dim}(V)}(\mathbb{C}).\label{peterweyl}
\end{equation}
For a proof, see Proposition 3.4 of \cite{williams2007crossed}. The complex representation ring $\mathcal{R}_\mathbb{C}(G)$ of $G$ is the Grothendieck group of the monoid of isomorphism classes of finite-dimensional unitary representations of $G$ under the direct sum. By complete reducibility, $\mathcal{R}_\mathbb{C}(G)$ is freely-generated by the elements of the unitary dual $\hat{G}$. In terms of $K$-theory, there is an isomorphism $\mathcal{R}_\mathbb{C}(G)\cong K_0(\mathbb{C}\rtimes G)\cong K_0^G(\mathbb{C})$ (see Section 11.1 of \cite{blackadar1998k}), where $K_0^G$ denotes the equivariant $K$-theory group. The real representation ring $\mathcal{R}_\mathbb{R}(G)$ requires a modification to the simple matrix algebras appearing in \eqref{peterweyl}, but we still have $\mathcal{R}_\mathbb{R}(G)\cong K_0(\mathbb{R}\rtimes G)\cong K_0^G(\mathbb{R})$. For projective unitary representations of $(G,\sigma)$, we can define the twisted representation group $\mathcal{R}_\mathbb{C}(G,\sigma)\coloneqq K_0(\mathbb{C}\rtimes_{(1,\sigma)} G)$.

If $G$ is locally compact but not compact, we choose the $K$-theoretic option and \emph{define} $\mathcal{R}_\mathbb{C}(G,\sigma)$ to be $K_0(\mathbb{C}\rtimes_{(\mathrm{1},\sigma)}G)$, generalising the compact case\footnote{This is a departure from the unitary representation theory of $G$. For instance, the complex group $C^*$-algebra of $\mathbb{Z}$ is isomorphic to $C(\mathbb{T})$, and its $K_0$-group comprises finitely-generated projective $C(\mathbb{T})$-modules, which are neither finite-dimensional nor unitary representations for $\mathbb{Z}$. However, the Serre--Swan theorem identifies such a $C(\mathbb{T})$-module with the sections of some finite-rank vector bundle over $\mathbb{T}$. This links well with Bloch theory in condensed matter physics (see Section \ref{section:bandinsulators}).}. Finally, for a general twisted dynamical system $(G,\mathcal{A},\alpha,\sigma)$, we \emph{define} the twisted representation group of $(G,\mathcal{A},\alpha,\sigma)$ to be $K_0(\mathcal{A}\rtimes_{(\alpha,\sigma)}G)$, which subsumes all the earlier definitions.

\subsection{Preliminaries on graded modules}
A graded module for a graded algebra $\mathcal{A}$ is an (ungraded) $\mathcal{A}$-module $W$ which admits a direct sum decomposition into $W=W_0\oplus W_1$, such that $\mathcal{A}_iW_j\subset W_{i+j\,(\text{mod}\,2)}$. The \emph{right} parity-reversed module $W^\Pi$ has the same underlying vectors as $W$ but with the reversed grading, $W_0^\Pi=W_1, W_1^\Pi=W_0$, and has the same graded action of $\mathcal{A}$. We write $\pi_R$ for the map $W\rightarrow W^\Pi$ fixing the underlying vectors. The \emph{left} parity reversal ${}^\Pi W$ also has ${}^\Pi W_0=W_1, {}^\Pi W_1=W_0$ but the $\mathcal{A}$-action is, on homogeneous elements,
\begin{equation}
	a\cdot (\pi_L (w))=\pi_L((-1)^{|a|}a\cdot w),\qquad a\in\mathcal{A}_0\cup\mathcal{A}_1,w\in W,\nonumber
\end{equation}
where $|a|\in\mathbb{Z}_2$ denotes the parity of $a$, and $\pi_L:W\rightarrow {}^\Pi W$ fixes the underlying vectors. Both $\pi_R$ and $\pi_L$ are odd maps and involutary operations on graded $\mathcal{A}$-modules: $\pi_R^2=\mathrm{id}=\pi_L^2$. We write $w^\pi\coloneqq\pi_R(w)$ and ${}^\pi w\coloneqq\pi_L(w)$, which distinguishes them from $w\in W$. As graded $\mathcal{A}$-modules, $W^\Pi$ and ${}^\Pi W$ are equivalent, under the even map $\varphi:W^\Pi \ni w^\pi\equiv w_0^\pi+w_1^\pi\mapsto {}^\pi w_0-{}^\pi w_1\in {}^\Pi W$. Despite this, $\pi_R$ commutes with $\mathcal{A}$, whereas $\pi_L$ \emph{graded} commutes with $\mathcal{A}$.

For a graded unital algebra $\mathcal{A}$, a graded finitely-generated free $\mathcal{A}$-module is one of the form $\mathcal{A}^m\oplus (\mathcal{A}^\Pi)^n\eqqcolon \mathcal{A}^{m|n}$, where $\mathcal{A}$ is regarded as a graded $\mathcal{A}$-module by left multiplication on itself, and $\mathcal{A}^\Pi$ is its right parity reverse\footnote{Note that the left parity reverse ${}^\Pi \mathcal{A}$ can also be used in this definition.}. A graded finitely-generated projective (f.g.p.) $\mathcal{A}$-module is defined to be a graded $\mathcal{A}$-module which is a direct summand of $\mathcal{A}^{m|n}$ for some $(m,n)$. It can be shown \cite{hazrat2014graded} that a graded f.g.p.\ $\mathcal{A}$-module is the same thing as a graded $\mathcal{A}$-module which is f.g.p.\ in the ungraded sense. In what follows, all modules are assumed to be f.g.p.\ unless otherwise stated.

\subsection{Super-representation groups of graded algebras}\label{superrepresentationdefinitionsection}
For a graded algebra $\mathcal{A}$, we define $\mathcal{G}\mathcal{V}(\mathcal{A})$ to be the commutative monoid of graded equivalence classes of graded $\mathcal{A}$-modules, under the direct sum operation. 
\begin{definition}[Trivial graded module]\label{definition:trivialgradedmodule}
A graded module for a unital graded $C^*$-algebra $\mathcal{A}$ is \emph{trivial} if it admits an odd involution $\mathcal{I}$ with $\mathcal{I}^2=+1$ which graded commutes with the action of $\mathcal{A}$,
\begin{equation}
	a\cdot (\mathcal{I}w)=(-1)^{|a|}\mathcal{I}(a\cdot w),\qquad a\in\mathcal{A}_0\cup\mathcal{A}_1, w\in W.\nonumber
\end{equation}
We define $\mathcal{T}(\mathcal{A})$ to be the set of graded equivalence classes of trivial graded $\mathcal{A}$-modules.
\end{definition}
The set $\mathcal{T}(\mathcal{A})$ is closed under direct sum and so forms a submonoid of $\mathcal{G}\mathcal{V}(\mathcal{A})$. It generates an equivalence relation on $\mathcal{G}\mathcal{V}(\mathcal{A})$ as follows: $W\sim W'$ iff there exists $T,T'\in\mathcal{T}(\mathcal{A})$ such that $W\oplus T\cong W'\oplus T'$. Note that this relation is a congruence on $\mathcal{G}\mathcal{V}(\mathcal{A})$, i.e.\ $W\sim W'$ and $Y\sim Y'$ implies that $W\oplus Y\sim W'\oplus Y'$.
\begin{definition}[Super-representation group]\label{definition:superrepresentationgroup}
The \emph{super-representation} group $\mathcal{S}\mathcal{R}(\mathcal{A})$ of a graded unital $C^*$-algebra $\mathcal{A}$ is the quotient set $\mathcal{G}\mathcal{V}(\mathcal{A})/\sim$.
\end{definition}
It is easy to see that ${}^\Pi W$ is the inverse of $W$ in the quotient $\mathcal{S}\mathcal{R}(\mathcal{A})$, since $W\oplus {}^\Pi W$ admits the odd trivialising involution $\mathcal{I}=\pi_L\oplus\pi_L$. It follows that $[{}^\Pi W]=-[W]$ in $\mathcal{S}\mathcal{R}(\mathcal{A})$, so $\mathcal{S}\mathcal{R}(\mathcal{A})$ is indeed a group. Definition \ref{definition:superrepresentationgroup} encapsulates the notion that a graded $\mathcal{A}$-module is ``cancelled'' by its left parity reverse. Furthermore, $\mathcal{T}(\mathcal{A})$ is closed under $\pi_L$: suppose $\mathcal{I}$ is an odd trivialising involution for $T$, then the operator $\tilde{\mathcal{I}} \coloneqq \pi_L\circ\mathcal{I}\circ\pi_L$ is an odd trivialising involution for ${}^\Pi T$.

An ungraded $C^*$-algebra $\mathcal{A}$ may be regarded as a graded algebra which is purely even. It is easy to see that $\mathcal{S}\mathcal{R}(\mathcal{A})$ and $K_0(\mathcal{A})$ are isomorphic, where in the latter the grading on $\mathcal{A}$ is ignored. An element in $\mathcal{S}\mathcal{R}(\mathcal{A})$ is represented by a graded $\mathcal{A}$-module $W=W_0\oplus W_1$, where the summands $W_0$ and $W_1$ are ungraded $\mathcal{A}$-modules. The corresponding element in $K_0(\mathcal{A})$ is the (class of the) virtual module $[W_0\ominus W_1]$. When $\mathcal{A}$ has a non-trivial grading, many possible ``representation groups'' exist \cite{landweber2005representation,landweber2006twisted,hazrat2014graded}. Besides $\mathcal{S}\mathcal{R}(\mathcal{A})$ and $K_0(\mathcal{A})$, there is also the \emph{graded representation group} $\mathcal{G}\mathcal{R}(\mathcal{A})$, defined to be the Grothendieck completion of $\mathcal{G}\mathcal{V}(\mathcal{A})$. Note that $\mathcal{G}\mathcal{R}(\mathcal{A})$ reduces to two copies of $K_0(\mathcal{A})$ for a purely even $\mathcal{A}$.

\subsubsection*{Non-unital graded algebras}
If $\mathcal{A}$ is a graded non-unital $C^*$-algebra, we define $\mathcal{S}\mathcal{R}(\mathcal{A})$ through a procedure similar to that of $K_0$ for non-unital ungraded algebras. As an ungraded algebra, $\mathcal{A}$ has its usual $C^*$-algebra unitisation $\mathcal{A}^+=\{(a,\lambda)\,:\,a\in\mathcal{A},\lambda\in\mathbb{F}\}$ with its unique $C^*$-norm, component-wise sum and adjoint, and multiplication given by $(a,\lambda)(b,\mu)=(ab+\lambda b + \mu a, \lambda\mu)$. We assign the grading $(\mathcal{A}^+)_0=\{(a,\lambda)\,:\,a\in\mathcal{A}_0,\lambda\in \mathbb{F}\},\,(\mathcal{A}^+)_1=\{(a,0)\,:\,a\in\mathcal{A}_1\}$, then $\mathcal{A}$ is a graded two-sided ideal in $\mathcal{A}^+$, i.e., $\mathcal{A}$ is a two-sided ideal in the ungraded sense, and $\mathcal{A}_i=\mathcal{A}\cap (\mathcal{A}^+)_i$. The quotient algebra $\mathcal{A}^+/\mathcal{A}=\mathbb{F}$ is purely even. Let $f:\mathcal{A}\rightarrow\mathcal{B}$ be an even homomorphism between unital graded algebras. Then $\mathcal{B}$ becomes a graded $\mathcal{A}$-bimodule via
\begin{equation}
	a\cdot b\coloneqq  f(a)b,\quad 	b\cdot a\coloneqq  bf(a),\qquad a\in\mathcal{A},\,b\in\mathcal{B}.\nonumber
\end{equation}
The map $f$ induces a graded $\mathcal{B}$-module from a graded $\mathcal{A}$-module $W$ via \mbox{$f_*(W)\coloneqq \mathcal{B}\,\hat{\otimes}_\mathcal{A}\,W$}, whose homogeneous subspaces $(f_*(W))_i$ are spanned by \mbox{$\{b\,\hat{\otimes} w\,:\, |b|+|w|=i\}$}. We verify that $\mathcal{G}\mathcal{V}(\cdot)$, along with $f\mapsto f_*$ (on equivalence classes), is a covariant functor from graded unital $C^*$-algebras to abelian monoids, then Grothendieck completion gives an induced homomorphism $\mathcal{G}\mathcal{R}(\mathcal{A})\xrightarrow{f_*}\mathcal{G}\mathcal{R}(\mathcal{B})$. These constructions are consistent with the usual $K$-theory ones when $\mathcal{A},\mathcal{B}$ are regarded as ungraded algebras. We can now define, for a non-unital graded algebra $\mathcal{A}$,
\begin{eqnarray}
K_0(\mathcal{A})&\coloneqq &\mathrm{ker}(K_0(\mathcal{A}^+)\xrightarrow{p_*}K_0(\mathbb{F})),\nonumber\\
\mathcal{G}\mathcal{R}(\mathcal{A})&\coloneqq &\mathrm{ker}(\mathcal{G}\mathcal{R}(\mathcal{A}^+)\xrightarrow{p_*} \mathcal{G}\mathcal{R}(\mathbb{F}))\nonumber
\end{eqnarray}
where $p$ is the projection $\mathcal{A^+}\xrightarrow{p}\mathbb{F}$. As for the super-representation group, we can verify that a trivial $T\in \mathcal{T}(\mathcal{A})$ with trivialising operator $\mathcal{I}$ gets mapped to a trivial $f_*(T)\in\mathcal{T}(\mathcal{B})$, with trivialising operator $\mathcal{I}':b\,\hat{\otimes}\,t\mapsto (-1)^{|b|}b\,\hat{\otimes}\,\mathcal{I}t$. Thus, we can define, for a non-unital graded $\mathcal{A}$,
\begin{equation}
\mathcal{S}\mathcal{R}(\mathcal{A})\coloneqq \mathrm{ker}(\mathcal{S}\mathcal{R}(\mathcal{A}^+)\xrightarrow{p_*} \mathcal{S}\mathcal{R}(\mathbb{F})).\nonumber
\end{equation}

\begin{remark}\label{leftversusrightremark}
In \cite{landweber2005representation}, the author used the right parity reversal in his definition of $\mathcal{S}\mathcal{R}(\cdot)$, in order to make contact with $KO^{-n}(\star), K^{-n}(\star)$ (the real and complex $K$-theory of a point) more directly; see Remark 4.1 in his paper. His odd trivialising involution is required to commute, rather than graded commute, with the $\mathcal{A}$-action. This has the effect of identifying $\mathcal{S}\mathcal{R}(Cl_{n,0})$ with $K^{-n}(\star)$, whereas our definition leads to $\mathcal{S}\mathcal{R}(Cl_{0,n})\cong K^{-n}(\star)$, see Remark \ref{remark:trivialityinvolutionvsantiinvolution}. We have chosen a different convention for two reasons. First, when $\mathcal{A}$ is the Clifford algebra associated with the symmetry data $(A,\sigma)$ of a $CT$-subgroup, we recover the $d=0$ column of the periodic table of Kitaev in the correct order. Second, the trivialising operator $\mathcal{I}$ has an interpretation in terms of the difference-group $\mathbf{K}_0(\mathcal{A})$ as defined in Section \ref{section:karoubidifferenceconstruction}, see Remark \ref{remark:subdifferencegroup}. 
\end{remark}

\begin{example}
For $\mathcal{A}=\mathbb{R}$, $K_0(\mathbb{R})=\mathbb{Z}$ is generated by the vector space $\mathbb{R}$. If we give $\mathcal{A}=\mathbb{R}$ the purely even grading, its graded modules are $\mathbb{R}^{m|n}, (m,n)\in\mathbb{N}\oplus\mathbb{N}$; thus, $\mathcal{G}\mathcal{R}(\mathbb{R})\cong\mathbb{Z}\oplus\mathbb{Z}$. Trivial graded $\mathbb{R}$-modules are of the form $\mathbb{R}^{s|s}$. Up to the addition of trivial modules, $\mathbb{R}^{m|n}$ can be represented by $\mathbb{R}^{m-n|0}$ or $\mathbb{R}^{0|n-m}$, depending on $m-n\geq 0$ or $n-m>0$. Thus $\mathcal{S}\mathcal{R}(\mathbb{R})\cong\mathbb{Z}\cong K_0(\mathbb{R})$.
\end{example}

\begin{example}
The Clifford algebras $Cl_{r,s}$ and $\mathbb{C}l_n$ are naturally graded $C^*$-algebras, with odd (resp.\ even) products of Clifford generators being odd (resp.\ even), and positive (resp.\ negative) generators being self-adjoint (resp.\ skew-adjoint). They are semisimple, so a finite-dimensional (graded) Clifford module is a f.g.p.\ Clifford module, which can be turned into a (graded) $*$-representation by an appropriate choice of inner product. The group $\mathcal{G}\mathcal{R}(Cl_{r,s})$ (resp.\ $\mathcal{G}\mathcal{R}(\mathbb{C}l_n)$) is freely generated by the equivalence classes of irreducible graded $Cl_{r,s}$ (resp.\ $\mathbb{C}l_n$) representations, and we have (see Chapter I.5 of \cite{lawson1989spin})
\begin{alignat}{3}
	\mathcal{G}\mathcal{R}(Cl_{r,s})&\cong  &\begin{cases} \mathbb{Z}\oplus\mathbb{Z}\;\; & r-s=0,4\, (\mathrm{mod}\,2),\\ \mathbb{Z} & \mathrm{otherwise};\end{cases}\qquad
	\mathcal{G}\mathcal{R}(\mathbb{C}l_n)&\cong &\begin{cases} \mathbb{Z}\oplus\mathbb{Z} \;\;& n\,\,\mathrm{even},\\ \mathbb{Z} & n\,\,\mathrm{odd}.\end{cases}\nonumber
\end{alignat}
A trivial graded $Cl_{r,s}$-module is a graded $Cl_{r,s+1}$-module with the action of the $(s+1)^{\mathrm{th}}$ positive Clifford generator forgotten. The inclusion $i:Cl_{r,s}\hookrightarrow Cl_{r,s+1}$ induces a ``restriction-of-scalars'' homomorphism $i^*:\mathcal{G}\mathcal{R}(Cl_{r,s+1})\rightarrow \mathcal{G}\mathcal{R}(Cl_{r,s})$. 
\begin{lemma}\label{lemma:cliffordsuperrepresentationisomorphism}
There are isomorphisms
\begin{eqnarray}
	\mathcal{S}\mathcal{R}(Cl_{r,s}) & \cong & \mathcal{G}\mathcal{R}(Cl_{r,s})/i^*\mathcal{G}\mathcal{R}(Cl_{r,s+1}),\nonumber\\
	\mathcal{S}\mathcal{R}(\mathbb{C}l_n)  &\cong & \mathcal{G}\mathcal{R}(\mathbb{C}l_n)/i^*\mathcal{G}\mathcal{R}(\mathbb{C}l_{n+1}).\nonumber
\end{eqnarray}
\end{lemma}
\begin{proof}
$W'\oplus {}^\Pi W'$ is a graded $Cl_{r,s+1}$-module for any graded $Cl_{r,s}$-module $W'$, so any element $[W\ominus W']\in \mathcal{G}\mathcal{R}(Cl_{r,s})$ can be written as $[(W\oplus {}^\Pi W') \ominus 0]$ after passing to the quotient. Let $[W]$ denote an element of $\mathcal{S}\mathcal{R}(Cl_{r,s})$ represented by a graded $Cl_{r,s}$-module $W$. The required isomorphism is $p:[W]\mapsto [W\ominus 0]$, with the inverse map $q:[W\ominus 0]\mapsto [W]$. It is easy to check that $p$ and $q$ are well-defined; the complex case is similar.
\qed
\end{proof}

A graded $Cl_{r,s}$-module $W$ is also a graded $Cl_{s,r}$-module by a simple adjustment of the actions of the Clifford generators. If $\Gamma$ is the grading operator on $W$, and $\{\mathsf{e}_1,\ldots,\mathsf{e}_r,\mathsf{f}_1,\ldots,\mathsf{f}_s\}$ are the odd anticommuting operators for the $r$ negative and $s$ positive $Cl_{r,s}$ generators, then $\mathsf{e}'_j\coloneqq\Gamma \mathsf{f}_j$ and $\mathsf{f}'_i\coloneqq\Gamma \mathsf{e}_i$ give $s$ negative and $r$ positive odd anticommuting operators on $W$. Furthermore, if a graded $Cl_{r,s}$-module $W$ admits the action of an extra negative Clifford generator $\mathsf{e}_{r+1}$, then the operator $\mathsf{f}'_{r+1}\coloneqq \Gamma \mathsf{e}_{r+1}$ acts as an extra \emph{positive} Clifford generator when we regard $W$ as a graded $Cl_{s,r}$-module. Thus, we have shown that $\mathcal{G}\mathcal{R}(Cl_{r,s})\cong \mathcal{G}\mathcal{R}(Cl_{s,r})$ in a manner that respects the homomorphisms $i^*$,
\begin{equation}
	\mathcal{G}\mathcal{R}(Cl_{r,s})/i^*\mathcal{G}\mathcal{R}(Cl_{r,s+1})\cong \mathcal{G}\mathcal{R}(Cl_{s,r})/i^*\mathcal{G}\mathcal{R}(Cl_{s+1,r}).\label{switchcliffordindices}
\end{equation}

The Atiyah--Bott--Shapiro isomorphisms \cite{atiyah1964clifford},
\begin{alignat}{3}
	K_n(\mathbb{R})&\cong & KO^{-n}(\star)& \cong   \mathcal{G}\mathcal{R}(Cl_{r,s})/i^*\mathcal{G}\mathcal{R}(Cl_{r+1,s}),\quad n=r-s\,\mathrm{(mod\,8)}\nonumber\\
	K_n(\mathbb{C})&\cong & K^{-n}(\star) \,& \cong  \mathcal{G}\mathcal{R}(\mathbb{C}l_n)/i^*\mathcal{G}\mathcal{R}(\mathbb{C}l_{n+1}).\nonumber
\end{alignat}
together with Lemma \ref{lemma:cliffordsuperrepresentationisomorphism} and \eqref{switchcliffordindices}, says that the super-representation groups of the Clifford algebras are the $K$-theory groups of a point:

\begin{corollary}\label{corollary:cliffordsuperrepresentationktheory}
\begin{subequations}
\begin{alignat}{4}
	\mathcal{S}\mathcal{R}(Cl_{r,s}) & \cong & K_{s-r}(\mathbb{R}) & \cong KO^{r-s}(\star),\label{cliffordsuperrepresentationkgroupreal}\\
	\mathcal{S}\mathcal{R}(\mathbb{C}l_n) & \cong & K_n(\mathbb{C})\,\, & \cong K^{-n}(\star)\label{cliffordsuperrepresentationkgroupcomplex}
\end{alignat}
\end{subequations}
\end{corollary}

\end{example}

\begin{remark}\label{remark:trivialityinvolutionvsantiinvolution}
If $\mathcal{A}$ is complex, it does not matter whether the $\mathcal{I}$ is defined to be an involution $\mathcal{I}^2=+1$ or an anti-involution $\mathcal{I}^2=-1$, since multiplication by $\im$ turns one into the other. This does \emph{not} work if $\mathcal{A}$ is real, since $\im$ is unavailable: if we had asked for an anti-involution (thus an extra \emph{negative} Clifford generator) in Definition \ref{definition:trivialgradedmodule}, we would have arrived at $\mathcal{S}\mathcal{R}(Cl_{r,s})\cong K_{r-s}(\mathbb{R})$ instead. Similarly, if we had required $\mathcal{I}$ to commute (rather than graded commute) with $Cl_{r,s}$, then $\mathcal{I}\Gamma$ becomes an anti-involution which graded commutes with $Cl_{r,s}$, and we would obtain a similar same reversal in the $K$-theory degree.
\end{remark}

\begin{remark}
One can also define higher super-representation groups $\mathcal{S}\mathcal{R}^{-n}(\mathcal{A})\coloneqq \mathcal{S}\mathcal{R}(\mathcal{A}\hat{\otimes}\mathcal{D}_n^\mathbb{F})$, where $\mathcal{D}_n^\mathbb{F}$ are central super-division algebras over $\mathbb{F}$ representing the super-Brauer group \cite{donovan1970graded} $\mathrm{sBr}(\mathbb{F})$ of $\mathbb{F}$. Each $\mathcal{D}_n^\mathbb{F}$ is Morita equivalent (in the super-sense) to a Clifford algebra. Super-Brauer multiplication is compatible with triviality as defined in Definition \ref{definition:trivialgradedmodule}, and $\mathcal{S}\mathcal{R}_\mathrm{total}(\mathcal{A})\coloneqq\bigoplus_n \mathcal{S}\mathcal{R}^{-n}(\mathcal{A})$ can be given the structure of a $\mathrm{sBr}(\mathbb{F})$-graded module over $\mathcal{S}\mathcal{R}_\mathrm{total}(\mathbb{F})$. For $\mathcal{A}=\mathbb{F}$, this recovers the usual real or complex $K$-theory ring of a point. Details can be found in \cite{thiang2014thesis}.
\end{remark}

\subsection{Some notions of topological triviality}\label{quotientingtopologicaltriviality}
It was suggested by Kitaev \cite{kitaev2009periodic}, in the context of band insulators, that we should consider two gapped phases as being ``equivalent'', if they become equivalent (in some pre-defined ordinary sense) upon adding some ``trivial'' bands. He argued that one can generally augment a given system by ``\ldots a set of local, disjoint modes, like inner atomic shells. This corresponds to adding an extra flat band on an insulator.'' This leads to stable equivalence as an equivalence relation for gapped phases modelled on vector bundles.

He also suggested that an admissible system $(X)$ is ``\ldots effectively cancelled by its particle-hole conjugate $(-X)$, resulting in a trivial system''. In Definition \ref{definition:trivialgradedmodule}, we introduced a notion of trivial graded modules for $\mathcal{A}$. In particular, the sum of a graded module with its particle-hole conjugate (left parity reversal) becomes trivial in $\mathcal{S}\mathcal{R}(\mathcal{A})$. Thus, Definition \ref{definition:superrepresentationgroup} formalises algebraically this second (distinct) notion of triviality. We note that Freed--Moore \cite{freed2013twisted} defined ``reduced topological phases'' via a related notion of trivial representations, which only requires the existence of an odd operator $\mathcal{I}$ that graded commutes with a graded group action. They do not impose an involution or anti-involution condition on $\mathcal{I}$. As discussed in Remark \ref{remark:trivialityinvolutionvsantiinvolution}, choosing one or the other yields $K$-theory indices running in \emph{opposite} directions.

Another interpretation of the $K$-theoretic invariants of $\mathcal{A}$, hinted at in \cite{kitaev2009periodic}, is in terms of \emph{differences} of symmetry-compatible phases. Unfortunately, this point of view seems to have been ignored subsequently, although it actually provides the most powerful and consistent way of understanding the role of $K$-theory in the study of gapped phases. This is the subject of the next section.

\section{The $K$-theoretic difference-group of symmetry-compatible gapped Hamiltonians}\label{section:karoubidifferenceconstruction}
Standard presentations of $K$-theory in terms of Grothendieck completions and suspension constructions (e.g., in \cite{lawson1989spin,blackadar1998k,wegge1993k}) do not directly relate to the study of gapped phases, because the latter entails studying \emph{graded} symmetry algebras whenever charge-conjugating symmetries are present. We have already seen this in Section \ref{section:cliffordtenfoldway}, when studying representations of $CT$-subgroups. Super-algebra is actually a unifying idea in $K$-theory: we say this in the relationship between super-representation groups of Clifford algebras and $K$-theory groups of a point (Corollary \ref{corollary:cliffordsuperrepresentationktheory}). Thus, we will instead utilise a version of $K$-theory introduced by Karoubi \cite{karoubi1978k}, which remains well-defined for graded $C^*$-algebras, and is consistent with the Grothendieck group-type definitions for ungraded $C^*$-algebras. The central object is the $K$-theoretic \emph{difference-group} $\mathbf{K}_0(\mathcal{A})$ (Definition \ref{definition:differencegroup}), which has good properties with natural physical interpretations.

Recall that the grading operator $\Gamma$ of a graded f.g.p.\ module $W$ for the graded algebra $\mathcal{A}=\mathcal{B}\rtimes_{(\alpha,\sigma)}G$ can be interpreted as a spectrally-flattened gapped Hamiltonian compatible with the symmetry data $(G,c,\mathcal{B},\alpha,\sigma)$. Strictly speaking, this interpretation requires $W$ to be a graded Hilbert space with a graded $*$-representation of $\mathcal{A}$. Although $\mathbf{K}_0(\mathcal{A})$ will be defined in terms of graded f.g.p.\ modules for $\mathcal{A}$, we may regard the latter as graded Hilbert $\mathcal{A}$-modules\footnote{A \emph{Hilbert $C^*$-module over $\mathcal{A}$} \cite{wegge1993k,blackadar1998k,lance1995hilbert}, is an $\mathcal{A}$-module with an $\mathcal{A}$-valued ``inner product'' $\langle\cdot,\cdot\rangle$, whose associated norm $\normlll{x}=\norm{\langle x,x\rangle}^{1/2}$ is complete. A f.g.p.\ $\mathcal{A}$-module can be endowed with the structure of a Hilbert $\mathcal{A}$-module, see Theorem 15.4.2 in \cite{wegge1993k}.}, on which there is a notion of self-adjointness for the adjointable (and bounded) operators $\mathscr{B}(W)$. In fact, $\mathscr{B}(W)$, along with the grading operator $\Gamma$, form a (evenly) graded $C^*$-algebra under the operator norm. The self-adjoint unitary grading operator for $\mathscr{B}(W)$ is $\Gamma$, and it makes sense to talk about continuous functions of $\Gamma$ and their homotopies.

\begin{example}[Noncommutative Bloch theory \cite{gruber2001noncommutative}] When $\mathcal{A}=\mathbb{C}\rtimes\mathbb{Z}^d\cong C(\mathbb{T}^d)$, a graded f.g.p.\ $C(\mathbb{T}^d)$-module is, as a Hilbert $C(\mathbb{T}^d)$-module, the set of continuous sections of some graded Hermitian vector bundle over $\mathbb{T}^d$, on which there is a continuous (i.e., $C(\mathbb{T}^d)$-valued) fibre-wise inner product. The restriction of the grading operator to a fibre can be viewed as the flattened version of a gapped Bloch Hamiltonian on that fibre. The positively-graded sub-bundle is the conduction band, while the negatively-graded sub-bundle is the valence band. The usual Bloch--Floquet picture of a direct integral decomposition of $L^2(\mathbb{R}^d)$ over the character space $\mathbb{T}^d$, can be recovered by passing to the GNS representation induced by a faithful trace on $C(\mathbb{T}^d)$ (e.g., integration over the Haar measure on $\mathbb{T}^d$). For a general noncommutative graded algebra $\mathcal{A}$, we interpret the grading operator $\Gamma$ on a f.g.p. graded $\mathcal{A}$-module $W$ as a spectrally-flattened gapped Hamiltonian on the noncommutative graded ``vector bundle'' corresponding to $W$.
\end{example}

Given an ungraded $\mathcal{A}$-module $W$, we can consider the set $\mathrm{Grad}_\mathcal{A}(W)$ of possible grading operators on $W$ turning it into a graded $\mathcal{A}$-module. There is a standard Banach space structure on $W$ (either from its Hilbert $\mathcal{A}$-module structure or induced from the free module $\mathcal{A}^n$ which it is a direct summand of), which determines a norm topology on the bounded linear maps $W\rightarrow W$. Thus, there is an induced topology on $\mathrm{Grad}_\mathcal{A}(W)\subset\mathrm{End}_\mathcal{A}(W)\subset\mathrm{End}(W)$ (e.g.\ see I.6.22 of \cite{karoubi1978k}, 11.2 of \cite{blackadar1998k}, or Chapter 15 of \cite{wegge1993k}). We can then talk about the homotopy classes of symmetry-compatible flattened Hamiltonians on $W$: 
\begin{definition}[Symmetry compatible gapped Hamiltonians]\label{definition:symmetrycompatiblehamiltonianshomotopy}
Let $(G,c,\mathcal{B},\alpha,\sigma)$ be a graded twisted $C^*$-dynamical system, and suppose $\mathcal{A}=\mathcal{B}\rtimes_{(\alpha,\sigma)}G$ is unital. Let $W$ be an ungraded f.g.p.\ $\mathcal{A}$-module. We call $\mathrm{Grad}_\mathcal{A}(W)$ the set of \emph{$(G,c,\mathcal{B},\alpha,\sigma)$-compatible}, or \emph{$\mathcal{A}$-compatible}, or simply \emph{symmetry-compatible (flattened) gapped Hamiltonians on $W$}. Two grading operators $\Gamma_1, \Gamma_2\in\mathrm{Grad}_\mathcal{A}(W)$ are said to be \emph{homotopic} if there is a norm-continuous path between $\Gamma_1$ and $\Gamma_2$ within $\mathrm{Grad}_\mathcal{A}(W)$; in this case, we write $\Gamma_1\sim_h\Gamma_2$.
\end{definition}

Recall that $\mathrm{End}_\mathcal{A}(W)$ can be given the structure of a $C^*$-algebra. We may then assume that the Hamiltonians $\Gamma_i$ are self-adjoint and unitary, and that a homotopy between $\Gamma_1$ and $\Gamma_2$ takes place within such self-adjoint grading operators (see 4.6 of \cite{blackadar1998k}). Intuitively, $\Gamma_1\sim_h\Gamma_2$ means that the two Hamiltonians can be continuously deformed into one another, while respecting the symmetries encoded by the algebra $\mathcal{A}$, and maintaining the gapped condition.

The set of homotopy classes of symmetry-compatible gapped Hamiltonians on $W$ (i.e.\ $\pi_0(\mathrm{Grad}_\mathcal{A}(W))$) is of some interest in the literature \cite{kitaev2009periodic,stone2011symmetries,ryu2010topological}, although the explicit reference to $W$ is usually not made. However, these sets do not have any additional structure, much less that of an abelian group, and are difficult to compute. There is a commonly proposed ``solution'': for some special $\mathcal{A}$ and families of modules $W^{(N)}$, the sets $\mathrm{Grad}_\mathcal{A}(W^{(N)})$ from a family of symmetric spaces. Their ``large-$N$'' limit is identified with a loop space of the stable orthogonal or unitary group. These loop spaces are classifying spaces $C_n$ and $R_n$ for (topological) $K$-theory and $KO$-theory respectively. One might be tempted to identify the $K$-theory group $K^{-n}(\star)\cong\pi_0(C_n)$ as the set of ``stable'' homotopy classes of symmetry-compatible Hamiltonians, but such a connection has not been made precise. Furthermore, the correct theory to use in the presence of antiunitary symmetries is Atiyah's $KR$-theory \cite{atiyah1966k}, which differs from $KO$-theory when $X$ is not a point.

In a modified version of Karoubi's $K$-theory, the elements of the $K$-theory of $\mathcal{A}$ represent \emph{differences} or \emph{obstructions} between $\mathcal{A}$-compatible Hamiltonians.
\begin{definition}[Trivial differences between Hamiltonians]\label{definition:trivialtriple}
Let $W$ be a graded f.g.p.\ module for a graded unital $C^*$-algebra $\mathcal{A}$, and let $\Gamma_1,\Gamma_2\in\mathrm{Grad}_\mathcal{A}(W)$ be a pair of compatible gapped Hamiltonians. We call $(W,\Gamma_1,\Gamma_2)$ a \emph{trivial triple} if $\Gamma_1\sim_h\Gamma_2$ in $\mathrm{Grad}_\mathcal{A}(W)$.
\end{definition}
A triple $(W,\Gamma_1,\Gamma_2)$ represents the (ordered) difference between two $\mathcal{A}$-compatible gapped Hamiltonians on $W$, and we do not distinguish between two Hamiltonians which can be continuously deformed into one another. We want to be able to consider all graded modules concurrently, and to combine two or more systems with the same symmetries. The direct sum operation gives a natural commutative monoid structure to the collection $\mathrm{Grad}_\mathcal{A}$ of all triples, where some obvious identifications have been made to ensure commutativity and associativity. The set of trivial triples forms a submonoid $\mathrm{Grad}_\mathcal{A}^t$. 
\begin{definition}[Difference-group of Hamiltonians]\label{definition:differencegroup}
Let $\mathbf{K}_0(\mathcal{A})$ be the quotient monoid of $\mathrm{Grad}_\mathcal{A}$ by the congruence generated by $\mathrm{Grad}_\mathcal{A}^t$, i.e., $[W,\Gamma_1,\Gamma_2]=[W',\Gamma'_1,\Gamma'_2]$ in $\mathbf{K}_0(\mathcal{A})$ iff there are trivial triples $(F,\zeta_1,\zeta_2)$ and $(F',\zeta'_1,\zeta'_2)$ such that $(W\oplus F,\Gamma_1\oplus\zeta_1,\Gamma_2\oplus\zeta_2)=(W'\oplus F',\Gamma'_1\oplus\zeta'_1,\Gamma'_2\oplus\zeta'_2)$ in $\mathrm{Grad}_\mathcal{A}$. We call $\mathbf{K}_0(\mathcal{A})$ the \emph{difference-group of $\mathcal{A}$-compatible gapped Hamiltonians}.
\end{definition}
\begin{proposition}
$\mathbf{K}_0(\mathcal{A})$ is an abelian group, with $[W,\Gamma_1,\Gamma_2]=-[W,\Gamma_2,\Gamma_1]$. Furthermore, two isomorphic triples (in the natural sense) define the same class in $\mathbf{K}_0(\mathcal{A})$.
\end{proposition}
\begin{proof}
$\Gamma_1\oplus\Gamma_2\sim_h \Gamma_2\oplus\Gamma_1$ in $\mathrm{Grad}_\mathcal{A}(W\oplus W)$ via the homotopy
\begin{equation}
	\Gamma(\theta)=\begin{pmatrix}
	\cos \theta & -\sin \theta \\ \sin \theta & \cos \theta
	\end{pmatrix}
	\begin{pmatrix}
	\Gamma_1 & 0 \\ 0 & \Gamma_2
	\end{pmatrix}
	\begin{pmatrix}
	\cos \theta & \sin \theta \\ -\sin \theta & \cos \theta
	\end{pmatrix},\quad \theta\in[0,\frac{\pi}{2}],\label{homotopyexample1}
\end{equation}
so $(W\oplus W,\Gamma_1\oplus\Gamma_2,\Gamma_2\oplus\Gamma_1)$ is trivial, and we can write $[W,\Gamma_2,\Gamma_1]=-[W,\Gamma_1,\Gamma_2]$ in $\mathbf{K}_0(\mathcal{A})$. For isomorphic triples $(W,\Gamma_1,\Gamma_2)$ and $(W',\Gamma'_1,\Gamma'_2)$, let $\alpha:W\rightarrow W'$ be the isomorphism of ungraded $\mathcal{A}$-modules, such that $\Gamma'_i=\alpha\Gamma_i\alpha^{-1}, \,i=1,2$. Then $\Gamma_2\oplus\Gamma'_1\sim_h \Gamma_1\oplus\Gamma'_2$ in $\mathrm{Grad}_\mathcal{A}(W\oplus W')$ via the homotopy
\begin{equation}
		\Gamma(\theta)=\begin{pmatrix}
	\cos \theta & -\alpha^{-1}\sin \theta \\ \alpha\sin \theta & \cos \theta
	\end{pmatrix}
	\begin{pmatrix}
	\Gamma_1 & 0 \\ 0 & \Gamma'_2
	\end{pmatrix}
	\begin{pmatrix}
	\cos \theta & \alpha^{-1}\sin \theta \\ -\alpha\sin \theta & \cos \theta
	\end{pmatrix},\quad \theta\in[0,\frac{\pi}{2}].\nonumber
\end{equation}
Therefore, $0=[W\oplus W',\Gamma_1\oplus\Gamma'_2,\Gamma_2\oplus\Gamma'_1]=[W,\Gamma_1,\Gamma_2]-[W',\Gamma'_1,\Gamma'_2]$. Note that $\mathcal{A}$ acts diagonally, so it is easy to see that it graded commutes with $\Gamma(\theta)$ in both cases. \qed
\end{proof}
\begin{proposition}[Path-independence and homotopy-independence of differences]
The equation \mbox{$[W,\Gamma_1,\Gamma_2]+[W,\Gamma_2,\Gamma_3]=[W,\Gamma_1,\Gamma_3]$} holds in $\mathbf{K}_0(\mathcal{A})$. Furthermore, $[W,\Gamma_1,\Gamma_2]$ depends only on the homotopy class of $\Gamma_i$ in $\mathrm{Grad}_\mathcal{A}(W)$.
\end{proposition}
\begin{proof}
We need to show that $[W\oplus W\oplus W, \Gamma_1\oplus\Gamma_2\oplus\Gamma_3,\Gamma_2\oplus\Gamma_3\oplus\Gamma_1]$ is trivial. Since $\Gamma_2\oplus\Gamma_3\oplus\Gamma_1$ can be obtained from $\Gamma_1\oplus\Gamma_2\oplus\Gamma_3$ by conjugation with the permutation matrix 
\begin{equation}
	\begin{pmatrix}
	0 & 1 & 0 \\ 0 & 0 & 1 \\ 1 & 0 & 0
	\end{pmatrix}\in \mathrm{SO}(3),\nonumber
\end{equation}
and $\mathrm{SO}(3)$ is path-connected, it follows that $\Gamma_2\oplus\Gamma_3\oplus\Gamma_1\sim_h \Gamma_1\oplus\Gamma_2\oplus\Gamma_3$ in $\mathrm{Grad}_\mathcal{A}(W\oplus W\oplus W)$. If $\Gamma_i'\sim_h\Gamma_i,\,i=1,2$, then $(W,\Gamma'_1,\Gamma_1)$ and $(W,\Gamma_2,\Gamma'_2)$ are trivial, so $[W,\Gamma'_1,\Gamma'_2]=[W,\Gamma'_1,\Gamma_1]+[W,\Gamma_1,\Gamma_2]+[W,\Gamma_2,\Gamma'_2]=[W,\Gamma_1,\Gamma_2]$.\qed
\end{proof}

For non-unital $\mathcal{A}$, we define $\mathbf{K}_0(\mathcal{A})$ to be the kernel of the homomorphism $\mathbf{K}_0(\mathcal{A}^+)\rightarrow\mathbf{K}_0(\mathbb{F})\cong\mathbb{Z}$ induced by the projection $p:\mathcal{A}^+\rightarrow \mathbb{F}$. Here, a triple $(W,\Gamma_1,\Gamma_2)$ gets mapped to a triple $(p_*(W),p_*(\Gamma_1),p_*(\Gamma_2))$, where $p_*(W)=\mathbb{F}\hat{\otimes}_{\mathcal{A}^+}W$, $p_*(\Gamma_i)=1\hat{\otimes}\Gamma_i$, and $\Gamma_1\sim_h\Gamma_2$ implies $p_*(\Gamma_1)\sim_h p_*(\Gamma_2)$ as is required for consistency.

\begin{remark}\label{remark:subdifferencegroup}
The classes in the difference-group which can be written as $[W,\Gamma,-\Gamma]$ form a subgroup of $\mathbf{K}_0(\mathcal{A})$, and can be ``represented'' by the single graded $\mathcal{A}$-module $(W,\Gamma)$. The right parity reversal $W^\Pi=(W,-\Gamma)$ ``represents'' the class $[W,-\Gamma,\Gamma]$, which ``cancels'' $(W,\Gamma)$ in the sense that $[W,\Gamma,-\Gamma]=-[W,-\Gamma,\Gamma]$. The left parity reversal $({}^\Pi W,-\Gamma)$ ``represents'' $[{}^\Pi W,-\Gamma,\Gamma]$, which also ``cancels'' $(W,\Gamma)$ using, again, the homotopy \eqref{homotopyexample1}. Of course, if $W$ is homotopic to $W^\Pi$, i.e., $\Gamma\sim_h-\Gamma$, then $[W,\Gamma,-\Gamma]=0$. This occurs, for example, if $(W,\Gamma)$ admits a trivialising operator $\mathcal{I}$ in the sense of Definition \ref{definition:trivialgradedmodule}. Then $\Gamma(\theta)=(\cos\theta) \Gamma+(\sin\theta)\mathcal{I},\,\theta\in[0,\pi]$ provides a homotopy between $\Gamma$ and $-\Gamma$.
\end{remark}

\begin{remark}
When an arbitrary reference Hamiltonian $\Gamma_0$ on $W$ has been chosen, all other $\Gamma\in\mathrm{Grad}_\mathcal{A}(W)$ may be measured in relation to $\Gamma_0$ through the difference class $[W,\Gamma_0,\Gamma]$. Then two homotopic $\Gamma,\Gamma'\in\mathrm{Grad}_\mathcal{A}(W)$ differ from $\Gamma_0$ by the same amount: $0=[W,\Gamma,\Gamma']=[W,\Gamma_0,\Gamma']-[W,\Gamma_0,\Gamma]$. They may be said to be in the same phase relative to $\Gamma_0$. A canonical $\Gamma_0$ can sometimes (but not always) be chosen, for example, if $\mathcal{A}$ is evenly-graded (i.e.\ $\mathcal{A}$ contains a self-adjoint unitary which induces its grading), and $W$ is a free $\mathcal{A}$-module.
\end{remark}

Let us see how $\mathbf{K}_0(\cdot)$ is related to the more familiar $K$-theory groups. For purely even $\mathcal{A}$, our definition of $\mathbf{K}_0(\mathcal{A})$ is equivalent to either of Karoubi's two definitions of $K'^{0,0}(\mathcal{A})$ in III.4.15 and III.4.19 of \cite{karoubi1978k}, and $K'^{0,0}(\mathcal{A})$ is itself isomorphic to his $K^{0,0}(\mathcal{A})$ as defined in III.4.11 and II.2.13 of the same reference. Both $K'^{0,0}(\mathcal{A})$ (and $K^{0,0}(\mathcal{A})$) were shown to be isomorphic to the ordinary $K$-theory group $K_0(\mathcal{A})$ (Theorem III.4.12 of \cite{karoubi1978k}). In particular, a virtual module $[W_0\ominus W_1]\in K_0(\mathcal{A})$ corresponds to the element $[W_0\oplus W_1,1\oplus-1,-1\oplus 1]$ in $\mathbf{K}_0(\mathcal{A})$. Since Karoubi's $K'^{0,0}(\cdot)$ and our $\mathbf{K}_0(\cdot)$ continue to make sense and coincide for graded algebras $\mathcal{A}$, we shall take the difference-group $\mathbf{K}_0(\mathcal{A})$ to be a \emph{definition} of the $K_0$-group of a graded algebra $\mathcal{A}$ \cite{karoubi2008twisted}. We denote this group using bold-faced notation $\mathbf{K}_0(\mathcal{A})$, to avoid confusion with the ordinary $K$-theory group $K_0(\mathcal{A})$ in which $\mathcal{A}$ is regarded as an ungraded algebra.

If we also define $\mathbf{K}_{s,r}(\mathcal{A})\coloneqq\mathbf{K}_0(\mathcal{A}\hat{\otimes} Cl_{r,s})$ for real graded algebras $\mathcal{A}$, we obtain an alternative definition of Karoubi's $K^{r,s}(\mathcal{A})$ (or $K'^{r,s}(\mathcal{A})$) as defined in III.4.11 of \cite{karoubi1978k}. Due to the periodicity properties of the Clifford algebras, the $\mathbf{K}_{s,r}(\mathcal{A})$ and $K'^{r,s}(\mathcal{A})$ groups depend only on $(r-s)\,\mathrm{(mod\,8)}$. Thus, the singly-indexed groups $\mathbf{K}_n(\mathcal{A})\coloneqq\mathbf{K}_0(\mathcal{A}\hat{\otimes} Cl_{0,n})\cong \mathbf{K}_0(C_0(\mathbb{R}^n,\mathcal{A}))$ have a period-$8$ Bott periodicity. In the complex case, we can similarly define $\mathbf{K}_n(\mathcal{A})\coloneqq \mathbf{K}_0(\mathcal{A}\hat{\otimes}\mathbb{C}l_n)$, which only depend on $n\,\mathrm{(mod\,2)}$. Karoubi went on to prove the difficult result that the ``suspension'' operation $\mathcal{A}\mapsto\mathcal{A}\hat{\otimes}Cl_{0,1}$ (or $\mathcal{A}\mapsto\mathcal{A}\hat{\otimes}\mathbb{C}l_1$ in the complex case) is $K$-theoretically compatible with the usual notion of suspension, in the sense that $K'^{0,1}(\mathcal{A})\equiv K'^{0,0}(\mathcal{A}\hat{\otimes}Cl_{0,1})\cong K'^{0,0}(C_0(\mathbb{R},\mathcal{A}))=\mathbf{K}_0(C_0(\mathbb{R},\mathcal{A}))$ \cite{karoubi1968algebres,karoubi2008twisted}. Thus, for purely even real algebras $\mathcal{A}^\mathrm{ev}$,
\begin{equation}
\mathbf{K}_n(\mathcal{A}^\mathrm{ev})=\mathbf{K}_0(\mathcal{A}^\mathrm{ev}\hat{\otimes}Cl_{0,n})\cong K_0(C_0(\mathbb{R}^n,\mathcal{A}^\mathrm{ev}))\equiv K_n(\mathcal{A}^\mathrm{ev}),
\end{equation}
and similarly for the complex case.

\section{Computing $\mathcal{S}\mathcal{R}(\mathcal{A})$ and $\mathbf{K}_0(\mathcal{A})$ by decomposing $\mathcal{A}$}\label{section:decompositioncrossedproducts}
We state a very useful decomposition theorem for twisted crossed products, which facilitates the computation of some super-representation groups and difference-groups arising from physical examples in condensed matter applications. In some special cases, a description of these groups in terms of topological $K$-theory is possible, but this is not generic.

\begin{theorem}[Packer--Raeburn decomposition theorem \cite{packer1989twisted}]\label{theorem:packerraeburndecomposition}
Let $(G,c,\mathcal{A},\alpha,\sigma)$ be a graded twisted $C^*$-dynamical system, and let $N$ be a closed normal subgroup of $G$ in the kernel of $c$. There is an isomorphism of graded $C^*$-algebras
\begin{equation}
\mathcal{A}\rtimes_{(\alpha,\sigma)}G\cong (\mathcal{A}\rtimes_{(\alpha,\sigma)}N)\rtimes_{(\beta,\nu)}G/N,\label{twisteddecompositionformula}
\end{equation}
where the twisting pair $(\beta,\nu)$ is determined by a choice of Borel section $s:G/N\ni p\mapsto s_p\in G$ such that $s_{eN}=1$. For each $x\in G$, there is a $\gamma_x\in\mathrm{Aut}_\mathbb{F}(\mathcal{A}\rtimes_{(\alpha,\sigma)}N)$ such that
\begin{subequations}
\begin{alignat}{3}
	\gamma_x(a)&=\alpha_x(a)\equiv a^x, & a\in\mathcal{A},\label{modifiedautomorphism}\\
	\gamma_x(n)&=\sigma(x,n)\sigma(xnx^{-1},x)^{-1}xnx^{-1}, &\qquad n\in N,\label{modifiedembedding}
\end{alignat}
\end{subequations}
where the canonical embeddings $j_\mathcal{A},j_N$ are implied. The formulae for $(\beta,\nu)$ are, for $p,q\in G/N$,
\begin{subequations}
\begin{eqnarray}
	\beta_p\coloneqq\beta(p)&=&\gamma_{s_p},\label{decomposedtwistingpair1}\\
	\nu(p,q)&=&\sigma(s_p,s_q)\sigma(s_ps_qs_{pq}^{-1},s_{pq})^{-1}s_ps_qs_{pq}^{-1}. \label{decomposedtwistingpair2}
\end{eqnarray}
\end{subequations}
\end{theorem}
Part of the theorem says that up to isomorphism, the iterated crossed product does not depend on the choice of section $s$. Theorem \ref{theorem:packerraeburndecomposition} was proved in \cite{packer1989twisted} for ungraded complex twisted crossed products, but the generalisation to the real and/or graded cases still holds. We have required $N\subset \mathrm{ker}(c)$ to ensure that $c$ descends to the quotient group $G/N$, and that the automorphisms $\beta_p$ and cocycle $\nu(\cdot,\cdot)$ are even, independently of $s$. Then one checks that the standard grading on either side of \eqref{twisteddecompositionformula} agrees with the other.

\subsection{Finitely-generated projective modules in equivariant $K$-theory}
We make a short digression to define the notion of a f.g.p.\ $(G,\mathcal{A},\alpha)$-module $W$, following\footnote{The author worked with right modules, but we prefer to use left modules.} Chapter 11.2 of \cite{blackadar1998k}. Here, $\mathcal{A}$ is a (ungraded) unital $C^*$-algebra, and $\alpha$ is a Borel homomorphism (hence continuous) from a compact group $G$ to $\mathrm{Aut}_\mathbb{F}(\mathcal{A})$. Such modules are needed to define the $G$-equivariant $K$-theory of $\mathcal{A}$, and when $\mathcal{A}$ is commutative, they provide the link to the corresponding topological equivariant $K$-theory. 

We write $\mathscr{L}(W)$ for the set of bounded linear operators\footnote{There are a number of equivalent Banach norms determining the same topology on $W$.} on $W$, $\mathscr{G}\mathscr{L}(W)$ for the subgroup of invertible operators, and $\mathscr{B}(W)$ for the subalgebra of module maps. A finitely-generated projective $(G,\mathcal{A},\alpha)$-module is a f.g.p\ $\mathcal{A}$-module $W$, along with a strongly continuous homomorphism $\theta:G\rightarrow \mathscr{G}\mathscr{L}(W)$, such that
\begin{equation}
	\theta_x(aw)=a^x(\theta_x w),\qquad x\in G,\,a\in\mathcal{A},\,w\in W.\nonumber
\end{equation}
The equivariant $K$-theory group $K^{G}_0(\mathcal{A})$ is defined to be the Grothendieck group of the monoid (under the direct sum) of equivalence classes of f.g.p.\ $(G,\mathcal{A},\alpha)$-modules. If $\mathcal{A}$ is non-unital, $K^G_0(\mathcal{A})$ is defined to be the kernel of the the induced map $p_*:K^G_0(\mathcal{A}^+)\rightarrow K^G_0(\mathbb{F})$, where $\mathcal{A}^+$ has the induced action from $\alpha$, $\mathbb{F}$ has the trivial action, and $p$ is the equivariant projection $\mathcal{A}^+\rightarrow\mathbb{F}$. The Green--Julg theorem says that an equivariant $K_0$-group is isomorphic to the ordinary $K_0$-group of the crossed product,
\begin{equation}
	K^G_0(\mathcal{A})\cong K_0(\mathcal{A}\rtimes_{(\alpha,1)}G).\nonumber
\end{equation}
We define graded f.g.p.\ $(G,c,\mathcal{A},\alpha)$-modules $W$ in a similar way, noting that the linear operators on $W$ acquire a natural grading. The operators $\theta_x$ are required to be odd or even according to $c(x)$.

\begin{example}[Decomposing group $C^*$-algebras over an abelian normal subgroup]\label{subsection:decompositionordinarygroupalgebras}
Consider the case where $\alpha=c=\sigma\equiv 1$, $N$ is a discrete abelian group, $G/N$ is compact, and $G$ is a topological semidirect product $G=N\rtimes G/N$. The Fourier transform gives an isomorphism $\mathbb{C}\rtimes N\cong C(\hat{N})$,
where $\hat{N}$ is the Pontryagin dual of $N$. In physical applications, $N$ is a lattice of translations, $G/N$ is a compact point group which may include internal symmetries such as spin $\mathrm{SU}(2)$, and $\hat{N}$ is the Brillouin torus, over which the Bloch bands of solid-state physics reside. Since $c\equiv 1$, $\mathcal{A}=\mathbb{C}\rtimes G$ is purely even and $K_0(\mathcal{A})\cong\mathcal{S}\mathcal{R}(\mathcal{A})\cong\mathbf{K}_0(\mathcal{A})$.

The standard homomorphic section $s:p\mapsto (e,p)\in N\rtimes G/N$ satisfies $(e,p)(n,e)(e,p^{-1})=(p\cdot n,e)$, where $n\mapsto p\cdot n$ is the defining automorphic action of $p\in G/N$ on $N$. Using \eqref{modifiedembedding} and \eqref{decomposedtwistingpair1}, we find that the automorphisms $\beta_p$ act on the canonical generators $\delta_n\in\mathbb{C}\rtimes N$ by $\beta_p(\delta_n)=\delta_{p\cdot n}$. In terms of functions $f:N\rightarrow \mathbb{C}$, this is $\beta_p(f)(n)=f(p^{-1}\cdot n)$. Under the Fourier transform $f\mapsto\hat{f}\in C(\hat{N})$, the automorphism $\beta_p$ becomes $\hat{\beta}_p$, defined by $\hat{\beta}_p(\hat{f})(\chi)\coloneqq\hat{f}(p^{-1}\cdot\chi)$, where $(p\cdot\chi)(n)\coloneqq\chi(p^{-1}\cdot n)$ is the dual $G/N$-action on $\hat{N}$. Also, \eqref{decomposedtwistingpair2} gives $\nu\equiv1$, so we may rewrite $\mathbb{C}\rtimes G\equiv\mathbb{C}\rtimes (N\rtimes G/N)$ as $C(\hat{N})\rtimes_{(\hat{\beta},1)}G/N$. By the Green--Julg theorem, $K_0(\mathcal{A})\cong K_0^{G/N}(C(\hat{N}))\cong K^0_{G/N}(\hat{N})$. A f.g.p.\ $\mathcal{A}$-module corresponds to the sections of some finite-rank $G/N$-equivariant vector bundle over $\hat{N}$. 

If $G$ is furthermore a direct product, the $G/N$ action $\hat{\beta}$ on the base space $\hat{N}$ is trivial, and each fibre becomes a finite-dimensional representation space for $G/N$. On the other hand, things get complicated when $G$ is not a semidirect product of $N$ and $G/N$, as is the case when $G$ is a non-symmorphic space group. There is a non-trivial central cocycle $\nu(p,q)=s_ps_qs_{pq}^{-1}\in N$ according to \eqref{decomposedtwistingpair2}. The automorphism $\hat{\beta}_p$ is the dual of conjugation by $s_p$, i.e.,\ $\hat{\beta}_p(\hat{f})(\chi)=\hat{f}(p^{-1}\cdot\chi)$, where $(p\cdot \chi)(n)\coloneqq\chi(s_p^{-1}ns_p)$. Note that $\hat{N}$ remains a $G/N$-space, since
\begin{eqnarray}
	(p\cdot(q\cdot \chi))(n)	&=&	\chi(s_q^{-1}s_p^{-1}ns_ps_q)\nonumber\\
																&=&	\chi(s_{pq}^{-1}\nu(p,q)^{-1}n\nu(p,q)s_{pq})\nonumber\\
																&=& \chi(s_{pq}^{-1}ns_{pq})	=((pq)\cdot\chi)(n)	\nonumber
\end{eqnarray}
The isomorphism $\mathbb{C}\rtimes G\cong C(\hat{N})\rtimes_{(\hat{\beta},\nu)}G/N$ suggests the interpretation of a ``f.g.p.\ $(G/N,C(\hat{N}),\hat{\beta},\nu)$-module'' as the sections of a vector bundle $E$ over $\hat{N}$, equipped with a ``$\nu$-twisted equivariant $G/N$-action'' on the fibres. There is a family of \emph{projective} representations of $G/N$, with $p\in G/N$ mapping the fibre over an intial basepoint $\chi$ linearly to the fibre over a final basepoint $p\cdot \chi$.

When $N$ is not discrete, $\hat{N}$ is non-compact. Topological equivariant $K$-theory groups must be interpreted using vector bundles trivialised outside a compact subspace of $\hat{N}$, i.e., $K$-theory with ``compact supports''. Such a situation arises, for instance, when $N=\mathbb{R}^d$, which has a very different topological nature to $\mathbb{Z}^d$.

\end{example}

\begin{example}[Decomposing twisted group $C^*$-algebras over an abelian normal subgroup]\label{subsection:decomposingtwistedgroupalgebras}
If we allow $\sigma\not\equiv 1$ in Example \ref{subsection:decompositionordinarygroupalgebras}, then we may not have \mbox{$\mathbb{C}\rtimes_{(1,\sigma)}N\cong C_0(X)$} for any topological space $X$ since $\mathbb{C}\rtimes_{(1,\sigma)}N$ is generally a noncommutative algebra. This situation occurs when there is \emph{magnetic} translational symmetry instead of ordinary commuting translational symmetry, e.g.\ in the Integer Quantum Hall Effect \cite{bellissard1994noncommutative}. Our noncommutative approach bears fruit here, since it still makes sense to study the $K$-theory of \mbox{$\mathcal{A}=\mathbb{C}\rtimes_{(1,\sigma)}G\cong (\mathbb{C}\rtimes_{(1,\sigma)}N)\rtimes_{(\beta,\nu)}G/N$}. If we wish to, we can interpret a f.g.p.\ $\mathcal{A}$-module as the space of ``sections'' of some $\nu$-twisted $G/N$-equivariant ``bundle'' over the noncommutative space corresponding to $\mathbb{C}\rtimes_{(1,\sigma)}N$.

\end{example}

\begin{example}[A Clifford algebra factorises in the crossed product $C^*$-algebra]\label{subsection:decompositionCliffordbundles}
We now consider the symmetry data given by $(G,c,\phi,\sigma)$. We assume that $G=(N\rtimes Q)\times A$, where $A\cong\mathrm{Im}(\phi,c)\subset \ztwo^2$, $N\subset\mathrm{ker}(\phi,c)=N\rtimes Q$ is abelian, and $G/N$ is compact. We will think of $A$ as one of the subgroups of the $CT$-group, by identifying $\ztwo^2$ with $\{1,T,C,S\}$ as in Section \ref{subsection:CTgrouptenfoldway}. Where present, we denote the lifts of $C,T,S$ to $G$ by the same symbols. We also assume that $\sigma(nq,\cdot)=1=\sigma(\cdot,nq)$ for all $nq$ in $N\rtimes Q$, then $\sigma$ is simply specified by its restriction to $A$. This is the setting (with $Q$ usually assumed trivial) that is often considered in the literature when studying band structures with time-reversal and/or charge conjugation symmetry (along with its translation symmetries).

Let $\mathcal{A}=\mathbb{C}\rtimes_{(\alpha,\sigma)}G$, with $\alpha$ determined by $\phi$ as usual. We denote the images, under $j_\mathcal{A}$, of $C,T,S$ in $\mathcal{M}\mathcal{A}$ by $\mathsf{C},\mathsf{T},\mathsf{S}$. There is a non-trivial grading on $\mathcal{A}$ if $c\not\equiv 1$. The commutation relations amongst $\mathsf{C},\mathsf{T},\mathsf{S}$ are the same as those amongst the representatives of $C,T,S$ in a graded PUA-rep of $(G,c,\phi,\sigma)$. As in Section \ref{subsection:CTgrouptenfoldway}, we can choose $\sigma$ such that $\mathsf{C}\mathsf{T}=\mathsf{T}\mathsf{C}=\mathsf{S}$ as elements of the crossed product, then $\sigma$ is simply specified by $\mathsf{T}^2=\pm 1, \mathsf{C}^2=\pm 1$, while $\mathsf{S}^2=+1$ can be assumed. There are ten possibilities for $(A,\sigma)$, each with a corresponding Clifford algebra, as listed in Table \ref{table:ctcliffordalgebras}.

We first decompose $\mathcal{A}$ with respect to the subgroup $N\rtimes Q$, noting that $\nu$ reduces to $\sigma$ in \eqref{decomposedtwistingpair2}. We obtain
\begin{eqnarray}
	\mathcal{A}\cong(\mathbb{C}\rtimes (N\rtimes Q))\rtimes_{(\gamma,\sigma)}A &\cong &\left(\mathbb{C}\rtimes (N\rtimes Q)\right)\rtimes_{(\gamma,\sigma)}A\nonumber\\
	&\cong &\left(C_0(\hat{N})\rtimes_{(\hat{\beta},1)}Q\right)\rtimes_{(\gamma,\sigma)}A,\nonumber
\end{eqnarray}
where $\hat{\beta}$ is determined as in Examples \ref{subsection:decompositionordinarygroupalgebras} and \ref{subsection:decomposingtwistedgroupalgebras}, and $\gamma_r, r\in A$ are some automorphisms of $\mathbb{C}\rtimes (N\rtimes Q)$. Since $A$ appears as a direct product factor in $G$, it acts only on $\mathbb{C}$ in $\mathbb{C}\rtimes (N\rtimes Q)$, so $\gamma_r$ effects complex conjugation if $\phi(r)=-1$ and does nothing otherwise.

In the complex case, i.e.\ $A\subset\{1,S\}$, we have
\begin{equation}
	\mathcal{A}	\cong \left(C_0(\hat{N})\rtimes_{(\hat{\beta},1)}Q\right)\hat{\otimes}_\mathbb{C}\left(\mathbb{C}\rtimes_{(1,\sigma)}A\right)\cong \left(C_0(\hat{N})\rtimes_{(\hat{\beta},1)}Q\right)\hat{\otimes}\,\mathbb{C}l_n,\label{complexcliffordsimplification}
\end{equation}
where the complex Clifford algebra is $\mathbb{C}l_1$ if $A=\{1,S\}$ and $\mathbb{C}l_0$ if $A=\{1\}$. For discrete $N$, we can understand a graded f.g.p.\ $(Q,C(\hat{N}),\hat{\beta})$-module as in Section \ref{subsection:decompositionordinarygroupalgebras}, in terms of a graded $Q$-equivariant complex vector bundle over $\hat{N}$, which is just the direct sum of two ungraded such bundles. When $n=1$, there is an additional graded action of $\mathbb{C}l_1$ on the fibres which commutes with the $Q$-action. 

In the real case, i.e., either of $C$ and $T$ is present, we first write $\mathbb{C}\rtimes (N\rtimes Q)=(\mathbb{R}\rtimes (N\rtimes Q))\otimes_\mathbb{R}\mathbb{C}$. Then we obtain
\begin{eqnarray}
	\mathcal{A}\cong \left(\left(\mathbb{R}\rtimes (N\rtimes Q)\right)\otimes_\mathbb{R}\mathbb{C}\right)\rtimes_{(\gamma,\sigma)}A
	&=&\left(\mathbb{R}\rtimes (N\rtimes Q)\right)\otimes_\mathbb{R}\left(\mathbb{C}\rtimes_{(\alpha,\sigma)}A\right ),
	\nonumber\\
	&\cong &\left(\mathbb{R}\rtimes (N\rtimes Q)\right)\hat{\otimes}\,Cl_{r,s},\label{realcliffordsimplification}
\end{eqnarray}
where the Clifford algebra $Cl_{r,s}$ is determined by $(A,\sigma)$, according to Section \ref{subsection:CTgrouptenfoldway}. Actually, when $A=\{1,T\}$, the factor $\mathbb{C}\rtimes_{(\alpha,\sigma)}A$ is purely even, so we have actually made a modification (see Remark \ref{gradedversusungradedclifford}) in replacing it by a graded Clifford algebra $Cl_{r,s}$. This detail is actually quite important, see Section \ref{bandinsulatorsktheory}.

It is possible to formulate things in terms of Real bundles in the sense of Atiyah \cite{atiyah1966k} with Clifford modules as fibres \cite{freed2013twisted,thiang2014thesis}. However, doing this directly by Fourier transforming $\mathbb{C}\rtimes N$ to $C(\hat{N})$ requires a fairly complicated and opaque auxiliary construction in the real case. The underlying reason is because the real $C^*$-algebra $\mathbb{R}\rtimes N$ does not simply translate into $C_0(\hat{N},\mathbb{R})$ under the Fourier transform. Instead, we have to consider $\hat{N}$ as a Real space with involution given by the map taking $\chi$ to the complex conjugate character $\bar{\chi}$. The algebra $\mathbb{R}\rtimes N$ is the real subalgebra of $\mathbb{C}\rtimes N$ which is fixed under complex conjugation. Upon taking the Fourier transform $\mathbb{C}\rtimes N\cong C_0(\hat{N})$, complex conjugation turns into the antilinear involution $\bar{\hat{f}}(\chi)\coloneqq\overline{\hat{f}(\bar{\chi})}$, and the correct isomorphism is \begin{equation}
	\mathbb{R}\rtimes N\cong C_0(\im\hat{N})\coloneqq\left\{\hat{f}\in C_0(\hat{N})\,:\,\overline{\hat{f}(\chi)}=\hat{f}(\bar{\chi})\right\}.\label{correctrealfourierisomorphism}
\end{equation}
If we had performed a Fourier transform in \eqref{realcliffordsimplification}, a Clifford algebra cannot be nicely factorised, and the analysis becomes unnecessarily obscured.

Equations \eqref{complexcliffordsimplification} and \eqref{realcliffordsimplification} express $\mathcal{A}$ as the tensor product of a purely even algebra with a graded Clifford algebra. According to Section \ref{section:karoubidifferenceconstruction}, this effects a degree shift in $K$-theory, so the difference-group $\mathbf{K}_0(\mathcal{A})$ is easy to compute. In the complex case, we have
\begin{equation}
\mathbf{K}_0(\mathcal{A})\cong K_n(C_0(\hat{N})\rtimes_{(\hat{\beta},1)}Q)\cong K^{-n}_Q(\hat{N}),\label{complexdifferencegroup}
\end{equation}
where the last expression is an equivariant topological $K$-theory group. In the real case, we have
\begin{equation}
	\mathbf{K}_0(\mathcal{A})\cong K_{s-r}(C_0(\im\hat{N})\rtimes_{(\hat{\beta},1)} Q).\nonumber
\end{equation}
We can use the correspondence between $C_0(\im\hat{N})$-modules and sections of a Real bundle over $\hat{N}$ (with antilinear involution $\overline{(\cdot)}$ lifting $\chi\mapsto\bar{\chi}$) which are fixed under the induced involution $\overline{s(\chi)}=s(\bar{\chi})$. The result is
\begin{equation}
	\mathbf{K}_0(\mathcal{A})\cong KR^{r-s\,\mathrm{(mod\,8)}}_Q(\hat{N}).\label{realdifferencegroup}
\end{equation}

\begin{remark}
In this example, we have assumed that $\mathrm{ker}(\phi,c)=N\rtimes Q$ in order to make a connection to topological $K$-theory. This splitting assumption is not necessary in our noncommutative framework, and is also not generic physically. There are other elaborate versions of Real $K$-theory, and even Quaternionic $K$-theory which may be useful for a geometrical picture, but such a picture can quickly become unwieldy in general.
\end{remark}

\end{example}

\section{Applications to topological band insulators}\label{section:bandinsulators}
\subsection{Band insulators and $K$-theory}\label{bandinsulatorsktheory}
The computations in Examples \ref{subsection:decompositionordinarygroupalgebras}-\ref{subsection:decompositionCliffordbundles} hold for the special case  of topological band insulators, in which $N=\mathbb{Z}^d$. Thus, the $\mathbb{Z}^d$ lattice translations are realised non-projectively, complex-linearly, and do not have any charge-reversing or time-reversing properties. The Brillouin torus is nothing but $\hat{N}=\mathbb{T}^d$. When we assumed the splitting $G=\mathbb{Z}^d\rtimes P$ (where $P=G/N$) as well as $\sigma\equiv 1$, we obtained $\mathcal{A} =\mathbb{C}\rtimes_{(\alpha,1)}G\cong C(\mathbb{T}^d)\rtimes_{(\hat{\beta},1)}P$. 

Underlying a graded f.g.p.\ $\mathcal{A}$-module is a graded f.g.p.\ $C(\mathbb{T}^d)$-module, and thus a graded complex vector bundle $E\rightarrow\mathbb{T}^d$ which we call a graded \emph{Bloch bundle}. The homomorphisms $(\phi,c)$ descend to $P$, telling us whether $p\in P$ acts complex linearly/antilinearly or preserves/reverses the ``particle-hole'' distinction. Thus $p$ acts real-linearly between fibres, moving points on the base space in a manner which depends on the action of $p$ on $\mathbb{Z}^d$ and whether $\phi(p)=\pm 1$. In a graded Bloch bundle, the positively-graded subbundle can be interpreted as the conduction band lying above the Fermi level $\mathcal{E}_F$ (which is taken to be $0$), while the negatively-graded subbundle is the valence band, lying below $\mathcal{E}_F$. Note that this description does not require any distinguished involutary time-reversal, charge-conjugation, or chiral symmetry element in $P$.

In this Bloch bundle picture, a trivial graded $\mathcal{A}$-module in the sense of Definition \ref{definition:trivialgradedmodule} is one which admits an odd complex bundle automorphism $\mathcal{I}$ satisfying $\mathcal{I}^2=1$, and $\mathcal{I}p=c(p)p\mathcal{I}$ for all $p\in P$. Such a Bloch bundle is called a \emph{topologically trivial $(G,c,\phi)$-compatible band insulator}, and is a commutative instance of Definition \ref{definition:trivialgradedmodule}. The super-representation group of $\mathcal{A}$ thus classifies the $(G,c,\phi)$-compatible band insulators, modulo the topologically trivial ones.

In general, we can allow $\sigma\not\equiv 1$, but require further splitting assumptions. In Section \ref{subsection:decompositionCliffordbundles}, we found that a significant simplification occurs when $G$ has the form $G=(\mathbb{Z}^d\rtimes Q)\times A$ with $A\cong \mathrm{Im}(\phi,c)\subset\{1,T,C,S\}$, and $\sigma$ is only non-trivial between elements of $A$. The data of $(A,\sigma)$ is associated with one of ten possible Clifford algebras, which factorises in $\mathcal{A}=\mathbb{C}\rtimes_{(\alpha,\sigma)} G$. In the two complex cases, the graded Bloch bundle corresponding to a graded f.g.p.\ $\mathcal{A}$-module is a graded $Q$-equivariant complex bundle $E\rightarrow\mathbb{T}^d$. If $A=\{1,S\}$, there is an additional \emph{commuting} graded $\mathbb{C}l_1$-action on each fibre, generated by the complex-linear, odd, and involutary map $\mathsf{I}_\mathsf{S}$ representing $\mathsf{S}$. A topologically-trivial $(G,c,\phi,\sigma)$-compatible band insulator (in a complex class) is one which admits an odd $Q$-equivariant complex bundle automorphism $\mathcal{I}$ satisfying $\mathcal{I}^2=+1$, and $\mathcal{I}\mathsf{I}_\mathsf{S}=-\mathsf{I}_\mathsf{S}\mathcal{I}$ if $S$ is present.

For the remaining eight real cases, the elements $T,C$ of $A$ (where present) act through the antilinear bundle maps $\mathsf{I}_\mathsf{T},\mathsf{I}_\mathsf{C}$ representing $\mathsf{T},\mathsf{C}$, which are even and odd respectively. The squares of $\mathsf{I}_\mathsf{T}$ and $\mathsf{I}_\mathsf{C}$ are $\pm 1$ according to $\sigma(T,T)$ and $\sigma(C,C)$. Unlike $\mathsf{I}_\mathsf{S}$, the bundle maps $\mathsf{I}_\mathsf{T},\mathsf{I}_\mathsf{C}$ do \emph{not} commute with the bundle projection, but instead take the fibre over $\chi\in\mathbb{T}^d$ to the fibre over $\bar{\chi}$. The map $\mathsf{I}_\mathsf{T}$ is the standard time-reversal operation on the Bloch bundle $E$, and both cases $\mathsf{I}_\mathsf{T}^2=\pm 1$ are considered in the literature. Similarly, $\mathsf{I}_\mathsf{C}$ is the standard particle-hole conjugation operation on $E$. A topologically-trivial $(G,c,\phi,\sigma)$-compatible band insulator is one which admits an odd $Q$-equivariant complex bundle automorphism, satisfying $\mathcal{I}^2=1$, $\mathcal{I}\mathsf{I}_\mathsf{T}=\mathsf{I}_\mathsf{T}\mathcal{I}$, and $\mathcal{I}\mathsf{I}_\mathsf{C}=-\mathsf{I}_\mathsf{C}\mathcal{I}$. Since $\mathsf{I}_\mathsf{T},\mathsf{I}_\mathsf{C}$ do \emph{not} commute with $C(\mathbb{T}^d)$, the graded Bloch bundle is not a bundle of Clifford modules. 

The difference-groups of $(G,c,\phi,\sigma)$-compatible band insulators are generally distinct from their super-representation groups. The former can be read off from \eqref{complexdifferencegroup} and \eqref{realdifferencegroup},
\begin{equation}
	\mathbf{K}_0(\mathcal{A})=\begin{cases}K_Q^{-n}(\mathbb{T}^d)\quad &\text{complex case},\\
	KR_Q^{r-s\,\mathrm{(mod\,8)}}(\mathbb{T}^d) \quad &\text{real case}, \nonumber\end{cases}
\end{equation}
with $n,r,s$ determined by Table \ref{table:ctcliffordalgebras}.

\subsection{The three special purely even cases}\label{subsection:threespecialcases}
Although $\mathcal{S}\mathcal{R}(\mathcal{A})$ and $\mathbf{K}_0(\mathcal{A})$ do not generally coincide for graded $\mathcal{A}$, they do when $\mathcal{A}$ is purely even. Important examples are when $G$ is of the form $(N\rtimes Q)\times A$ with $A\subset\{1,T\}$, and only $\sigma(T,T)$ may be non-trivial. In each of these cases, $\mathbb{C}\rtimes_{(\alpha,\sigma)}A$ is purely even: it is either the complex algebra $\mathbb{C}$, the real algebra $M_2(\mathbb{R})$, or the real algebra $\mathbb{H}$, with the latter two real algebras generated by $\im$ and $\mathsf{T}$. This means that \eqref{complexcliffordsimplification} and \eqref{realcliffordsimplification} should give 
\begin{equation}
\mathbb{C}\rtimes_{(\alpha,\sigma)}G=\begin{cases}
\mathbb{C}\rtimes (N\rtimes Q) \qquad & A=\{1\} ,\\ \mathbb{R}\rtimes (N\rtimes Q)\otimes_\mathbb{R} M_2(\mathbb{R}) \quad\quad & A=\{1,T\}, \mathsf{T}^2=+1,\\ \left(\mathbb{R}\rtimes (N\rtimes Q)\right)\otimes_\mathbb{R}\mathbb{H} \quad\quad & A=\{1,T\}, \mathsf{T}^2=-1.
\end{cases}\nonumber
\end{equation}
Recall that in the third case, we replaced $\mathbb{H}$ by the graded Clifford algebra $Cl_{4,0}\cong Cl_{0,4}$ in \eqref{realcliffordsimplification}, which we used to arrive at \eqref{realdifferencegroup}. We could also use Theorem III.4.12 of \cite{karoubi1978k}, which says that $K'^{0,0}(\mathcal{B}\otimes_\mathbb{R}\mathbb{H})\cong K'^{0,4}(\mathcal{B})$ for any ungraded $\mathcal{B}$. Taking $\mathcal{B}=\mathbb{R}\rtimes (N\rtimes Q)$, we obtain $\mathbf{K}_0(\mathcal{B}\otimes_\mathbb{R}\mathbb{H})\equiv K'^{0,0}(\mathcal{B}\otimes_\mathbb{R}\mathbb{H})\cong K'^{0,4}(\mathcal{B})\equiv \mathbf{K}_0(\mathcal{B}\hat{\otimes}Cl_{0,4})$, which justifies our replacement. The super-representation groups for the three cases where $\mathcal{A}=\mathbb{C}\rtimes_{(\alpha,\sigma)}G$ is purely even, are
\begin{equation}
\mathcal{S}\mathcal{R}(\mathcal{A})\cong\mathbf{K}_0(\mathcal{A})=\begin{cases}
K_Q^0(\hat{N}) \quad\quad & A=\{1\} ,\\ KR_Q^0(\hat{N}) \quad\quad & A=\{1,T\}, \mathsf{T}^2=+1,\\ KR_Q^{-4}(\hat{N}) \quad\quad & A=\{1,T\}, \mathsf{T}^2=-1,
\end{cases}\nonumber
\end{equation}
which should be compared with Corollary 10.28 of \cite{freed2013twisted}. Specialising to $N=\mathbb{Z}^d$ yields the super-representation/difference-groups of $(G,\phi,\sigma)$-compatible band insulators. These three classes of insulators are usually labelled by $A, AI$ and $AII$ repsectively.

We remark that the somewhat awkward treatment of the third case ($\mathsf{T}^2=-1$), which  resulted in a $KR^{-4}$ group, can be modified to resemble the first two cases more closely. In that case, we identify a f.g.p.\ $\mathcal{A}$-module with the sections of a $Q$-equivariant ``Quaternionic'' vector bundle over $\hat{N}$. Like Real vector bundles, ``Quaternionic'' vector bundles are complex vector bundles over a Real space $(X,\varsigma)$, but with a lift of $\varsigma$ to an antilinear \emph{anti-involution} $\Theta$ such that $\Theta^2=-1$. Such bundles were considered in a classification scheme for Type $AII$ topological insulators in \cite{de2014classification2}, in which detailed definitions and references for ``Quaternionic'' bundles can be found. The corresponding topological $K$-theory is called $KQ$-theory, and there is a useful isomorphism $KR^{-4}(X,\varsigma)\cong KQ^0(X,\varsigma)$ (see \cite{dupont1969symplectic} and the Appendix of \cite{de2014classification2}). This establishes a $KR^{-4}$ group as a Grothendieck group $KQ^0$ of ``Quaternionic'' bundles. Therefore, $\mathbf{K}_0(\mathcal{A})$ may be interpreted as a Grothendieck group of vector bundles in the three cases where $\mathcal{A}=\mathbb{C}\rtimes_{(\alpha,\sigma)}G$ is purely even \emph{but not in the other seven}.

\subsection{Finite versus infinite rank bundles}\label{subsection:finiteversusinfinite}
In many realistic Hamiltonians compatible with a subgroup $N=\mathbb{Z}^d\subset G$ of translational symmetries of a lattice, the Bloch Hamiltonians are unbounded operators on a bundle of infinite-dimensional Hilbert spaces over the Brillouin torus $\hat{N}=\mathbb{T}^d$. In our application of $K$-theory, we have confined ourselves to the finite-rank situation. We can motivate this physically by imposing an energy cut-off, or by assuming that there is a finite-rank conduction band which we restrict attention to. As such, $\Gamma$ refers not to the flattened version of the full Hamiltonian, but to the flattened version of its restriction to the relevant low energy modes. Likewise, the symmetries act only on this restricted (invariant) subspace. Also, in certain discretised models, the Bloch Hamiltonians do have finite-rank, so the $K$-theory classification makes sense for the full Hamiltonian. It is actually very important to distinguish the finite and infinite rank cases.

In the Hilbert $C^*$-module description of Bloch theory \cite{gruber2001noncommutative}, the Bloch Hamiltonians act on a continuous field of infinite-dimensional separable Hilbert spaces over $\mathbb{T}^d$, whose sections form a countably-generated Hilbert $C(\mathbb{T}^d)$-module. In \cite{freed2013twisted}, the authors considered graded bundles $E^+\oplus E^-$ over $\mathbb{T}^d$, such that $E^+$ is infinite-dimensional, while $E^-$ is a finite-rank bundle. They called such bundles ``Type $I$'' insulators, while graded bundles with both $E^+$ and $E^-$ having finite rank were called ``Type $F$'' insulators. We can make this distinction in more general noncommutative terms.
\begin{definition}[Type $I$ and Type $F$ insulators]
Let $\mathcal{A}$ be a graded unital $C^*$-algebra. A \emph{Type $I$ insulator} is a graded $\mathcal{A}$-module $E=E_0\oplus E_1\eqqcolon E^+\oplus E^-$ such that $E^-$ is an ungraded f.g.p.\ $\mathcal{A}$-module and $E^+$ is a countably-generated (but not finitely-generated) Hilbert $\mathcal{A}$-module. A \emph{Type $F$ insulator} is a graded f.g.p.\ $\mathcal{A}$-module.
\end{definition}
Suppose $\mathcal{A}$ arises from the symmetry data $(G,c,\mathcal{B},\phi,\sigma)$. For Type $I$ insulators, there is no possibility of an invertible odd operator taking $E^+$ to the f.g.p.\ submodule $E^-$, so $c\equiv 1$ and $\mathcal{A}$ is purely even. We define $\mathcal{G}\mathcal{V}^I(\mathcal{A})$ to be the commutative monoid, under the direct sum, of equivalence classes of Type $I$ insulators. Note that we have to formally introduce the zero module as a zero element, since it is not actually in $\mathcal{G}\mathcal{V}^I(\mathcal{A})$. This monoid is really the direct sum of the monoid $\mathcal{V}^+(\mathcal{A})$ of ungraded countably-generated Hilbert $\mathcal{A}$-modules $E^+$, and the monoid $\mathcal{V}^-(\mathcal{A})\coloneqq\mathcal{V}(\mathcal{A})$ of ungraded f.g.p.\ $\mathcal{A}$-modules. The notion of triviality in Definition \ref{definition:trivialgradedmodule} is vacuous for Type $I$ insulators, and we cannot construct its ``super-representation group''. Instead, an abelian group is obtained via the Grothendieck completion $\mathcal{G}\mathcal{R}^I(\mathcal{A})$ of $\mathcal{G}\mathcal{V}^I(\mathcal{A})=\mathcal{V}^+(\mathcal{A})\oplus \mathcal{V}^-(\mathcal{A})$. We can perform an ``Eilenberg swindle'' on the $\mathcal{V}^+$ part of $\mathcal{G}\mathcal{V}^I(\mathcal{A})$:
\begin{proposition}
The Grothendieck completion of $\mathcal{V}^+(\mathcal{A})$ is the trivial group.
\end{proposition}
\begin{proof}
There is a ``standard'' countably-generated Hilbert $\mathcal{A}$-module $\hs_\mathcal{A}$ (see 15.1.7 of \cite{wegge1993k}), defined by
\begin{equation} \hs_\mathcal{A}\coloneqq\left\{(a_k)\in\prod_{k=1}^\infty A\,:\,\sum_{k=1}^\infty a_k^*a_k\,\, \mathrm{converges\,\,in\,\,norm\,\,in}\,\,\mathcal{A}\right\},\nonumber
\end{equation}
with component-wise operations and $\mathcal{A}$-valued inner product $\langle (a_k),(b_k)\rangle\coloneqq \sum_k a_k^*b_k$. The Kasparov stabilisation theorem (see 15.4.6 of \cite{wegge1993k} for a proof and original references), says that $\hs_\mathcal{A}$ is absorbing in the sense that every countably-generated Hilbert $\mathcal{A}$-module $E^+$ satisfies $E^+\oplus\hs_\mathcal{A}\cong\hs_\mathcal{A}$. Then any virtual module $[E^+\ominus E'^+]$ in the Grothendieck completion of $\mathcal{V}^+(\mathcal{A})$ satisfies $[E^+\ominus E'^+]=[(E^+\oplus\hs_\mathcal{A}) \ominus (E'^+\oplus\hs_\mathcal{A})]=[\hs_\mathcal{A}\ominus \hs_\mathcal{A}]=0$.\qed
\end{proof}
\begin{corollary}\label{corollary:grothendieckinfinite}
$\mathcal{G}\mathcal{R}^I(\mathcal{A})$ is isomorphic to the Grothendieck completion of $\mathcal{V}^-(\mathcal{A})=\mathcal{V}(\mathcal{A})$, i.e.,\ $\mathcal{G}\mathcal{R}^I(\mathcal{A})\cong K_0(\mathcal{A})$.
\end{corollary}
Corollary \ref{corollary:grothendieckinfinite} provides yet another interpretation of $K_0(\mathcal{A})$ (with $\mathcal{A}$ necessarily purely even). Namely, a virtual class $[E^-\ominus 0]\in K_0(\mathcal{A})\cong \mathcal{G}\mathcal{R}^I(\mathcal{A})$ can be represented by a gapped Type $I$ insulator whose ``bundle of negative eigenstates'' is $E^-$. Therefore, $\mathcal{G}\mathcal{R}^I(\mathcal{A})$ retains only the information of $E^-$. More generally, $[E^-\ominus E'^-]\in K_0(\mathcal{A})$ is represented by the \emph{formal difference} of (the negative eigenstates of) two Type $I$ insulators. This is actually more familiar than it looks. The Hall conductivity in the Integer Quantum Hall Effect is related to the $K$-theory element $[E^-\ominus 0]$ associated to a Type $I$ Landau Hamiltonian \cite{bellissard1994noncommutative}. Here, it is important to realise that there is generally no commutative Brillouin zone, so our definition of a ``noncommutative'' Type $I$ insulator has genuine physical relevance. Similarly, the Kane--Mele $\mathbb{Z}_2$ invariant for $T$-invariant insulators is determined by (formal differences of) the negative energy bundles of Type $I$ insulators, and was studied in \cite{freed2013twisted}.

On the other hand, the group $\mathcal{G}\mathcal{R}^I(\mathcal{A})$ does not make sense for Type $F$ insulators. For purely-even $\mathcal{A}^\mathrm{ev}$, the ordinary $K_0(\mathcal{A}^\mathrm{ev})$ group exists and may be interpreted as a virtual class of $\mathcal{A}^\mathrm{ev}$-modules, or as $\mathcal{S}\mathcal{R}(\mathcal{A}^\mathrm{ev})\cong\mathbf{K}_0(\mathcal{A}^\mathrm{ev})\cong K_0(\mathcal{A}^\mathrm{ev})$. However, the collection of Type $F$ insulators includes those with $\mathcal{A}$ having non-trivial grading, so only the super-representation/difference-group has general applicability.

In summary: \emph{the interpretation of ordinary $K$-theory groups either as $\mathcal{G}\mathcal{R}^I(\mathcal{A})$ in the Type $I$ case, or as virtual classes of ungraded $\mathcal{A}^\mathrm{ev}$-modules in the Type $F$ case, cannot be used in a unified picture which includes charge-conjugating symmetries ($c\not\equiv 1$)}; the super-representation group or difference-group is more appropriate.

\section{Periodic Table of difference-groups and dimension shifts}\label{section:dimensionshift}
\subsection{Zero-dimensional gapped phases}
In Section \ref{subsection:CTgrouptenfoldway}, we showed that the graded PUA-reps of each of the ten choices of $(A,\sigma)$ are in $1\textrm{--}1$ correspondence with the graded representations of an associated Clifford algebra according to Table \ref{table:ctcliffordalgebras}. The super-representation group of $\mathbb{C}\rtimes_{(\alpha,\sigma)}A$ and that of its corresponding Clifford algebra coincide, and are precisely one of the $K$-theory groups of a point. Therefore, we say that $0$-dimensional gapped topological phases compatible with the symmetries specified by $(A,\sigma)$, modulo topologically trivial phases, are classified by $K^{-n}(\star)\cong K_n(\mathbb{C})$ or $KO^{-n}(\star)\cong K_n(\mathbb{R})$, according to Table \ref{table:0dgappedphasestable}. The same $K$-theory groups also classify the differences of $(A,\sigma)$-compatible gapped Hamiltonians, due to the isomorphisms $\mathbf{K}_0(Cl_{r,s})\cong K_{s-r}(\mathbb{R})$ and $\mathbf{K}_0(\mathbb{C}l_n)\cong K_n(\mathbb{C})$.

\begin{table}
\begin{center}
\begin{tabular}{l | l | l |c c |  c  c}
	$n$	&	\parbox[c]{1.4cm}{Cartan label} & \parbox[c]{2.3cm}{Generators of $A$} & $\mathsf{C}^2$ & $\mathsf{T}^2$  & \parbox[c]{2.3cm}{$Cl_{0,n}$ or $\mathbb{C}l_n$}	&	\parbox[l]{1.7cm}{$K_n(\mathbb{R})$ or $K_n(\mathbb{C})$}\\
	\hline
	$0$	&	$AI$	&	$T$	    &	 	& $+1$	& $Cl_{0,0}$	& $\mathbb{Z}$ \\
	$1$	&	$BDI$	&	$C,T$ 	& $+1$  & $+1$	& $Cl_{0,1}$ 	& $\mathbb{Z}_2$ \\	
	$2$	&	$D$		&	$C$		& $+1$  &		& $Cl_{0,2}$ 	& $\mathbb{Z}_2$ \\	
	$3$	&	$DIII$	&	$C,T$ 	& $+1$  & $-1$	& $Cl_{0,3}$ 	& $0$ \\	
	$4$	&	$AII$	&	$T$		&    	& $-1$	& $Cl_{0,4}$ 	& $\mathbb{Z}$ \\	
	$5$	&	$CII$	&	$C,T$ 	& $-1$ 	& $-1$	& $Cl_{0,5}$ 	& $0$ \\	
	$6$	&	$C$		&	$C$		& $-1$ 	&		& $Cl_{0,6}$ 	& $0$ \\	
	$7$	&	$CI$	&	$C,T$ 	& $-1$ 	& $+1$	& $Cl_{0,7}$	& $0$ \\
	
	\hline\hline
	$0$	&	$A$		&	N/A		& \multicolumn{2}{c|}{N/A}   	& $\mathbb{C}l_0$ & $\mathbb{Z}$ \\
	$1$	&	$AIII$	&	$S$	& \multicolumn{2}{c|}{$\mathsf{S}^2=+1$}		& $\mathbb{C}l_1$ & $0$ \\
\end{tabular}
\caption{Classification of $0$-dimensional gapped topological phases and their difference-classes. The next-to-last column lists the Clifford algebra $Cl_{0,n}$ or $\mathbb{C}l_n$ in the graded Morita class of the algebra $\mathcal{A}$ associated with $(A,\sigma)$. The $K$-theory group in the last column is isomorphic to the super-representation group $\mathcal{S}\mathcal{R}(\mathcal{A}$), as well as the difference-group $\mathbf{K}_0(\mathcal{A})$. The second column lists the Cartan label of the symmetric space of quantum mechanical time evolutions with corresponding $C,T$ symmetries \cite{ryu2010topological,heinzner2005symmetry}.}
\label{table:0dgappedphasestable}
\end{center}
\end{table}

\subsection{Higher dimensional gapped phases}
In the literature, it has been suggested \cite{kitaev2009periodic,stone2011symmetries} that the $K$-theoretic classification of gapped topological phases in $d$ spatial dimensions is the same as that in $0$ dimensions, except for a shift in the $K$-theory index by $d$. We have located the appropriate $K$-theoretic object (the difference-group) for which such a phenomenon might be plausible. We will prove a robust version of this dimension shift using some powerful results from the $K$-theory of crossed product $C^*$-algebras.

\begin{theorem}[Packer--Raeburn stabilisation trick \cite{packer1989twisted}]\label{theorem:packerraeburnstabilisationtrick}
Let $(G,\mathcal{A},\alpha,\sigma)$ be a twisted $C^*$-dynamical system, and let $\mathcal{K}$ denote the compact operators on the Hilbert space $L^2(G)$. There is an isomorphism
\begin{equation}
(\mathcal{A}\rtimes_{(\alpha,\sigma)}G)\otimes\mathcal{K}\cong(\mathcal{A}\otimes\mathcal{K})\rtimes_{(\alpha',1)}G,\nonumber
\end{equation}
for some untwisted action $\alpha'$ of $G$ on $\mathcal{A}\otimes\mathcal{K}$.
\end{theorem}

\begin{corollary}[Dimension shifts]\label{corollary:dimension shifts}
Let $(\mathbb{R}^d,\mathcal{A},\alpha,\sigma)$ be a twisted $C^*$-dynamical system. Then \mbox{$K_n(\mathcal{A}\rtimes_{(\alpha,\sigma)}\mathbb{R}^d)\cong  K_{n-d}(\mathcal{A})$}.
\end{corollary}
\begin{proof}
Iterating Theorem \ref{theorem:packerraeburndecomposition} yields
\begin{eqnarray}
\mathcal{A}\rtimes_{(\alpha,\sigma)}\mathbb{R}^d &\cong &(\mathcal{A}\rtimes_{(\alpha_1,\sigma_1)}\mathbb{R})\rtimes_{(\alpha_2,\sigma_2)}\mathbb{R}^{d-1}\nonumber\\
&\vdots &\nonumber\\
&\cong &\left(\ldots\left((\mathcal{A}\rtimes_{(\alpha_1,\sigma_1)}\mathbb{R})\rtimes_{(\alpha_2,\sigma_2)}\mathbb{R}\right)\rtimes_{(\alpha_3,\sigma_3)}\ldots\right)\rtimes_{(\alpha_d,\sigma_d)}\mathbb{R},\nonumber
\end{eqnarray}
for some sequence of twisting pairs $(\alpha_i,\sigma_i),\,i=1,\ldots,d$. Then, Theorem \ref{theorem:packerraeburnstabilisationtrick} says that we can untwist each of the iterated crossed products, provided we stabilise the algebras. Since $K$-theory is invariant under stabilisation, we have
\begin{eqnarray}
	K_n(\mathcal{A}\rtimes_{(\alpha,\sigma)}\mathbb{R}^n)&\cong & K_n\left(\left(\ldots\left((\mathcal{A}\rtimes_{(\alpha'_1,1)}\mathbb{R})\rtimes_{(\alpha'_2,1)}\mathbb{R}
	\right)\rtimes_{(\alpha'_3,1)}\ldots\right)\rtimes_{(\alpha'_d,1)}\mathbb{R}\right)\nonumber\\
	&\cong & K_{n-d}(\mathcal{A}),\nonumber
\end{eqnarray}
so the effect on $K$-theory is to lower the index by $d$.\qed
\end{proof}

\begin{theorem}[Connes--Thom isomorphism \cite{connes1981analogue,kasparov1995k,schroder1993k}]\label{theorem:connesthom}
Let $(\mathbb{R},\mathcal{A},\alpha)$ be a $C^*$-dynamical system, with $\mathcal{A}$ a real or complex (ungraded) $C^*$-algebra. Then \mbox{$K_n(\mathcal{A}\rtimes_{(\alpha,1)}\mathbb{R})\cong K_{n-1}(\mathcal{A})$}.
\end{theorem}
Remarkably, Theorem \ref{theorem:connesthom} holds for any continuous $\alpha:\mathbb{R}\rightarrow \mathrm{Aut}_\mathbb{F}(\mathcal{A})$.

\begin{lemma}\label{lemma:pm1cocycle}
Let $\sigma$ be a generalised $\mathrm{U}(1)$-valued 2-cocycle for $(G,\phi)$ (i.e.\ $\sigma$ satisfies equation \eqref{generalised2cocycle}) and suppose there is an element $w\in G$ with $\phi(w)=-1$. Then the restriction of $\sigma$ to the centraliser $\mathcal{Z}_G(w)$ of $w$ in $G$ is equivalent to one which takes only $\ztwo$ values.
\end{lemma}
\begin{proof}
Let $\varsigma(x,y)\coloneqq\sigma(x,y)/\sigma(xyx^{-1},x)$, so that $\theta_x\theta_y\theta_x^{-1}=\varsigma(x,y)\theta_{xyx^{-1}}$ is an identity in any PUA-rep $\theta$ of $(G,\phi,\sigma)$, and let $y,z\in\mathcal{Z}_G(w)$. Then
\begin{eqnarray}
	\theta_w\theta_y\theta_w^{-1}\theta_w\theta_z\theta_w^{-1}&=& \varsigma(w,y)\theta_{wyw^{-1}}\varsigma(w,z)\theta_{wzw^{-1}}\nonumber\\
	&=& \varsigma(w,y)\varsigma(w,z)^y\theta_y\theta_z\nonumber\\
	&=& \varsigma(w,y)\varsigma(w,z)^y\sigma(y,z)\theta_{yz}.\label{pm1cocycleproofequation1}
\end{eqnarray}
The left-hand-side of \eqref{pm1cocycleproofequation1} can also be written as
\begin{equation}
\theta_w\theta_y\theta_z\theta_w^{-1}=\sigma(y,z)^w\theta_w\theta_{yz}\theta_w^{-1}=
\sigma(y,z)^{-1}\varsigma(w,yz)\theta_{yz},\nonumber
\end{equation}
so $\sigma(y,z)^{-2}=\varsigma(w,y)\varsigma(w,z)^y\varsigma(w,yz)^{-1}$ is the coboundary of the function $\lambda:y\mapsto\varsigma(w,y)$. The function $\lambda^{\frac{1}{2}}$ acts on $\sigma$ to give an equivalent 2-cocycle $\sigma_w$, i.e.\ $\sigma_w(y,z)\coloneqq \{\lambda(y)\lambda(z)^y\lambda(yz)^{-1}\}^{\frac{1}{2}}\sigma(y,z)$, which corresponds to the phase modification $\theta_y\mapsto\lambda^{\frac{1}{2}} \theta_y$ in terms of PUA-reps. Then
\begin{equation}
	\sigma_w(y,z)^2=\frac{\lambda(y)\lambda(z)^y}{\lambda(yz)}\sigma(y,z)^2=1,
\end{equation}
and the new 2-cocyle $\sigma_w$ is $\ztwo$-valued when restricted to $\mathcal{Z}_G(w)$.
\end{proof}

We can now state the main result of this paper. In Theorem \ref{theorem:generaldimensionshift}, $\mathcal{B}^\mathrm{ev}$ is a purely even complex $C^*$-algebra if $\phi\equiv 1$, or the complexification of a purely even real $C^*$-algebra $\mathcal{B}^\mathrm{ev}=\mathcal{B}^\mathrm{ev}_\mathbb{R}\otimes_\mathbb{R}\mathbb{C}$ if $\phi\not\equiv 1$. The automorphism $\alpha_x, x\in \tilde{G}$ is either complex conjugation or the identity according to $\phi(x)=\pm 1$.

\begin{theorem}[General periodicity and dimension shift theorem]\label{theorem:generaldimensionshift}
Let $\tilde{\mathcal{A}}$ be the graded twisted crossed product for the twisted dynamical system $(\tilde{G},c,\mathcal{B}^\mathrm{ev},\alpha,\sigma)$. Assume that $\tilde{G}=\tilde{G}_0\times A$ where $\tilde{G}_0=\mathrm{ker}(\phi,c)$, and that $\sigma$ is $\mathrm{U}(1)$-valued and trivial between elements of $\tilde{G}_0$ and $A$. Suppose $\tilde{G}_0$ is an extension of a group $G_0$ by $\mathbb{R}^d$. Let $\mathcal{A}^\mathrm{ev}_{(\mathbb{R})}$ be the (even) twisted crossed product associated with the subsystem $(G_0,\mathcal{B}^\mathrm{ev}_{(\mathbb{R})},1,\sigma)$. Then
\begin{equation}
	\mathbf{K}_0(\tilde{\mathcal{A}})\cong \left\{\begin{alignedat}{2}
	& K_{s-r-d}(\mathcal{A}^\mathrm{ev}_\mathbb{R}),&&\qquad \phi\equiv 1,\\ & K_{n-d}(\mathcal{A}^\mathrm{ev}),&&\qquad \phi\not\equiv 1,\end{alignedat}\right.\label{generaldimensionshift}
\end{equation}
where $(r,s)$ or $n$ is determined by the Clifford algebra associated with $(A,c,\phi,\sigma)$.
\end{theorem}

\begin{proof}
In the $\phi\not\equiv 1$ case, there is a $w\in A$ such that $\phi(w)=-1$. By Lemma \ref{lemma:pm1cocycle}, we may assume that $\sigma$ is $\ztwo$-valued on $\tilde{G}_0\subset\mathcal{Z}_{\tilde{G}}(w)$, so the imaginary unit $\im\in\mathcal{B}^\mathrm{ev}$ is completely decoupled from the sub-dynamical system $(\tilde{G}_0,\mathcal{B}^\mathrm{ev}_\mathbb{R},1,\sigma)$. Thus $\tilde{A}$ has the form
\begin{equation}
	\tilde{\mathcal{A}}\cong\left(\mathcal{A}^\mathrm{ev}_\mathbb{R}\rtimes_{(\beta,\nu)}\mathbb{R}^d\right)\hat{\otimes}\,Cl_{r,s}\nonumber
\end{equation}
for some $Cl_{r,s}$ determined by $(A,c,\phi,\sigma)$ as in Section \ref{section:cliffordtenfoldway}, and some twisting pair $(\beta,\nu)$ arising from the $\mathbb{R}^d$ extension of $G_0$ as in Theorem \ref{theorem:packerraeburndecomposition}. Similarly, a complex Clifford algebra $\mathbb{C}l_n$ can be factorised in the $\phi\equiv 1$ case, giving $\tilde{\mathcal{A}}\cong\left(\mathcal{A}^\mathrm{ev}\rtimes_{(\beta,\nu)}\mathbb{R}^d\right)\hat{\otimes}\,\mathbb{C}l_n$. Using the dimension shift in $\mathbf{K}_\star$ effected by tensoring with a Clifford algebra, along with $\mathbf{K}_\star(\mathcal{A}^\mathrm{ev}_{(\mathbb{R})})\cong K_\star(\mathcal{A}^\mathrm{ev}_{(\mathbb{R})})$ and Corollary \ref{corollary:dimension shifts}, we arrive at \eqref{generaldimensionshift}.
\end{proof}

In Theorem \ref{theorem:generaldimensionshift}, $\mathcal{B}^\mathrm{ev}$ could be an algebra used to model disorder \cite{bellissard1994noncommutative}, and $(\tilde{G},c,\mathcal{B}^\mathrm{ev},\alpha,\sigma)$ can be interpreted as the full set of symmetry data for the gapped dynamics in question. Even when the additional $\mathbb{R}^d$ symmetries are realised projectively (e.g.\ as \emph{magnetic} translations in the presence of a constant magnetic field as in the Integer Quantum Hall Effect), equation \eqref{generaldimensionshift} continues to hold. Therefore, under fairly general assumptions about additional $\mathbb{R}^d$ symmetries and disorder, the difference-group for the symmetry algebra becomes that for $d=0$, provided we \emph{lower} the $K$-theory index by $d$.

Some other arguments for the dimension shift phenomenon are model-dependent and utilise suspensions \cite{kitaev2009periodic,stone2011symmetries}, which only works if the extra $\mathbb{R}^d$ symmetries are assumed to enter in a trivial way. Crucially, the ordinary suspension operation \emph{raises} rather than lowers the index in $K$-theory. This is not a problem in the complex case because Bott periodicity is 2-periodic there, but it matters greatly in the real case where there are in fact two distinct notions of suspension (see Appendix \ref{appendix:ktheory}). A restricted notion of degree shifts can be explained by using the correct type of suspension\footnote{In the topological approach, $KR$-theory does admit two types of suspensions, which are related to $S$ and $\bar{S}$ in real operator $K$-theory (see Appendix \ref{appendix:ktheory}). However, Real vector bundles in $KR$-theory are not directly models for gapped phases, and an auxiliary construction is required \cite{freed2013twisted}.}, and is a special case of our general result.

There is a weaker discretised version of the dimension shift phenomenon. To understand this, we first consider the simplest example of the graded twisted group $C^*$-algebra of $G=\mathbb{Z}^d\times A$, where $\mathbb{Z}^d$ acts trivially on $\mathbb{C}$, and $\sigma$ is only non-trivial between elements of $A$. Then 
\begin{equation}
\mathcal{A}=\mathbb{C}\rtimes_{(\alpha,\sigma)}(\mathbb{Z}^d\times A)\cong\begin{cases} (\mathbb{R}\rtimes\mathbb{Z}^d)\hat{\otimes}Cl_{r,s}\cong C(\im\mathbb{T}^d)\hat{\otimes}Cl_{r,s}\quad & \text{real case},\,\\
(\mathbb{C}\rtimes\mathbb{Z}^d)\hat{\otimes}\mathbb{C}l_n\cong C(\mathbb{T}^d,\mathbb{C})\hat{\otimes}\mathbb{C}l_n & \text{complex case},\end{cases}\nonumber
\end{equation}
where $C(\im\mathbb{T}^d)$ is defined in Appendix \ref{appendix:ktheory}. The $\mathbf{K}_0(\mathcal{A})$ groups reduce to ordinary $K$-theory groups, and when $d=1$, we can use \eqref{ktheorycrossedproductwithzreal}-\eqref{ktheorycrossedproductwithzcomplex} to obtain
\begin{equation}
	\mathbf{K}_0(\mathbb{C}\rtimes_{(\alpha,\sigma)}(\mathbb{Z}\times A))\cong\begin{cases}K_{s-r}(\mathbb{R})\oplus K_{s-r-1}(\mathbb{R})\quad &\text{real case},\\
K_n(\mathbb{C})\oplus K_{n-1}(\mathbb{C})\quad &\text{complex case}.	\end{cases}\label{discretedimensionshift1}
\end{equation}
Thus, a trivial crossed product with $\mathbb{Z}$ results in $\mathbf{K}_0(\cdot)$ acquiring an extra $K$-theory group with index \emph{lowered} by 1. As in Theorem \ref{theorem:generaldimensionshift}, we can replace $\mathbb{C}$ by $\mathcal{B}^\mathrm{ev}$, but it will then be necessary to assume that $\mathcal{B}^\mathrm{ev}$ is trivially acted upon by $\mathbb{Z}$. Assuming this and using \eqref{ktheorycrossedproductwithzreal}-\eqref{ktheorycrossedproductwithzcomplex} repeatedly, we obtain
\begin{equation}
	\mathbf{K}_0\left(\mathcal{B}^\mathrm{ev}\rtimes_{(\alpha,\sigma)}(\mathbb{Z}^d\times A)\right)\cong\begin{cases}\bigoplus\limits_{k=0}^d \binom{d}{k}K_{s-r-k}(\mathcal{B}_\mathbb{R}^\mathrm{ev})\quad &\text{real case},\\
\bigoplus\limits_{k=0}^d \binom{d}{k}K_{n-k}(\mathcal{B}^\mathrm{ev})\quad &\text{complex case}.	\end{cases}\label{differencegroupcrossedproductwithz}
\end{equation}
Note that there is always a single extra $K_{s-r-d}(\mathcal{B}_\mathbb{R}^\mathrm{ev})$ or $K_{n-d}(\mathcal{B}^\mathrm{ev})$ factor, just as in the case of a crossed product with $\mathbb{R}^d$. Equation \eqref{differencegroupcrossedproductwithz} generalises the topological $KR$-theory calculations of Kitaev in equations 25--27 of \cite{kitaev2009periodic}. In particular, if we take $A=\{1,T\},\,\mathsf{T}^2=-1,\, \mathcal{B}^\mathrm{ev}=\mathbb{C}$ and $d=3$, we obtain the difference-group for 3D $T$-invariant insulators,
\begin{eqnarray}
	\mathbf{K}_0(\mathbb{C}\rtimes_{(\alpha,\sigma)}(\mathbb{Z}^d\times \{1,T\}))&\cong & \bigoplus\limits_{k=0}^3 \binom{3}{k}K_{4-k}(\mathbb{R})\nonumber\\
	&=&K_4(\mathbb{R})\oplus\, 3K_3(\mathbb{R})\oplus\,3 K_2(\mathbb{R})\oplus\, K_1(\mathbb{R})\nonumber\\
	&=&\mathbb{Z}\oplus 4\mathbb{Z}_2,\nonumber
\end{eqnarray}
a result obtained by $KR$-theory methods\footnote{Note that we are in the Class AII situation, and $\mathbb{C}\rtimes_{(\alpha,\sigma)}(\mathbb{Z}^d\times \{1,T\})$ is purely even. Thus its difference-group coincides its ordinary $K$-theory group.} in Theorem 11.14 of \cite{freed2013twisted}. We stress that the expression \eqref{differencegroupcrossedproductwithz} assumes some specific structure of the symmetry data $(G,c,\phi,\sigma)$, in particular, the way in which $\mathbb{Z}^d$ sits inside $G$. \emph{These assumptions do \emph{not} hold when there is spatial inversion symmetry, which acts on $\mathbb{Z}^d$ non-trivially.}

\begin{table}
\begin{center}
\begin{tabular}{l | l | l  l| c c c c}
	\multirow{2}{*}{$n$}	&	\multirow{2}{*}{\parbox[c]{1.4cm}{Cartan label}} & \multirow{2}{*}{$\mathsf{C}^2$} & \multirow{2}{*}{$\mathsf{T}^2$}  &	\multicolumn{4}{c}{$\mathbf{K}_0(\mathcal{A})\cong K_{n-d}(\mathbb{R})$ or $K_{n-d}(\mathbb{C})$} \\\cline{5-8}
	& & & & $d=0$ & $d=1$& $d=2$& $d=3$ \\
	\hline
	$0$	&	$AI$	&		& $+1$	& $\mathbb{Z}$ 	& $0$			& $0$			& $0$	\\
	$1$	&	$BDI$	& $+1$ 	& $+1$	& $\mathbb{Z}_2$& $\mathbb{Z}$ 	& $0$			& $0$	\\	
	$2$	&	$D$		& $+1$ 	&		& $\mathbb{Z}_2$& $\mathbb{Z}_2$& $\mathbb{Z}$	& $0$	\\	
	$3$	&	$DIII$	& $+1$ 	& $-1$	& $0$ 			& $\mathbb{Z}_2$& $\mathbb{Z}_2$	& $\mathbb{Z}$\\	
	$4$	&	$AII$	&    	& $-1$	& $\mathbb{Z}$ 	& $0$			& $\mathbb{Z}_2$& $\mathbb{Z}_2$\\	
	$5$	&	$CII$	& $-1$ 	& $-1$	& $0$ 			& $\mathbb{Z}$  & $0$			& $\mathbb{Z}_2$\\
	$6$	&	$C$		& $-1$ 	&		& $0$			& $0$ 			& $\mathbb{Z}$ 	& $0$	\\	
	$7$	&	$CI$	& $-1$ 	& $+1$	& $0$ 			& $0$			& $0$			& $\mathbb{Z}$\\
	
	\hline\hline
	$0$	&	$A$		&\multicolumn{2}{c|}{N/A}& $\mathbb{Z}$ 	& $0$			& $\mathbb{Z}$	& $0$\\
	$1$	&	$AIII$	&\multicolumn{2}{c|}{$\mathsf{S}^2=+1$}& $0$			& $\mathbb{Z}$	& $0$			& $\mathbb{Z}$\\
\end{tabular}\label{table:periodictable}
\caption{Periodic Table of difference-groups for gapped topological phases in $d$ dimensions \emph{with only $CT$-type symmetries}, showing how they are related to the $K$-theory groups of a point (first column). The $K$-theory degree shifts in both the vertical and horizontal directions are direct manifestations of isomorphisms in $K$-theory (Theorem \ref{theorem:generaldimensionshift}. In the vertical direction, it is due to the effect on $\mathbf{K}_0(\mathcal{A})$ of tensoring with a Clifford algebra; in the horizontal direction, it is due to the Connes--Thom isomorphism and the stabilisation and decomposition theorems of Packer--Raeburn. The twofold and eightfold periodicities are due to Bott periodicity in the Clifford algebras and $K$-theory.}
\end{center}
\end{table}

We have treated all symmetries on an equal footing in this paper: they include time/charge preserving/reversing symmetries, projectively-realised symmetries, $\mathbb{Z}^d$-symmetries underlying band insulators, and $\mathbb{R}^d$ translations in extra spatial dimensions. Furthermore, Theorem \ref{theorem:generaldimensionshift} shows that the phenomenon of ``dimension shifts'' is very robust, and does not depend on the details of how the extra dimensions come into play. Table \ref{table:periodictable} summarises the periodicities in the difference-groups, \emph{in the special case where $G$ is a $CT$-subgroup}. The groups appearing there are the same as those in various Periodic Tables in the literature \cite{kitaev2009periodic,schnyder2008classification,ryu2010topological}, \emph{but their physical interpretation is very different}.

Despite the projectively-realised $\mathbb{R}^2$ symmetry in the physically important case of the Integer Quantum Hall Effect, it fits nicely into our version of the Periodic Table. We also see, in a model-independent way, why time-reversal symmetry needs to be broken (by a magnetic field or otherwise) in order to exhibit a quantized Hall conductivity. For Type $I$ insulators in two dimensions with $\mathsf{T}^2=+1$ (assuming spin-$0$), the relevant $K$-theory group is trivial; however, a $\mathbb{Z}_2$-invariant is possible if the spin-$\frac{1}{2}$ nature of electrons are taken into account, so that $\mathsf{T}^2=-1$.

%\begin{acknowledgements}
\section*{Acknowledgements}
The author wishes to thank Keith Hannabuss for his advice and the proof of Lemma \ref{lemma:pm1cocycle}, Christopher Douglas and Thomas Wasserman for helpful discussions, as well as the Clarendon Fund and Balliol College for financial support. He also acknowledges helpful discussions with Hermann Schulz--Baldes, Gregory Moore, Johannes Kellendonk, Emil Prodan, Martin Zirnbauer, Giuseppe de Nittis, Shinsei Ryu, and Terry Loring, as well as the organisers of the \emph{Topological Phases of Quantum Matter} workshop at the Erwin Schr\"{o}dinger Institute where part of this work was done.
%\end{acknowledgements}

\appendix
\section{Some remarks on $K$-theory}\label{appendix:ktheory}
A reference for the topological $K$-theory of spaces, with an exposition of Clifford algebras and the Atiyah--Bott--Shapiro isomorphisms, is \cite{lawson1989spin}. For the ordinary (ungraded) $K$-theory of $C^*$-algebras, we refer to \cite{wegge1993k} which also treats Hilbert $C^*$-modules and f.g.p.\  modules, and \cite{blackadar1998k} which discusses graded $C^*$-algebras and $KK$-theory. The $K$-theory of real $C^*$-algebras is covered in detail in \cite{schroder1993k}; many of the results from complex $K$-theory carry over to the real case, but require rather different proofs. A different approach emphasising Clifford algebras, from which we have borrowed heavily, can be found in Karoubi's textbook \cite{karoubi1978k}, especially Chapter III, as well as references therein. The following results are taken from these references.

Let $C_0(\mathbb{R},\mathbb{C})$ (resp.\ $C_0(\mathbb{R},\mathbb{R})$) be the complex (resp.\ real) non-unital $C^*$-algebra of continuous functions $\mathbb{R}\rightarrow\mathbb{C}$ (resp.\ $\mathbb{R}\rightarrow\mathbb{R}$) vanishing at infinity. The \emph{suspension} $S\mathcal{A}$ of a complex (resp.\ real) $C^*$-algebra $\mathcal{A}$ is defined to be $S\mathcal{A}\coloneqq\mathcal{A}\otimes C_0(\mathbb{R},\mathbb{C})$ (resp.\ $S\mathcal{A}\coloneqq\mathcal{A}\otimes C_0(\mathbb{R},\mathbb{R})$). The higher $K$-theory groups are defined to be $K_n(\mathcal{A})\equiv K_0(S^n\mathcal{A})$. In the complex case, we have $\mathcal{A}\rtimes\mathbb{R}\cong \mathcal{A}\otimes C_0(\hat{\mathbb{R}},\mathbb{C})\cong S\mathcal{A}$, where $\hat{\mathbb{R}}$ is the Fourier transform (dual) of $\mathbb{R}$. In the real case, we have $\mathcal{A}\rtimes\mathbb{R}\cong \mathcal{A}\otimes C_0(\im \hat{\mathbb{R}})$ instead (see Equation \eqref{correctrealfourierisomorphism}), which suggests the definition $\bar{S}\mathcal{A}\coloneqq \mathcal{A}\otimes C_0(\im\hat{\mathbb{R}})$. It turns out that $\bar{S}$ is the $K$-theoretic ``inverse'' operation to $S$,
\begin{equation}
	K_n(S\bar{S}\mathcal{A})\cong K_n(\bar{S}S\mathcal{A})= K_n(\mathcal{A}\otimes C_0(\im\hat{\mathbb{R}})\otimes C_0(\mathbb{R},\mathbb{R}))\cong K_n(\mathcal{A}),\nonumber
\end{equation}
and so $K_n(\bar{S}\mathcal{A})\cong K_{n-1}(\mathcal{A})$, which is a special case of Theorem \ref{theorem:connesthom}. Bott periodicity in $K$-theory says that
\begin{alignat}{3}
	K_n(\mathcal{A})&\cong & K_{n+8}(\mathcal{A})\qquad &\mathrm{real\,\,case},\nonumber\\
	K_n(\mathcal{A})&\cong & K_{n+2}(\mathcal{A})\qquad &\mathrm{complex\,\,case},\nonumber
\end{alignat}
and leads to cyclic long exact sequences in $K$-theory, with six terms in the complex case and 24 terms in the real case.

Let $(\mathbb{T}^d,\varsigma)$ be the Real space $\mathbb{T}^d$ with involution $\mathbf{z}\mapsto\varsigma(\mathbf{z})\coloneqq\bar{\mathbf{z}}$, and let $C(\im\mathbb{T}^d)$ be the real $C^*$-algebra of continuous functions $f:\mathbb{T}^d\rightarrow\mathbb{C}$ such that $f(\varsigma(\mathbf{z}))=\overline{f(\mathbf{z})}$. The group $C^*$-algebras of $\mathbb{Z}^d$ are $\mathbb{R}\rtimes \mathbb{Z}^d \cong C(\im\mathbb{T}^d)$ and $\mathbb{C}\rtimes \mathbb{Z}^d \cong C(\mathbb{T}^d,\mathbb{C})$. There are isomorphisms
\begin{subequations}
\begin{alignat}{1}
	K_n(\mathcal{A}\otimes C(\im\mathbb{T}^1))\cong K_n(\mathcal{A})\oplus K_{n-1}(\mathcal{A})\qquad & \mathrm{real\,\,case},\label{ktheorycrossedproductwithzreal}\\
	K_n(\mathcal{A}\otimes C(\mathbb{T}^1,\mathbb{C}))\cong K_n(\mathcal{A})\oplus K_{n-1}(\mathcal{A})\qquad & \mathrm{complex\,\,case}\label{ktheorycrossedproductwithzcomplex},
\end{alignat}
\end{subequations}
which can be shown using the Pimsner--Voiculescu (PV) exact sequence (Theorem 10.2.1 of \cite{blackadar1998k} and Theorem 1.5.5 of \cite{schroder1993k}) for the $K$-theory of crossed products of $\mathcal{A}$ by $\mathbb{Z}$.

One way in which Clifford algebras are related to higher $K$-theory groups is through the approach of Karoubi. We have followed his definitions closely in defining $\mathbf{K}_n(\cdot)$ in Section \ref{section:karoubidifferenceconstruction}. For commutative $\mathcal{A}\cong C(X)$, Karoubi developed another model for $K$-theory \cite{karoubi1970algebres}, defining his $\bar{K}^{r,s}(X)$ groups in terms of homotopy classes of graded bundles, over $X$, of odd self-adjoint Fredholm operators which are $Cl_{r,s}$-antilinear (i.e.\ anticommutes with the Clifford generators). He showed that his $\bar{K}^{r,s}(X)$ are isomorphic to his $K'^{r,s}(C(X))$ and hence to our $\mathbf{K}_{s,r}(C(X))$. This should be compared to the use of $Cl_{n-1,0}$-antilinear skew-adjoint Fredholm operators (or odd $Cl_{n,0}$-linear self-adjoint Fredholm operators) as classifying spaces for $K^{-n}(\cdot)$ in \cite{atiyah1969index}. Some of the connections between the ABS constructions in \cite{atiyah1966k}, Karoubi's $K$-theory, as well as the twisted $K$-theory of \cite{donovan1970graded}, are decribed in \cite{karoubi2008clifford}. Finally, there is the most general $KK$-theory description\footnote{\emph{Right} Hilbert $C^*$-modules are more standard in $KK$-theory, whereas we have used \emph{left} modules throughout this paper.} (Definition 2.4.11 of \cite{schroder1993k},\cite{kasparov1995k})
\begin{alignat}{2}
	K_n(\mathcal{A})\cong KKO_{n,0}(\mathbb{R},\mathcal{A})&\coloneqq KKO(\mathbb{R},Cl_{n,0}\hat{\otimes}\mathcal{A})\qquad & \mathrm{real\,\,case},\nonumber\\
	K_n(\mathcal{A})\cong KK(\mathbb{C},\mathcal{A})&\coloneqq KK_n(\mathbb{C},\mathbb{C}l_n\hat{\otimes}\mathcal{A})\qquad & \mathrm{complex\,\,case}.\nonumber
\end{alignat}

%
% For one-column wide figures use
%\begin{figure}
% Use the relevant command for your figure-insertion program
% to insert the figure file.
% For example, with the option graphics use
%\centering
%\resizebox{0.75\textwidth}{!}{%
%  \includegraphics{leer}
%}
% If not, use
%\vspace{5cm}       % Give the correct figure height in cm
%\caption{Please write your figure caption here}
%\label{fig:1}       % Give a unique label
%\end{figure}
%

% For tables use
%\begin{table}
%\caption{Please write your table caption here}
%\label{tab:1}       % Give a unique label
% For LaTeX tables use
%\centering
%\begin{tabular}{lll}
%\hline\noalign{\smallskip}
%first & second & third  \\
%\noalign{\smallskip}\hline\noalign{\smallskip}
%number & number & number \\
%number & number & number \\
%\noalign{\smallskip}\hline
%\end{tabular}
% Or use
%\vspace*{5cm}  % with the correct table height
%\end{table}

%
% BibTeX users please use
\bibliographystyle{plain}
\bibliography{topologicalmattercmprevised}

\begin{thebibliography}{10}

\bibitem{abramovici2012clifford}
G.~Abramovici and P.~Kalugin.
\newblock Clifford modules and symmetries of topological insulators.
\newblock {\em Int. J. Geom. Methods Mod. Phys.}, 09(03):1250023, 2012.

\bibitem{altland1997nonstandard}
A.~Altland and M.R. Zirnbauer.
\newblock Nonstandard symmetry classes in mesoscopic normal-superconducting
  hybrid structures.
\newblock {\em Phys. Rev. B}, 55(2):1142, 1997.

\bibitem{atiyah1966k}
M.F. Atiyah.
\newblock ${K}$-theory and reality.
\newblock {\em Q. J. Math.}, pages 367--386, 1966.

\bibitem{atiyah1964clifford}
M.F. Atiyah, R.~Bott, and A.~Shapiro.
\newblock Clifford modules.
\newblock {\em Topology}, 3:3--38, 1964.

\bibitem{atiyah1969index}
M.F. Atiyah and I.M. Singer.
\newblock Index theory for skew-adjoint {F}redholm operators.
\newblock {\em Inst. Hautes {\'E}tudes Sci. Publ. Math.}, 37(1):5--26, 1969.

\bibitem{avron1990quantum}
J.E. Avron, R.~Seiler, and B.~Simon.
\newblock Quantum {H}all effect and the relative index for projections.
\newblock {\em Phys. Rev. Lett.}, 65(17):2185, 1990.

\bibitem{avron1994charge}
J.E. Avron, R.~Seiler, and B.~Simon.
\newblock Charge deficiency, charge transport and comparison of dimensions.
\newblock {\em Comm. Math. Phys.}, 159(2):399--422, 1994.

\bibitem{bellissard1994noncommutative}
J.~Bellissard, A.~van Elst, and H.~Schulz-Baldes.
\newblock The noncommutative geometry of the quantum hall effect.
\newblock {\em J. Math. Phys.}, 35(10):5373--5451, 1994.

\bibitem{blackadar1998k}
B.~Blackadar.
\newblock {\em ${K}$-theory for operator algebras}, volume~5 of {\em Math. Sci.
  Res. Inst. Publ.}
\newblock Cambridge Univ. Press, Cambridge, 2nd edition, 1998.

\bibitem{busby1970representations}
R.C. Busby and H.A. Smith.
\newblock Representations of twisted group algebras.
\newblock {\em Trans. Amer. Math. Soc.}, 149(2):503--537, 1970.

\bibitem{chen2010local}
X.~Chen, Z.-C. Gu, and X.-G. Wen.
\newblock Local unitary transformation, long-range quantum entanglement, wave
  function renormalization, and topological order.
\newblock {\em Phys. Rev. B}, 82(15):155138, 2010.

\bibitem{chiu2013classification}
C.-K. Chiu, H.~Yao, and S.~Ryu.
\newblock Classification of topological insulators and superconductors in the
  presence of reflection symmetry.
\newblock {\em Phys. Rev. B}, 88(7):075142, 2013.

\bibitem{connes1981analogue}
A.~Connes.
\newblock An analogue of the {T}hom isomorphism for crossed products of a
  ${C}^\star$-algebra by an action of $\mathbb{R}$.
\newblock {\em Adv. Math.}, 39(1):31--55, 1981.

\bibitem{de2014classification2}
G.~De~Nittis and K.~Gomi.
\newblock Classification of ``{Q}uaternionic'' {B}loch-bundles: {T}opological
  insulators of type ${AII}$.
\newblock {\em arXiv preprint arXiv:1404.5804}, 2014.

\bibitem{derezinski2010positive}
J.~Derezi{\'n}ski and C.~G{\'e}rard.
\newblock Positive energy quantization of linear dynamics.
\newblock {\em Banach Center Publications}, 89:75--104, 2010.

\bibitem{derezinski2013mathematics}
J.~Derezi{\'n}ski and C.~G{\'e}rard.
\newblock {\em Mathematics of Quantization and Quantum Fields}.
\newblock Cambridge Monogr. Math. Phys. Cambridge Univ. Press, Cambridge, 2013.

\bibitem{donovan1970graded}
P.~Donovan and M.~Karoubi.
\newblock Graded {B}rauer groups and ${K}$-theory with local coefficients.
\newblock {\em Inst. Hautes {\'E}tudes Sci. Publ. Math.}, 38(1):5--25, 1970.

\bibitem{dupont1969symplectic}
J.L. Dupont.
\newblock Symplectic bundles and ${KR}$-theory.
\newblock {\em Math. Scand.}, 24:27--30, 1969.

\bibitem{freed2013twisted}
D.S. Freed and G.W. Moore.
\newblock Twisted equivariant matter.
\newblock {\em Ann. Henri Poincar{\'e}}, 14(8):1927--2023, 2013.

\bibitem{gracia2001elements}
J.M. Gracia-Bond{\'\i}a, J.C. V{\'a}rilly, and H.~Figueroa.
\newblock {\em Elements of noncommutative geometry}.
\newblock Birkh{\"a}user Advanced Texts Basler Lehrb{\"u}cher. Birkh{\"a}user,
  Boston, 2001.

\bibitem{gruber2001noncommutative}
M.J. Gruber.
\newblock Noncommutative {B}loch theory.
\newblock {\em J. Math. Phys.}, 42(6):2438--2465, 2001.

\bibitem{hasan2010colloquium}
M.Z. Hasan and C.L. Kane.
\newblock Colloquium: topological insulators.
\newblock {\em Rev. Mod. Phys.}, 82(4):3045, 2010.

\bibitem{hazrat2014graded}
R.~Hazrat.
\newblock Graded rings and graded {G}rothendieck groups.
\newblock {\em arXiv preprint arXiv:1405.5071}, 2014.

\bibitem{heinzner2005symmetry}
P.~Heinzner, A.~Huckleberry, and M.R. Zirnbauer.
\newblock Symmetry classes of disordered fermions.
\newblock {\em Comm. Math. Phys.}, 257(3):725--771, 2005.

\bibitem{horava2005stability}
Petr Ho{\v{r}}ava.
\newblock Stability of {F}ermi surfaces and ${K}$ theory.
\newblock {\em Phys. Rev. Lett.}, 95(1):016405, 2005.

\bibitem{karoubi1968algebres}
M.~Karoubi.
\newblock Alg{\`e}bres de {C}lifford et ${K}$-th{\'e}orie.
\newblock {\em Ann. Sci. {\'E}c. Norm. Sup{\'e}r.}, 1(2):161--270, 1968.

\bibitem{karoubi1970algebres}
M.~Karoubi.
\newblock Alg{\`e}bres de {C}lifford et op{\'e}rateurs de {F}redholm.
\newblock In {\em S{\'e}minaire Heidelberg-Saarbr{\"u}cken-Strasbourg sur la
  $K$-Th{\'e}orie}, volume 136 of {\em Lecture Notes in Math.}, pages 66--106.
  Springer-Verlag, Berlin, 1970.

\bibitem{karoubi1978k}
M.~Karoubi.
\newblock {\em ${K}$-theory: {A}n Introduction}, volume 226 of {\em Grundlehren
  Math. Wiss.}
\newblock Springer-Verlag, Berlin, 1978.

\bibitem{karoubi2008clifford}
M.~Karoubi.
\newblock Clifford modules and twisted ${K}$-theory.
\newblock {\em Adv. Appl. Clifford Algebr.}, 18(3-4):765--769, 2008.

\bibitem{karoubi2008twisted}
M.~Karoubi.
\newblock Twisted ${K}$-theory, old and new.
\newblock In {\em $K$-theory and {N}oncommutative {G}eometry}, EMS Ser. Congr.
  Rep., pages 117--149. EMS, Z{\"u}rich, 2008.

\bibitem{kasparov1995k}
G.G. Kasparov.
\newblock ${K}$-theory, group ${C}^\star$-algebras, and higher signatures
  ({c}onspectus).
\newblock In {\em {N}ovikov conjectures, index theorems and rigidity}, volume
  226 of {\em {L}ondon {M}ath. {S}oc. Lecture Notes Ser.}, pages 101--146.
  Cambridge Univ. Press, 1995.

\bibitem{kellendonk2004boundary}
J.~Kellendonk and H.~Schulz-Baldes.
\newblock Boundary maps for ${C^\star}$-crossed products with with an
  application to the quantum {H}all effect.
\newblock {\em Comm. Math. Phys.}, 249(3):611--637, 2004.

\bibitem{kellendonk2004quantization}
J.~Kellendonk and H.~Schulz-Baldes.
\newblock Quantization of edge currents for continuous magnetic operators.
\newblock {\em J. Funct. Anal.}, 209(2):388--413, 2004.

\bibitem{kitaev2009periodic}
A.~Kitaev.
\newblock Periodic table for topological insulators and superconductors.
\newblock {\em arXiv preprint arXiv:0901.2686}, 2009.

\bibitem{lance1995hilbert}
E.C. Lance.
\newblock {\em Hilbert ${C}^\star$-modules: {A} toolkit for operator
  algebraists}, volume 210 of {\em London Math. Soc. Lecture Notes Ser.}
\newblock Cambridge Univ. Press, Cambridge, 1995.

\bibitem{landweber2005representation}
G.D. Landweber.
\newblock Representation rings of {L}ie superalgebras.
\newblock {\em ${K}$-Theory}, 36(1):115--168, 2005.

\bibitem{landweber2006twisted}
G.D. Landweber.
\newblock Twisted representation rings and {D}irac induction.
\newblock {\em J. Pure Appl. Algebra}, 206(1):21--54, 2006.

\bibitem{lawson1989spin}
H.B. Lawson and M.-L. Michelsohn.
\newblock {\em Spin geometry}, volume~38 of {\em Princeton Math. Ser.}
\newblock Princeton Univ. Press, Princeton, 1989.

\bibitem{leptin1965verallgemeinerte}
H.~Leptin.
\newblock Verallgemeinerte ${L}^1$-{A}lgebren.
\newblock {\em Math. Ann.}, 159(1):51--76, 1965.

\bibitem{packer1989twisted}
J.A. Packer and I.~Raeburn.
\newblock Twisted crossed products of ${C}^\star$-algebras.
\newblock {\em Math. Proc. Cambridge Philos. Soc.}, 106(02):293--311, 1989.

\bibitem{packer1990twisted}
J.A. Packer and I.~Raeburn.
\newblock Twisted crossed products of ${C}^\star$-algebras. {II}.
\newblock {\em Math. Ann.}, 287(1):595--612, 1990.

\bibitem{parthasarathy1969projective}
K.R. Parthasarathy.
\newblock Projective unitary antiunitary representations of locally compact
  groups.
\newblock {\em Comm. Math. Phys.}, 15(4):305--328, 1969.

\bibitem{pedersen1979c}
G.K. Pedersen.
\newblock {\em ${C}^\star$-algebras and their automorphism groups}, volume~14
  of {\em London Math. Soc. Monogr. Ser.}
\newblock Academic Press, London, 1979.

\bibitem{prodan2014non}
E.~Prodan.
\newblock The non-commutative geometry of the complex classes of topological
  insulators.
\newblock {\em Top. Quant. Matter}, 1(1):1--16, 2014.

\bibitem{qi2011topological}
X.-L. Qi and S.-C. Zhang.
\newblock Topological insulators and superconductors.
\newblock {\em Rev. Mod. Phys.}, 83(4):1057, 2011.

\bibitem{ryu2010topological}
S.~Ryu, A.P. Schnyder, A.~Furusaki, and A.W.W. Ludwig.
\newblock Topological insulators and superconductors: tenfold way and
  dimensional hierarchy.
\newblock {\em New J. Phys.}, 12(6):065010, 2010.

\bibitem{schnyder2008classification}
A.P. Schnyder, S.~Ryu, A.~Furusaki, and A.W.W. Ludwig.
\newblock Classification of topological insulators and superconductors in three
  spatial dimensions.
\newblock {\em Phys. Rev. B}, 78(19):195125, 2008.

\bibitem{schroder1993k}
H.~Schr{\"o}der.
\newblock {\em ${K}$-theory for real ${C}^\star$-algebras and applications},
  volume 290 of {\em Pitman Res. Notes Math. Ser.}
\newblock Longman, Harlow, 1993.

\bibitem{spanier1966algebraic}
E.H. Spanier.
\newblock {\em Algebraic topology}.
\newblock McGraw-Hill, New York, 1966.

\bibitem{stone2011symmetries}
M.~Stone, C.-K. Chiu, and A.~Roy.
\newblock Symmetries, dimensions and topological insulators: {T}he mechanism
  behind the face of the {B}ott clock.
\newblock {\em J. Phys. A}, 44(4):045001, 2011.

\bibitem{thiang2014thesis}
G.C. Thiang.
\newblock Topological phases of matter, symmetries, and ${K}$-theory.
\newblock DPhil.\ thesis, University of Oxford, 2014 (forthcoming).

\bibitem{wegge1993k}
N.E. Wegge-Olsen.
\newblock {\em ${K}$-theory and ${C}^\star$-algebras: {A} Friendly Approach}.
\newblock Oxford Univ. Press, New York, 1993.

\bibitem{williams2007crossed}
D.P. Williams.
\newblock {\em Crossed products of ${C}^\star$-algebras}, volume 134 of {\em
  Math. Surveys Monogr.}
\newblock Amer. Math. Soc., Providence, 2007.

\end{thebibliography}
%
% Non-BibTeX users please use
%\begin{thebibliography}{}
%
% and use \bibitem to create references.
%
%\bibitem{RefJ}
% Format for Journal Reference
%Author, Journal \textbf{Volume}, (year) page numbers.
% Format for books
%\bibitem{RefB}
%Author, \textit{Book title} (Publisher, place year) page numbers
% etc
%\end{thebibliography}

\end{document}